\documentclass{amsart}

\usepackage{amsfonts}
\usepackage{mathrsfs}
\usepackage{amsmath}
\usepackage{amsthm}
\usepackage{amssymb,bm}
\usepackage{latexsym,hyperref}
\usepackage{enumitem,nccmath}
\usepackage{etoolbox}
\usepackage{xcolor}

\usepackage{hyperref}
\hypersetup{%
    pdfborder = {0 0 0},
    colorlinks,
    citecolor=red!45!black,
    filecolor=Darkgreen,
    linkcolor=blue!70!black!90,
    urlcolor=cyan!50!black!90
}

\hypersetup{hypertexnames=false}

\newtheorem{thrm}{Theorem}[section]
\newtheorem{prop}[thrm]{Proposition}
\newtheorem{crl}[thrm]{Corollary}
\newtheorem{lemma}[thrm]{Lemma}

\newtheorem{defn}[thrm]{Definition}
\newtheorem{exam}[thrm]{Example}

\theoremstyle{remark}
\newtheorem{rmk}[thrm]{Remark}

\newcommand{\nc}{\newcommand}
\newcommand{\rnc}{\renewcommand}

\DeclareMathOperator{\tr}{tr}

\nc{\Res}[1]{\underset{\;{#1}\;}{\rm Res}}

\nc{\End}{\mathrm{End}}
\nc{\Sym}{\mathrm{Sym}}

\nc\bb{\mathbb}
\nc\mf{\mathfrak}
\nc\ms{\mathsf}
\nc\mc{\mathcal}
\nc\mr{\mathscr}
\nc\mt[1]{{\tt #1}}

\nc{\mfg}{\mf{g}}
\nc{\mfh}{\mf{h}}
\nc{\mfsl}{\mf{sl}}
\nc{\mfgl}{\mf{gl}}
\nc{\mfso}{\mf{so}}
\nc{\mfsp}{\mf{sp}}

\nc{\sca}{\mathscr{a}}
\nc{\scb}{\mathscr{b}}
\nc{\scc}{\mathscr{c}}
\nc{\scd}{\mathscr{d}}

\nc{\nn}{\nonumber}   
\nc{\el}{\nonumber\\} 

\nc{\equ}[1]{\begin{equation}#1\end{equation}}
\nc{\eqa}[1]{\begin{equation}\begin{alignedat}{50}#1\end{alignedat}\end{equation}}
\nc{\eqn}[1]{\begin{equation*}\begin{alignedat}{50}#1\end{alignedat}\end{equation*}}
\nc{\eqg}[1]{\begin{equation}\begin{gathered}#1\end{gathered}\end{equation}}

\nc{\ali}[1]{\begin{alignat}{50}#1\end{alignat}}
\nc{\als}[1]{\begin{subequations}\begin{alignat}{50}#1\end{alignat}\end{subequations}}
\nc{\aln}[1]{\begin{alignat*}{50}#1\end{alignat*}}

\nc{\gat}[1]{\begin{gather}#1\end{gather}}
\nc{\gas}[1]{\begin{subequations}\begin{gather}#1\end{gather}\end{subequations}}
\nc{\gan}[1]{\begin{gather*}#1\end{gather*}}

\nc{\mcA}{\mc{A}}
\nc{\mcB}{\mc{B}}
\nc{\mcU}{\mc{U}}
\nc{\mfU}{\mf{U}}
\nc{\mcP}{\mc{P}}
\nc{\mcQ}{\mc{Q}}
\nc{\mcX}{\mc{X}}
\nc{\mcZ}{\mc{Z}}
\nc{\mcT}{\mc{T}}
\nc{\mcG}{\mc{G}}
\nc{\msS}{\ms{S}}
\nc{\mss}{\ms{s}}
\nc{\msc}{\ms{c}}
\nc{\msd}{\ms{d}}
\nc{\msv}{\ms{v}}
\nc{\msq}{\ms{q}}
\nc{\msw}{\ms{w}}
\nc{\mcS}{\mc{S}}
\nc{\mcI}{\mc{I}}

\nc{\ol}{\overline}

\nc{\C}{\mathbb{C}}
\nc{\N}{\mathbb{N}}
\nc{\Z}{\mathbb{Z}}

\nc{\RR}{I\hspace{-1.3mm}R}

\nc{\ot}{\otimes}
\nc{\op}{\oplus}

\nc{\lan}{\langle}
\nc{\ran}{\rangle}

\nc{\qu}{\quad}
\nc{\qq}{\qquad}

\nc\Tr{{\rm tr}}

\nc{\al}{\alpha}
\nc{\del}{\delta}
\nc{\eps}{\epsilon}
\nc{\veps}{\varepsilon}
\nc{\ga}{\gamma}
\nc{\Ga}{\Gamma}
\nc{\ka}{\kappa}
\nc{\la}{\lambda}
\nc{\om}{\omega}
\nc{\Om}{\Omega}
\nc{\si}{\sigma}
\nc{\Si}{\Sigma}
\nc{\bsi}{\boldsymbol\sigma}
\nc{\bSi}{\boldsymbol\Sigma}
\nc{\Ups}{\upsilon}
\nc{\vphi}{\varphi}

\nc{\btau}{\boldsymbol\tau}
\nc{\bdel}{\boldsymbol\delta}

\nc{\id}{\mathrm{id}}
\nc{\gr}{\mathrm{gr}}

\nc{\lrh}{\leftrightharpoons}
\nc{\iso}{\stackrel{\sim}{\longrightarrow}}
\nc{\liso}{\stackrel{\sim}{\longleftarrow}}
\nc{\wh}{\widehat}
\nc{\wt}{\widetilde}
\nc{\tl}{\tilde}
\nc{\lra}{\longrightarrow}
\nc{\ra}{\rightarrow}
\nc{\into}{\hookrightarrow}
\nc{\onto}{\twoheadrightarrow}

\usepackage[scr=boondoxo]{mathalfa}
\nc{\key}{{\mathscr{k}}}
\nc{\ley}{{\mathscr{l}}}

\numberwithin{equation}{section}

\nc{\even}[2]{\left\{ #2 \right\}^{#1}}

\nc{\bethe}[2]{\beta_{\at_1 a_1 \dots \at_{#1} a_{#1}}(#2)}

\renewcommand{\,}{\kern 0.1em} 

\begin{document}

\title[Bethe ansatz for twisted Yangian]{Nested algebraic Bethe ansatz for open spin chains \\ with even twisted Yangian symmetry}

\author{Allan Gerrard}
\email{ajg569@york.ac.uk}
\address{University of York, Department of Mathematics, York, YO10 5DD, UK.}

\author{Niall MacKay}
\email{niall.mackay@york.ac.uk}
\address{University of York, Department of Mathematics, York, YO10 5DD, UK.}

\author{Vidas Regelskis}
\email{vidas.regelskis@gmail.com}
\address{University of York, Department of Mathematics, York, YO10 5DD, UK, and Institute of Theoretical Physics and Astronomy, Vilnius University, Saul\.etekio av.~3, Vilnius 10257, Lithuania}



\begin{abstract}  
We present a nested algebraic Bethe ansatz for a one dimensional open spin chain whose boundary quantum spaces are irreducible $\mfso_{2n}$- or $\mfsp_{2n}$-representations and the monodromy matrix satisfies the defining relations of the Olshanskii twisted Yangian $Y^\pm(\mfgl_{2n})$.
We use a generalization of the Bethe ansatz introduced by De Vega and Karowski which allows us to relate the spectral problem of a $\mfso_{2n}$- or $\mfsp_{2n}$-symmetric open spin chain to that of a $\mfgl_{n}$-symmetric periodic spin chain. We explicitly derive the structure of the Bethe vectors and the nested Bethe equations. 
\end{abstract}

\maketitle

%

\section{Introduction}

The algebraic Bethe ansatz (ABA) introduced in \cite{STF} has proven to be a powerful method to study quantum integrable models. It provides an effective approach for determining the spectrum of quantum Hamiltonians by reducing the problem of diagonalizing the Hamiltonian to a set of algebraic equations, known as the Bethe ansatz equations (BAE), that in many cases can be solved using numerical methods \cite{BeRa1,BeRa2,FKPR,KuRs,PRS1,Sk}. The key idea of ABA is to construct the so-called Bethe vectors that depend on sets of complex parameters. In the case when these parameters satisfy BAE, the corresponding Bethe vectors become eigenvectors of the quantum Hamiltonian \cite{PRS1,PRS2}. Finding eigenvectors and their eigenvalues provides the necessary first step that needs to be taken in the study of scalar products and norms \cite{HLPRS1,HLPRS2,Ko}, correlation functions and form factors \cite{IzKo,KKMST1,KKMST2,KMST,KMT,Sl1}; see also the reviews \cite{PRS3,Sl2} and references therein. 

In this paper we construct a nested algebraic Bethe ansatz for a one-dimensional quantum spin chain with open boundaries, whose underlying symmetry is the even Olshanskii twisted Yangian $Y^\pm(\mfgl_{2n})$ \cite{Ol,MNO}. Such integrable models have ``soliton non-preserving'' boundary conditions and can be described by a Hamiltonian of an alternating type; they have been intensively studied using techniques of the analytic Bethe ansatz in \cite{AACDFR,ACDFR1,ACDFR2,ADK,Do}. The full quantum space of the model is a tensor product of a finite number of arbitrary irreducible finite-dimensional ``bulk'' $\mfgl_{2n}$-representations and an arbitrary irreducible finite-dimensional ``boundary'' $\mfso_{2n}$- or $\mfsp_{2n}$-representation. This space is then equipped with the structure of a lowest weight $Y^\pm(\mfgl_{2n})$-module, so that the double-row monodromy matrix satisfies the defining relations of the twisted Yangian $Y^\pm(\mfgl_{2n})$. We will call this model a $Y^\pm(\mfgl_{2n})$-system.

Studying the spectral problem of spin chains with orthogonal and symplectic symmetries requires elaborate algebraic methods: the usual nesting approach for a $\mfgl_n$-symmetric spin chain, put forward in \cite{KuRs}, fails since there are no natural Yangian analogues of the chains of subalgebras $\mfso_{2n} \supset \mfso_{2n-2} \supset \ldots \supset \mfso_2$ and $\mfsp_{2n} \supset \mfsp_{2n-2} \supset \ldots \supset \mfsp_2$. (This problem for the Yangians $Y(\mfso_{N})$ and $Y(\mfsp_{N})$ was addressed in \cite{JLM}.) It was shown in \cite{DVK,Rs} that the spectral problem of such a system can be addressed using the algebraic Bethe ansatz if the $R$-matrix intertwining the monodromy matrices of the model can be written in a six-vertex block-form. 
This approach has recently been used in \cite{GoPa} to study orthogonal quantum spin chains with open boundary conditions, whose monodromy matrix satisfies the defining relations of the twisted Yangian of type~$\mfso_{2n}$ \cite{GuRe}. Prior, in \cite{LSY} a different method was used to study the spectral problem of a symplectic quantum spin chain with open boundary conditions. 
Their approach follows ideas from the ABA for the Izergin-Korepin model, introduced by Tarasov~\cite{Ta} for periodic boundary conditions and further studied by Fan~\cite{Fa} for open boundary conditions. 

Our strategy for solving the spectral problem of the $Y^\pm(\mfgl_{2n})$-system is as follows. We interpret the generating matrix of $Y^\pm(\mfgl_{2n})$ as the monodromy matrix of the model. Inspired by the ideas put forward in \cite{DVK,Rs}, we write the defining relations of $Y^\pm(\mfgl_{2n})$ in a block-form, i.e., in terms of the matrix operators $A$, $B$, $C$ and $D$, that are matrix analogous of the conventional creation, annihilation and diagonal operators of the six-vertex model. Indeed, the exchange relations between these matrix operators turn out to be reminiscent of those of the six-vertex model. 
We then introduce creation operators that are constructed using matrix entries of the $B$ operator. 
We define the vacuum sector of the quantum space as the subspace annihilated by the $C$ operator. 
The Bethe vector is constructed by acting with creation operators on vectors in the vacuum sector. Then, following Sklyanin's pioneering work \cite{Sk}, we require the Bethe vector to be an eigenvector of the double-row transfer matrix. Surprisingly, this procedure reduces the initial $Y^\pm(\mfgl_{2n})$-system to a residual $Y(\mfgl_n)$-system, the nested algebraic Bethe ansatz for which is well-known \cite{KuRs}.
Our main results are construction of the Bethe vectors and derivation of their eigenvalues and nested Bethe equations, see Theorems \ref{T:eig} and \ref{T:BE}, and Proposition \ref{P:tf} stating a trace formula for the Bethe vectors. The resulting Bethe equations are comparable to those obtained in \cite{ACDFR1} for an open spin chain with a one-dimensional boundary representation see~Remark \ref{R:NBE}.

The plan of the paper is as follows. In Section 2 we provide the necessary preliminaries and definitions. We recall the definition of the Yangian $Y(\mfgl_{2n})$, the twisted Yangian $Y^\pm(\mfgl_{2n})$ and relevant details of their representation theory. We then obtain a six-vertex block-form of the Yang $R$-matrix and its twisted counterpart, which allows us to write the defining relations of both Yangian and twisted Yangian in terms of the matrix operators $A$, $B$, $C$ and $D$. 
In Section~3 we provide the technical details of the main ingredients necessary for the nested Bethe ansatz. We introduce the creation operator of multi-excitations and describe its algebraic properties. We derive the exchange relations for the operators that lead to the so-called wanted and unwanted terms. 
In Section~4 we present the nested Bethe ansatz, in two steps. First, we demonstrate the method for a single top-level excitation. Then we generalize the method to multi-excitations at the top-level and provide the complete set of Bethe equations and a trace formula for the Bethe vectors.
In Appendix~A we provide in detail the nested algebraic Bethe ansatz for $Y(\mfgl_n)$, first presented in \cite{KuRs}, to which the Bethe ansatz for $Y^\pm(\mfgl_{2n})$ reduces.


\section{Definitions and preliminaries}


\subsection{Notation} \label{sec:notation}

Choose $N \in\N$. Let $\mfgl_N$ denote the general linear Lie algebra and let $E_{ij}$ with $1\le i,j \le N$ be the standard basis elements of $\mfgl_N$ satisfying
\[
[E_{ij}, E_{kl}] = \del_{jk} E_{il} - \del_{il} E_{kj}.
\]
The orthogonal Lie algebra $\mfso_N$ or the symplectic Lie algebra $\mfsp_N$ can be regarded as a subalgebra of $\mfgl_N$ as follows. For any $1 \le i,j\le N$ set $\theta_{ij}=\theta_i\theta_j$ with $\theta_i=1$ in the orthogonal case and $\theta_i=\del_{i>N/2}-\del_{i\le N/2}$ in the symplectic case.
Introduce elements $F_{ij}= E_{ij} - \theta_{ij} E_{\bar\jmath\,\bar\imath}$ with $\bar\imath = N-i+1$ and $\bar\jmath = N-j+1$. These elements satisfy the relations 
\gat{ 
\label{[F,F]}
[F_{ij},F_{kl}] = \del_{jk} F_{il} - \del_{il} F_{kj} + \theta_{ij} (\del_{j \bar l}F_{k \bar\imath} - \del_{i \bar k} F_{\bar\jmath\, l}) ,
\\
\label{F+F=0}
F_{ij} + \theta_{ij} F_{\bar\jmath\,\bar\imath}=0,
}
which in fact are the defining relations of the Lie algebra $\mfso_N$ or $\mfsp_N$. Namely, we may identify $\mfso_N$ or $\mfsp_N$ with $\mathrm{span}_{\C} \{ F_{ij} : 1\le i,j\le N \}$ and we will use $\mfh_N=\mathrm{span}_{\C} \{ F_{ii} : 1 \le i \le \lfloor N/2 \rfloor \}$ as a Cartan subalgebra. In this work we will focus on the chains of Lie algebras $\mfgl_{2n} \supset \mfso_{2n} \supset \mfgl_{n}$ and $\mfgl_{2n} \supset \mfso_{2n} \supset \mfgl_{n}$, hence we will assume that $N=2n$ or $N=n$. Given a Lie algebra $\mfg$ its universal enveloping algebra will be denoted by $U(\mfg)$.

Next, we need to introduce some operators acting on $\C^N\ot \C^N$, where the tensor product $\ot$ is defined over the field of complex numbers, that is $\ot = \ot_\C$. Let $e_{ij}\in\End(\C^N)$ be the standard matrix units with entries in $\C$, and let $e_i$ be the standard basis vectors of $\C^N$. Then $P$ will denote the permutation operator on $\C^N \ot \C^N$ and we set $Q=P^{t_1}=P^{t_2}$, where the transpose $t$ is defined by $(e_{ij})^t = \theta_{ij} e_{\bar\jmath\,\bar\imath}$; explicitly,
\[
P = \sum_{1\le i,j \le N} e_{ij} \ot e_{ji}, \qq 
Q = \sum_{1\le i,j \le N} \theta_{ij}\, e_{ij} \ot e_{\bar\imath\,\bar\jmath} .
\]

Let $I$ denote the identity matrix on $\C^N\ot \C^N$ or $\C^N$ (it will always be clear from the context which $I$ is used). Then $P^2=I$, $PQ=QP=\pm Q$, $Q^2=N Q$, which will be useful below. Here (and henceforth in this paper) the upper sign in $\pm$ and $\mp$ corresponds to the orthogonal case and the lower sign to the symplectic case.  
Also note that $P (e_{ij} \ot I) = (I \ot e_{ij}) P$. Taking the transpose of this, we obtain a pair of relations for $Q$ and any $M\in\End(\C^N)$:
\equ{ \label{Q_rel}
Q\, (M \ot I) =  Q\, (I \ot M^t) , \qq  (M \ot I)\, Q  = (I \ot M^t)\, Q .
}
For a matrix $X$ with entries $x_{ij}$ in an associative algebra $A$ we write
\[
X_s = \sum_{1\le i,j\le N} \underbrace{ I \ot \cdots \ot I}_{s-1} \ot\;e_{ij} \ot I \ot \cdots \ot I \ot x_{ij} \in \End(\C^N)^{\ot k} \ot A.
\]
Here $k \geq 2$ and $1\le s\le k$; it will always be clear from the context what $k$ is.
Products of matrix operators will be ordered using the following rules:
\equ{ \label{orderprod} 
\prod_{i=1}^s X_i = X_1 \, X_2 \cdots X_s  \qq \text{and} \qq \prod_{i=s}^1 X_i = X_s \, X_{s-1} \cdots X_1.  
}


\subsection{The Yangian \texorpdfstring{$Y(\mfgl_{2n})$}{} and twisted Yangian \texorpdfstring{$Y^\pm(\mfgl_{2n})$}{} } \label{sec:Y}

We briefly recall necessary details of the Yangian $Y(\mfgl_{2n})$, the twisted Yangian $Y^\pm(\mfgl_{2n})$ and their representation theory, adhering closely to \cite{Mo3}. 
Introduce a rational function acting on $\C^{2n}\ot \C^{2n}$
\equ{
R(u) = I - u^{-1} P \label{R(u)}
}
called the \emph{Yang's $R$-matrix}.
It satisfies $R(u)\,R(-u) = (1-u^{-2})\,I$ and is a solution of the quantum Yang-Baxter equation,
\equ{
R_{12}(u-v)\,R_{13}(u-z)\,R_{23}(v-z) = R_{23}(v-z)\,R_{13}(u-z)\,R_{12}(u-v).
}

We introduce elements $t_{ij}^{(r)}$ with $1 \le i,j \le 2n$ and $r\ge 0$ such that $t^{(0)}_{ij}= \del_{ij}$. Combining these into formal power series $t_{ij}(u) = \sum_{r\ge 0} t_{ij}^{(r)} u^{-r}$, we can then form the generating matrix $T(u)= \sum_{1\le i,j\le 2n} e_{ij} \ot t_{ij}(u)$.

\begin{defn}
The Yangian $Y(\mfgl_{2n})$ is the unital associative $\C$-algebra generated by elements $t_{ij}^{(r)}$ with $1 \le i,j \le 2n$ and $r\in\Z_{\ge 0}$ satisfying the relations
\gat{ 
R_{12}(u-v)\,T_1(u)\,T_2(v) = T_2(v)\,T_1(u)\,R_{12}(u-v) . \label{Y:RTT} 
}
The Hopf algebra structure of $Y(\mfgl_{2n})$ is given by 
\equ{ \label{Hopf:Y}
\Delta: T(u)\mapsto T(u)\ot T(u), \qq S: T(u)\mapsto T^{-1}(u),\qquad \veps: T(u)\mapsto I. 
}
\end{defn}

For the explicit form of \eqref{Y:RTT} we refer the reader to Section 1.1 in \cite{Mo3}.
We now recall the definition of the lowest weight representation of $Y(\mfgl_{2n})$. It is a historic tradition (but merely a convenion) to consider lowest weight representations in the algebraic Bethe ansatz instead of the highest ones.

\begin{defn}
A representation $V$ of $Y(\mfgl_{2n})$ is called a lowest weight representation if there exists a nonzero vector $\eta\in V$ such that $V=Y(\mfgl_{2n})\,\eta$ and
\aln{
t_{ij}(u)\,\eta &= 0 &&\qu\text{for}\qu 1 \le j<i \le 2n 
\qu\text{and}\\
t_{ii}(u)\,\eta &= \la_i(u)\,\eta &&\qu\text{for}\qu 1 \leq i \leq 2n,
}
where $\la_i(u)$ is a formal power series in $u^{-1}$ with a constant term equal to $1$. The vector $\eta$ is called the lowest vector of $V$, and the $2n$-tuple $\la(u)=(\la_1(u),\ldots,\la_{2n}(u))$ is called the lowest weight of $V$. 
\end{defn}

The Yangian $Y(\mfgl_{2n})$ contains the universal enveloping algebra $U(\mfgl_{2n})$ as a Hopf subalgebra. An embedding  $U(\mfgl_{2n})\into Y(\mfgl_{2n})$ is given by $E_{ij} \mapsto -t_{ji}^{(1)}$ for all $1\le i,j\le 2n$. We will identify $U(\mfgl_{2n})$ with its image in $Y(\mfgl_{2n})$ under this embedding. Conversely, the map $t_{ij}^{(1)} \mapsto -E_{ji}$ and $t_{ij}^{(s)} \mapsto 0$ for all $s\ge2$ defines a surjective homomorphism $ev : Y(\mfgl_{2n}) \to U(\mfgl_{2n})$ called the \emph{evaluation homomorphism}. By composing the map $ev$ with the algebra automorphism called the \emph{shift automorphism}, 
\[
\si_c \;:\; Y(\mfgl_{2n}) \to Y(\mfgl_{2n}) , \qu T(u) \mapsto T(u-c)
\]
for any $c\in\C$, we obtain the map
\equ{
ev_c \;:\; t_{ij}(u) \mapsto \del_{ij} - E_{ji} (u-c)^{-1}  .  \label{ev-hom}
}

Given an $2n$-tuple $\la=(\la_1,\ldots,\la_{2n})\in\C^{2n}$ we will denote by $L(\la)$ the irreducible representation of the Lie algebra $\mfgl_{2n}$ with the highest weight $\la$. In particular, $L(\la)$ is a cyclic $U(\mfgl_{2n})$-module generated by a nonzero vector $1_{\la}$ such that 
\aln{
E_{ij}\,1_{\la} &= 0 \qu&&\text{for}\qu 1 \le i<j \le 2n 
\qu\text{and}\\
E_{ii}\,1_{\la} &=\la_i\,1_{\la} \qu&&\text{for}\qu 1 \leq i \leq 2n.
}
The representation $L(\la)$ is finite-dimensional if and only if $\la_i-\la_{i+1} \in \Z_{\ge0}$ for all $1\le i\le 2n-1$.
By the virtue of the map $ev$, any $\mfgl_{2n}$-representation can be regarded as $Y(\mfgl_{2n})$-module. Moreover, any irreducible $\mfgl_{2n}$-representation remains irreducible over $Y(\mfgl_{2n})$, by surjectivity of $ev$. We will denote by $L(\la)_c$ the $Y(\mfgl_{2n})$-module obtained from the irreducible representation $L(\la)$ of $\mfgl_{2n}$ via the map \eqref{ev-hom}. Clearly, $L(\la)_c$ is a lowest weight $Y(\mfgl_{2n})$-module with the components of the lowest weight given by
\[
\la_i(u) = 1 - \la_i(u-c)^{-1} \qu\text{for}\qu 1\le i\le 2n .
\]

Fix $\ell\in\N$ and consider the tensor product of the $Y(\mfgl_{2n})$ evaluation modules 
\equ{
L := L(\la^{(1)})_{c_1} \ot L(\la^{(2)})_{c_2} \ot \ldots \ot L(\la^{(\ell)})_{c_\ell} , \label{L}
}
where $c_i\in\C$ are arbitrary complex numbers and each $\la^{(i)}$ is a partition of length not exceeding $\ell$. Set $\Delta^{(1)} := id$ and define recursively $\Delta^{(\ell)} = (id\ot\cdots\ot id \ot \Delta)\circ \Delta^{(\ell-1)}$ with $\Delta^{(2)} := \Delta$. The comultiplication $\Delta^{(\ell)}$ allows us to equip $L$ with the structure of a lowest weight $Y(\mfgl_{2n})$-module using the rule 
\[
t_{ij}(u)\cdot L = (ev_{c_1} \ot \cdots \ot ev_{c_l}) \circ \Delta^{(\ell)}(t_{ij}(u))\, L. 
\]
In particular, the generating matrix $T(u)$ acts on the space $L$ by
\equ{
T_a(u)\cdot L = \Bigg( \prod_{i=1}^\ell \mc{L}_{ai}(u-c_i) \Bigg) L \in \End(\C^{2n}) \ot L \,[[u^{-1}]], \label{T(u).L}
}
where 
\equ{
\mc{L}(u-c):=(id\ot ev_c)(T(u)) = \sum_{i,j=1}^{2n} e_{ij} \ot (\del_{ij} - E_{ji}(u-c)^{-1}) \label{lax}
}
are the Lax operators. The components of the lowest weight of $L$ are
\equ{
\la_i(u) = \prod_{j=1}^{\ell} \big(1-\la^{(j)}_i(u-c_j)^{-1}\big) \qu\text{for}\qu 1\le i \le 2n. \label{L:la(u)}
}
The \emph{binary property} of the tensor products of Yangian modules states that, for a suitable choice of weights $\la^{(j)}_i$ and parameters $c_j$, the $Y(\mfgl_{2n})$-module $L$ is irreducible, see Theorem~6.5.8 in \cite{Mo3}.


We now focus on the twisted Yangian $Y^\pm(\mfgl_{2n})$ and its representation theory. Following \cite{ACDFR1} we introduce an additional ``shift'' parameter $\rho\in\C$ in the definition of $Y^\pm(\mfgl_{2n})$.

\begin{defn} 
The twisted Yangian $Y^\pm(\mfgl_{2n})$ is the subalgebra of $Y(\mfgl_{2n})$ generated by the coefficients of the entries of the matrix 
\equ{
S(u) = T(u)\,T^t(-u-\rho) . \label{S=TT}
}
\end{defn}

The ``$\rho$-shifted'' twisted Yangian defined above is isomorphic to the usual one studied in \cite{Mo3}. The isomorphism is provided by the mapping $S(u) \mapsto S(u+\rho/2)$.
The matrix $S(u)$ defined in \eqref{S=TT} satisfies the reflection equation 
\equ{
R_{12}(u-v)\,S_1(u)\,R_{12}^t(-u-v-\rho)\,S_2(v) = S_2(v)\,R_{12}^t(-u-v-\rho)\,S_1(u)\,R_{12}(u-v) \label{RE}
}
and the symmetry relation
\equ{
S^t(-u-\rho) = S(u) \pm \frac{S(u)-S(-u-\rho)}{2u+\rho} \,. \label{symm}
}
The above two relations are in fact the defining relations of $Y^\pm(\mfgl_{2n})$. Their form in terms of matrix elements $s_{ij}(u)$ of $S(u)$, for $\rho=0$, can be found in (2.4) and (2.5) of \cite{Ol} (note that indices $i,j,k,l$ are indexed by $-n,-n+1,\ldots,n-1,n$ in \emph{loc.~cit.}); also see Section 4.1 in \cite{Mo3}.

By Proposition 5.3 in \cite{AACDFR}, any invertible matrix solution of the reflection equation \eqref{RE} (with $S(u)$ replaced by an element of $\End(\C^{2n})(u)$) is a constant, up to multiplication by a scalar function, matrix $K\in\End(\C^{2n})$ such that $K^t = \veps K$ with $\veps = +1$ or $-1$. Consider the subalgebra $Y^\pm_K(\mfgl_{2n})\subset Y(\mfgl_{2n})$ generated by coefficients of the entries of the matrix
\equ{
S^K(u) = T(u)\,K\,T^t(-u-\rho) . \label{SK=TKT}
}
Without loss of generality we can assume that $K = A\,I^\veps A^t$, where $I^\veps = \sum_{i=1}^n (e_{ii} + \veps\,e_{\bar\imath\bar\imath})$ and $A \in \End(\C^{2n})$ is an invertible matrix.

\begin{prop} \label{P:TY-iso}
The mapping
\equ{
\psi_K : S(u) \mapsto A S^{K}(u) A^t I^\veps,
}
where matrices $I^\veps$ and $A$ are as described above, defines an isomorphism of algebras $Y^\pm(\mfgl_{2n}) \to Y^\pm_K(\mfgl_{2n})$ if $\veps=+1$ and $Y^\pm(\mfgl_{2n}) \to Y^\mp_K(\mfgl_{2n})$ if $\veps=-1$.
\end{prop}

\begin{proof}
Let $t_+$ (resp.~$t_-$) denote the orthogonal (resp.~symplectic) transposition. Then $I^- (e_{ij})^{t_{\pm}} I^- = (e_{ij})^{t_{\mp}}$. In other words, conjugation with the matrix $I^\veps$ when $\veps=-1$ interchanges the orthogonal and symplectic transpositions.
Recall that the mapping $\al_A : T(u) \mapsto A\,T(u) A^{-1}$ for any invertible matrix $A\in\End(\C^{2n})$ defines an automorphism of $Y(\mfgl_{2n})$. We may thus rewrite the image of the matrix $S(u)$ under the mapping $\psi_K$ as
\aln{
A S(u) A^{t_\pm} &= A\, T(u) A^{-1} A A^{t_\pm} (A^{t_\pm})^{-1} T(-u-\rho)))^{t_\pm} A^{t_\pm} \\
& = \al_A (T(u) K T^{t_\pm}(-u-\rho))\\ & \in Y^\pm_K(\mfgl_{2n})[[u^{-1}]]
\intertext{if $\veps=+1$ and}
A S(u) A^{t_\pm} I^- &= A\, T(u) A^{-1} A I^-I^- A^{t_\pm} (A^{t_\pm})^{-1} T(-u-\rho)))^{t_\pm} A^{t_\pm} I^- \\
& = \al_A (T(u) K I^- T^{t_\pm}(-u-\rho) I^-)\\ & \in Y^\mp_K(\mfgl_{2n})[[u^{-1}]]
}
if $\veps=-1$.
\end{proof}

\begin{rmk}
In the algebraic Bethe ansatz approach the matrix $K$ defines the right boundary conditions of the open spin chain. Proposition \ref{P:TY-iso} implies that it is sufficient to consider the case when $K=I$. However, the mapping $\psi_K$ has an effect on the left boundary conditions; this will be discussed in more detail in Section~\ref{sec:NABA}. 
\end{rmk}

We now turn to representation theory of $Y^\pm(\mfgl_{2n})$. As in the case of $Y(\mfgl_{2n})$, we will be interested in the lowest weight representations.

\begin{defn}
A representation $V$ of $Y^\pm(\mfgl_{2n})$ is called a lowest weight representation if there exists a nonzero vector $\xi\in V$ such that $V=Y^\pm(\mfgl_{2n})\,\xi$ and
\aln{
s_{ij}(u)\,\xi &= 0 \qu &&\text{for}\qu 1 \le j < i \le {2n} \qu\text{and}\\
s_{ii}(u)\,\xi &= \mu_i(u)\,\xi \qu&&\text{for}\qu 1 \leq i \leq n,
}
where $\mu_i(u)$ are formal power series in $u^{-1}$ with constant terms equal to $1$. The vector $\xi$ is called the lowest weight vector of $V$, and the $n$-tuple $\mu(u)=(\mu_{1}(u),\ldots,\mu_{n}(u))$ is called the lowest weight of $V$. 
\end{defn}

Note that $\xi$ is also an eigenvector for the action of $s_{ii}(u)$ with $n+1\le i\le {2n}$. Indeed, the symmetry relation \eqref{symm} implies that 
\[
s_{\bar\imath\,\bar\imath}(u)\,\xi = \left( \mu_i(-u-\rho) \pm \frac{\mu_i(u)-\mu_i(-u-\rho)}{2u+\rho}\right)\xi \qu\text{for}\qu1\le i \le n.
\]

Writing $s_{ij}(u)=\sum_{r\ge0} s_{ij}^{(r)} u^{-r}$, the map $F_{ij} \mapsto -s_{ji}^{(1)}$ defines an embedding $U(\mfg_{2n})\into Y^\pm(\mfgl_{2n})$, where $\mfg_{2n} = \mfso_{2n}$ or $\mfsp_{2n}$. Conversely, the map 
\[
s_{ij}(u) \mapsto \del_{ij} - F_{ji}(u+(\rho\pm1)/2)^{-1}
\]
defines the evaluation homomorphism $ev^\pm : Y^\pm(\mfgl_{2n}) \to U(\mfg_{2n})$. (Note that there is no analogue of the shift automorphism $\si_c$ of $Y(\mfgl_{2n})$ for the twisted Yangian.)

Given an $n$-tuple $\mu=(\mu_1,\dots,\mu_n)\in\C^n$ we will denote by $M(\mu)$ the irreducible representation of the Lie algebra $\mfg_{2n}$ with the highest weight $\mu$. That is, $M(\mu)$ is a cyclic $U(\mfg_{2n})$-module generated by a nonzero vector $1_\mu$ such that
\aln{
\qq\;\; F_{ij}\,1_\mu &= 0 \qu&&\text{for}\qu 1 \le i<j \le 2n 
\qu\text{and}\\ 
F_{ii}\,1_\mu &= \mu_i\,1_\mu \qu&&\text{for}\qu 1 \leq i \leq n.
}
The representation $M(\mu)$ is finite-dimensional if and only if there exist integers $\mu_i$ satisfying
\aln{
 \mu_{i}-\mu_{i+1} \in \Z_{\ge 0} \qu &\text{for} \qu  &&1\le i\le n-1, \qq \\
 \mu_{n-1} + \mu_n \in \Z_{\ge 0} \qu &\text{if} \qu  &&\mfg_{2n} = \mfso_{2n}, \\
 \mu_n \in \Z_{\ge 0} \qu &\text{if} \qu &&\mfg_{2n} = \mfsp_{2n}.
 }
Using the evaluation homomorphism $ev^\pm$ we can extend each representation $M(\mu)$ to a lowest weight $Y^\pm(\mfgl_{2n})$-module with the lowest weight given by
\equ{
\mu_i(u) = 1-\mu_i(u+(\rho\pm 1)/2)^{-1} \qu\text{for}\qu 1\le i \le n. \label{mu(u)}
} 

The twisted Yangian $Y^\pm(\mfgl_{2n})$ is a left coideal subalgebra of $Y(\mfgl_{2n})$. In particular,
\equ{
\Delta : S(u) \mapsto (T(u)\ot 1)(1\ot S(u))(T^t(-u-\rho)\ot1) ,  \label{cop-s}
}
which is an element in $\End(\C^{2n})\ot Y(\mfgl_{2n}) \ot Y^\pm(\mfgl_{2n}) [[u^{-1}]]$.
This allows us to equip the space
\equ{
M := L\ot M(\mu) = L(\la^{(1)})_{c_1} \ot L(\la^{(2)})_{c_2} \ot \ldots \ot L(\la^{(\ell)})_{c_\ell} \ot M(\mu) \label{M}
}
with the structure of a lowest weight $Y^\pm(\mfgl_{2n})$-module. In particular, $S(u)$ acts on the space $M$ by
\equ{
S(u)\cdot M = \Bigg( \prod_{i=1}^{\ell} \mc{L}_i(u-c_i) \Bigg) \mc{L}^\pm(u)\, \Bigg( \prod_{i=\ell}^1 \mc{L}^t_i(-u-\rho-c_i) \Bigg) M, \label{S(u).M}
}
as an element of $\End(\C^{2n})\ot M \, [[u^{-1}]]$, where 
\equ{
\mc{L}^\pm(u) := (id \ot ev^\pm)(S(u)) = \sum_{i,j=1}^{2n} e_{ij} \ot (\del_{ij}- F_{ji}(u+(\rho\pm1)/2)^{-1}) \label{blax}
}
is the ``boundary'' Lax operator. Let $\xi\in M(\mu)$ be a lowest vector. Denote by $\eta_i$ the lowest vector of $L(\la^{(i)})_{c_i}$ and set $\zeta = \eta_1\ot \dots \ot \eta_{\ell} \ot \xi$. Then the submodule $Y^\pm(\mfgl_{2n})\,\zeta$ of $Y^\pm(\mfgl_{2n})$-module $M$ is a lowest weight representation with a lowest vector $\zeta$. It is given by 
\[
\la_i(u)\,\la_{2n-i+1}(-\rho-u)\,\mu_i(u)  \qu\text{for}\qu 1\le i \le n, 
\]
with $\la_i(u)$ defined in \eqref{L:la(u)} and $\mu_i(u)$ defined in \eqref{mu(u)}, see Proposition 4.2.11 in \cite{Mo3}. To the best of our knowledge, there are currently no irreducibility criteria known for a tensor product of irreducible representations of $Y(\mfgl_{2n})$ and $Y^\pm(\mfgl_{2n})$.

\begin{rmk}
Let $\la=(\la_1,\ldots,\la_n)$ and $\nu=(\nu_1,\ldots,\nu_n)$ be any $\mfgl_n$-weights. Set $\la'=(\la_2,\dots,\la_n)$ and $\nu'=(\nu_2,\dots,\nu_n)$. The algebraic Bethe anstaz for $Y(\mfgl_{n})$ relies on the fact that if the $Y(\mfgl_{N})$-module $L(\la)\ot L(\mu)$ is irreducible, so is the $Y(\mfgl_{N-1})$-module $L(\la')\ot L(\nu')$, see Lemma 6.2.2 in \cite{Mo3}. This property combined with the binary property allows one to solve the spectral problem of a $Y(\mfgl_N)$-system recursively, via the chain of subalgebras $Y(\mfgl_{n}) \supset Y(\mfgl_{n-1}) \supset \cdots \supset Y(\mfgl_{2})$. 
In this paper we will use the fact that the restriction of an irreducible $\mfso_{2n}$- or $\mfsp_{2n}$-representation of weight $\mu=(\mu_1,\ldots,\mu_n)$ to its natural $\mfgl_{n}$ subalgebra is irreducible. Moreover, any irreducible $\mfgl_{2n}$-module, upon restriction to the natural $\mfgl_n\oplus\mfgl_n\subset\mfgl_{2n}$ subalgebra, factors as a tensor product of its natural irreducible $\mfgl_n$-submodules. Hence, starting with the $Y^\pm(\mfgl_{2n})$-module $M$ we can restrict to an irreducible $Y(\mfgl_n)$-module $M^0\subset M$, provided the binary property holds. This resctriction will allow us to solve the spectral problem for a $Y^\pm(\mfgl_{2n})$-system using the chain of subalgebras $Y^\pm(\mfgl_{2n}) \supset Y(\mfgl_{n}) \supset Y(\mfgl_{n-1}) \supset \cdots \supset Y(\mfgl_{2})$.
\end{rmk}


\subsection{Block decomposition} \label{sec:block}

In this section, inspired by the arguments presented in \cite{Rs,DVK}, we demonstrate a block decomposition of the Yangian $Y(\mfgl_{2n})$ and the twisted Yangian $Y^\pm(\mfgl_{2n})$.  
We write matrices $T(u)$ and $S(u)$ in the block form:
\equ{
T(u) = \left(\begin{array}{cc} \ol{A}(u) & \ol{B}(u) \\ \ol{C}(u) & \ol{D}(u) \end{array}\right) , \qq 
S(u) = \left(\begin{array}{cc} A(u) & B(u) \\ C(u) & D(u) \end{array}\right) , \label{block}
}
Our goal is to derive the algebraic relations between these smaller matrix operators (blocks), which is the crucial first step of the algebraic Bethe ansatz. 
We will denote the matrix elements of $A(u)$ by $\sca_{ij}(u)$ with $1\le i,j\le n$, and similarly for matrices $B(u)$, $C(u)$ and $D(u)$, and their barred counterparts.

Recall that $\C^{2n} \cong \C^2 \ot \C^n$. Let ${\tt e}_{ij}$ with $1\le i,j\le 2n$ denote the standard matrix units of $\End(\C^{2n})$. Moreover, let $x_{ij}$ with $1\le i,j\le 2$ (resp.~$e_{ij}$ with $1\le i,j\le n$) denote the standard matrix units of $\End(\C^2)$ (resp.~$\End(\C^n)$).
Then, for any $1\le i,j \le n$, we may write
\equ{
{\tt e}_{ij} = x_{11} \ot e_{ij}, \qu {\tt e}_{n+i,j} = x_{21} \ot e_{ij}, \label{e=x*e}
}
and similarly for ${\tt e}_{i,n+j}$ and ${\tt e}_{n+i,n+j}$.
Hence any matrix $M\in \End(\C^{2n})$ with entries $(M)_{ij}\in \C$ can be equivalently written as
\[
M = \sum_{a,b=1}^2 x_{ab} \ot [M]_{ab} \in \End(\C^2) \ot \End(\C^n),
\]
where $[M]_{ab} = \sum_{i,j=1}^{n} (M)_{i+n(a-1),j+n(b-1)}\, e_{ij}$ are blocks of $M$, viz.~\eqref{block}.
Now let $M \in \End(\C^{2n} \ot \C^{2n})$. Then we may write
\[
M = \sum_{a,b,c,d=1}^2 x_{ab} \ot x_{cd} \ot [M]_{abcd} \in \End(\C^2\ot \C^2) \ot \End(\C^n \ot \C^n),
\]
where $[M]_{abcd}$ are obtained as follows. Writing $M = \sum_{i,j,k,l=1}^{2n} (M)_{ijkl}\,{\tt e}_{ij} \ot {\tt e}_{kl}$ we have
\equ{
[M]_{abcd} = \sum_{i,j,k,l=1}^{n} (M)_{i+n(a-1),j+n(b-1),k+n(c-1),l+n(d-1)}\, e_{ij} \ot e_{kl} . \label{M-block}
}

Denote the $R$-matrix \eqref{R(u)} acting on $\C^{2n}\ot \C^{2n}$ by $\RR(u)$ and its $t$-trans\-pose by $\RR^t(u)$. Viewing them as elements in $\End(\C^2\ot \C^2) \ot \End(\C^n \ot \C^n)[[u^{-1}]]$ and using \eqref{M-block} we recover the six-vertex block structure
\ali{
\RR(u) = \left(\begin{array}{cccc}  \!R(u) & \\ & I& \!-u^{-1}P\\ & \!-u^{-1}P & I \\ &&&  R(u)\! \end{array}\right) ,
\qu
\RR^t(u) = \left(\begin{array}{cccc} \;I\; & \\ &  R^t(u) & \;\mp u^{-1}Q\; \\ & \;\mp u^{-1}Q\; &  R^t(u)\\ &&& \;I\; \end{array}\right) , \label{R:new}
}
where the operators inside the matrices are each acting on $\C^n \ot \C^n$; note that $R^t(u)=I - u^{-1} Q$ and $Q=\sum_{1\le i,j\le n}e_{ij} \ot e_{\bar\jmath\,\bar\imath}$ in both cases of $\mp$ above are of the orthogonal type (recall the notation $\bar\imath=n-i+1$).

In a similar way, the matrices $T_1(u)=T(u)\ot I$ and $T_2(u)= I \ot T(u)$, as elements of $\End(\C^2\ot \C^2) \ot \End(\C^n \ot \C^n) \ot Y(\mfgl_{2n})[[u^{-1}]]$, take the form
\ali{
T_1(u) = \left(\begin{array}{cccc} \ol{A}_1(u) & & \ol{B}_1(u) \\ & \ol{A}_1(u)&& \ol{B}_1(u)\\ \ol{C}_1(u) & &\ol{D}_1(u) \\ & \ol{C}_1(u)&& \ol{D}_1(u)\end{array}\right) ,
\qu
T_2(u) = \left(\begin{array}{cccc} \ol{A}_2(u) & \ol{B}_2(u) \\ \ol{C}_2(u) & \ol{D}_2(u)\\ && \ol{A}_2(u) & \ol{B}_2(u) \\ && \ol{C}_2(u) & \ol{D}_2(u)\\ \end{array}\right) ,  \label{T:new}
}
where $\ol{A}_1(u)$ means $\ol{A}(u) \ot I \in \End(\C^n \ot \C^n) \ot Y(\mfg_{2n})[[u^{-1}]]$ with $I$ being the identity matrix, and similarly for the other blocks.
Substituting \eqref{R:new} and \eqref{T:new} to \eqref{Y:RTT} allows us to rewrite the defining relations of $Y(\mfgl_{2n})$ in terms of the matrices $\ol{A}(u)$, $\ol{B}(u)$, $\ol{C}(u)$ and $\ol{D}(u)$. The relations that we will need are:
\gat{
R_{12}(u-v)\, \ol{A}_1(u)\, \ol{A}_2(v) = \ol{A}_2(v)\, \ol{A}_1(u)\, R_{12}(u-v), \label{Y:AA}\\
R_{12}(u-v)\, \ol{D}_1(u)\, \ol{D}_2(v) = \ol{D}_2(v)\, \ol{D}_1(u)\, R_{12}(u-v), \label{Y:DD}\\
\ol{C}_1(u)\, \ol{A}_2(v) = \ol{A}_2(v)\, \ol{C}_1(u)\, R_{12}(u-v) + \frac{P_{12}\, \ol{A}_1(u)\, \ol{C}_2(v)}{u-v}, \label{Y:CA}\\
\ol{C}_1(u)\, \ol{D}_2(v) = R_{12}(v-u)\, \ol{D}_2(v)\, \ol{C}_1(u) - \frac{P_{12}\, \ol{D}_2(u)\, \ol{C}_1(v)}{u-v}, \label{Y:CD}\\
\ol{D}_1(u)\, \ol{A}_2(v) - \ol{A}_2(v)\, \ol{D}_1(u) = \frac{P_{12}\, \ol{B}_1(u)\, \ol{C}_2(v) - \ol{B}_2(v)\, \ol{C}_1(u)\, P_{12}}{u-v}. \label{Y:AD}
}
In particular, the coefficients of the matrix entries of $\ol{A}(u)$ generate a $Y(\mfgl_n)$ subalgebra of $Y(\mfgl_{2n})$. The same is true for $\ol{D}(u)$.

We now repeat the same steps for the twisted Yangian $Y^\pm(\mfgl_{2n})$. We substitute \eqref{R:new} to \eqref{RE} and view matrices $S_1(u)$ and $S_2(u)$ as elements of $\End(\C^2\ot \C^2) \ot \End(\C^n \ot \C^n) \ot Y^\pm(\mfgl_{2n})[[u^{-1}]]$, so that they take the same form as in \eqref{T:new}. 
This allows us to write the defining relations of $Y^\pm(\mfgl_{2n})$ in terms of the matrices $A(u)$, $B(u)$, $C(u)$ and $D(u)$. 
The relations that we will need are: 
\ali{
& A_2(v) B_1(u) = R_{12}(u-v) B_1(u) R_{12}^t(-u-v-\rho) A_2(v) \el & \hspace{2.15cm} + \frac{P_{12} B_1(v) R_{12}^t(-u-v-\rho) A_2(u)}{u-v} \mp \frac{B_2(v) Q_{12} D_1(u)}{u+v+\rho} , \!\!\label{AB}
\\ 
& R_{12}(u-v) B_1(u) R_{12}^t(-u-v-\rho) B_2(v) \el & \hspace{3.6cm} = B_2(v) R_{12}^t(-u-v-\rho) B_1(u) R_{12}(u-v) , \label{BB}
\\
& R_{12}(u-v) A_1(u) A_2(v) - A_2(v) A_1(u) R_{12}(u-v) \el & \hspace{1.76cm} = \mp\frac{ R_{12}(u-v) B_1(u) Q_{12} C_2(v)-B_2(v) Q_{12} C_1(u) R_{12}(u-v) }{u+v+\rho} , \!\! \label{AA} 
\\
& C_1(u) A_2(v) = A_2(v) R_{12}^t(-u-v-\rho) C_1(u) R_{12}(u-v) \el & \hspace{2.15cm} + \frac{P_{12} A_1(u) R_{12}^t(-u-v-\rho) C_2(v)}{u-v}\mp \frac{D_1(u) Q_{12} C_2(v)}{u+v+\rho} . \!\! \label{CA}
}

It remains to cast the symmetry relation \eqref{symm} in the block form. Observe that 
\[
S^t(u) = \left(\begin{array}{cc} D^t(u) & \pm B^t(u) \\ \pm C^t(u) & A^t(u) \end{array}\right) .
\]
This allows us immediately to extract linear relations between matrices $A(u)$, $B(u)$, $C(u)$ and $D(u)$, of which we will need the following two only: 
\ali{
D^t(-u-\rho) &=A(u)\pm \frac1{2u+\rho} \big(A(u)-A(-u-\rho) \big) ,  \label{symmD} \\
\pm B^t(-u-\rho) &= B(u) \pm \frac{1}{2u+\rho}(B(u) - B(-u-\rho)).  \label{symmB}
}


\section{Exchange identities}

In this section we derive algebraic relations between certain elements of the twisted Yangian $Y^\pm(\mfgl_{2n})$ that will be used in the derivation of the nested Bethe equations in the section that follows below. In particular, we recast the exchange relations (\ref{AB}-\ref{CA}) so that they can be applied directly to the algebraic Bethe ansatz. We then introduce the nested monodromy matrix and show its relevant algebraic properties. We note that all the operators in (\ref{AB}-\ref{CA}), viewed as matrices, will act on the vector spaces $V_a,V_{a_1},V_{a_2},\ldots,$ and $V_{\tl{a}},V_{\tl{a}_1},V_{\tl{a}_2},\ldots,$ all isomorphic to $\C^n$, which we call the \emph{auxiliary spaces}.  
We will often make use of the following rational function 
\equ{
p(u)=1 \pm \frac1{2u+\rho}. \label{p(v)}
}


\subsection{Creation operator for a single excitation}

The key operators in the construction of the Bethe vector will come from the $B$ block, viz.~\eqref{block}. However, rather than use a matrix of creation operators, we reinterpret $B(u)$ as a row vector in two auxiliary spaces, with components given by the matrix elements of $B(u)$. 
\nc{\at}{\tilde{a}}
\nc{\bee}[2]{\beta_{\at_{#1} a_{#1}}(#2)}
\begin{defn} \label{D:beta} The creation operator is given by
\equ{
\beta(u) := \sum_{1\le i,j\le n} e^*_i \ot e^*_j \ot \scb_{\bar \imath j}(u) \in (\C^n)^* \ot (\C^n)^* \ot Y^\pm(\mfgl_{2n})[[u^{-1}]]. \label{beta}
}
\end{defn}

The two auxiliary spaces are labelled in the same order as the tensor product, that is, $\beta_{\at a}(u) \in V_{\at}^* \ot V_{a}^* \ot Y^\pm(\mfgl_{2n})[[u^{-1}]]$.
The exchange and symmetry relations involving the $B$ operator may now be rewritten using the above notation.

\begin{lemma} \label{L:betaex}
The creation operator satisfies the following identities:
\ali{
&\beta_{\at_1 a_1}(u_1) \, \beta_{\at_2 a_2}(u_2) \, R_{a_1 \at_2}(-u_1-u_2-\rho) \, \check{R}_{\at_1 \at_2}(u_1-u_2) \el
& \qq = \beta_{\at_1 a_1}(u_2) \, \beta_{\at_2 a_2}(u_1) \, R_{a_1 \at_2}(-u_1-u_2-\rho) \, \check{R}_{a_1 a_2}(u_1-u_2) , \label{betabeta} \\
& \bee{i}{u}\, Q_{a_ia} =\pm \bigg(p(-u-\rho)\, \bee{i}{-u-\rho}\pm \frac{\bee{i}{u} }{2u+\rho} \bigg)\, Q_{\at_ia}\, Q_{a_ia}, \label{BPQ}
}
where $\check{R}(u) := P R(u)$.
\end{lemma}

\begin{proof}
We start by proving \eqref{betabeta}. From \eqref{BB}, begin by acting from the left with $P_{12}$, then use the defining property of the permutation operator to move it to the right on the r.h.s.~of the equation to obtain
\equ{ 
\check{R}_{12}(u_1-u_2) B_{1}(u_1) R^t_{12}(-u_1-u_2-\rho) B_{2}(u_2) = B_{1}(u_2) R^t_{12}(-u_1-u_2-\rho) B_{2}(u_1)\check{R}_{12}(u_1-u_2). \label{BB_check}
}
We want to rewrite this in terms of the creation operators defined in Definition~\ref{D:beta}. Choose bases for $V_{1}$ and $V_{2}$, then denote the matrix components of $R_{12}(-u_1-u_2-\rho)$ by $r_{i_1 j_1 i_2 j_2}$, and the matrix components of $\check{R}_{12}(u_1-u_2)$ by $\check{r}_{i_1 j_1 i_2 j_2}$. In components, \eqref{BB_check} becomes
\aln{
& \sum_{j_1,j_2,k_1,k_2=1}^n \check{r}_{i_1 j_1 i_2 j_2} \, b_{j_1 k_1}(u_1) \, r_{k_1 l_1 \bar{k}_2 \bar{\jmath}_2} \, b_{k_2 l_2}(u_2) = \sum_{j_1,j_2,k_1,k_2=1}^n b_{i_1 j_1}(u_2) \, r_{j_1 k_1 \bar \jmath_2 \bar \imath_2} \, b_{j_2 k_2}(u_1) \, \check{r}_{k_1 l_1 k_2 l_2}.
}
Relabelling $i_1 \rightarrow \bar \imath_1$ and $i_2 \rightarrow \bar \imath_2$, and relabelling the summation indices $j_1 \rightarrow \bar \jmath_1$ and $j_2 \rightarrow \bar \jmath_2$ yields an equivalent expression:
\aln{
& \sum_{j_1,j_2,k_1,k_2=1}^{n}  b_{\bar \jmath_1 k_1}(u_1) \, b_{\bar k_2 l_2}(u_2) \, r_{k_1 l_1 k_2 j_2} \, \check{r}_{\bar \imath_1 \bar \jmath_1 \bar \imath_2 \bar \jmath_2} = \sum_{j_1,j_2,k_1,k_2=1}^{n} b_{\bar \imath_1 j_1}(u_2) \, b_{\bar \jmath_2 k_2}(u_1) \, r_{j_1 k_1 j_2 i_2}  \, \check{r}_{k_1 l_1 k_2 l_2}.
}
Finally, we note that $\check{r}_{\bar \imath_1 \bar \jmath_1 \bar \imath_2 \bar \jmath_2} = \check{r}_{j_1 i_1 j_2 i_2}$, as $\check{R}_{a b}(u)^{t_a t_b} =\check{R}_{a b}(u)$. Then taking the tensor product with $e^*_{i_1} \ot e^*_{l_1} \ot e^*_{i_2} \ot e^*_{l_2} \in V_{\at_1}^* \ot V_{a_1}^* \ot V_{\at_2}^* \ot V_{a_2}^*$ and summing over these indices yields
\aln{
&\beta_{\at_1 a_1}(u_1) \, \beta_{\at_2 a_2}(u_2) \, R_{a_1 \at_2}(-u_1-u_2-\rho) \, \check{R}_{\at_1 \at_2}(u_1-u_2) \\
& \qq = \beta_{\at_1 a_1}(u_2) \, \beta_{\at_2 a_2}(u_1) \, R_{a_1 \at_2}(-u_1-u_2-\rho) \, \check{R}_{a_1 a_2}(u_1-u_2) ,
}
as required.

We now focus on \eqref{BPQ}. From \eqref{symmB} in matrix components, we make the assignment $u \mapsto -u-\rho$ and multiply by $\pm$ to obtain
\[
\scb_{\bar \jmath \bar \imath}(u) =  \pm\, p(-u-\rho)\, \scb_{ij}(-u-\rho) + \frac{\scb_{ij}(u)}{2u+\rho} .
\]
Then, taking the tensor product with $e^*_j \ot e^*_{\bar \imath} \in V^*_{\at_i} \ot V^*_{a_i}$, and summing over $i,j$ yields the following expression in terms of the creation operator:
\aln{
\beta_{\at_i a_i}(u) &= \pm p(-u-\rho) \beta_{a_i \at_i}(-u-\rho) + \frac{\beta_{a_i \at_i}(u)}{2u+\rho} =\pm \bigg( p(-u-\rho) \beta_{\at_i a_i}(-u-\rho) \pm \frac{\beta_{\at_i a_i}(u)}{2u+\rho} \bigg) P_{\at_i a_i}.
}
To obtain \eqref{BPQ} from here, we multiply on the right by the operator $Q_{a_i a}$ and use the identity $P_{\at_i a_i} Q_{a_i a} = Q_{\at_i a} Q_{a_i a}$.
\end{proof}


\subsection{Rewriting the AB exchange relation} 

We need to rewrite the AB exchange relation \eqref{AB} in terms of the creation operator \eqref{beta}.
The form of this relation will be the cornerstone of the nesting procedure.

\begin{lemma} \label{L:ABetaHalf} The following identity holds:
\ali{
A_a(v)\, \bee{i}{u} &= \bee{i}{u}\, R^t_{\at_ia}(u-v)\, R^t_{a_ia}(-u{-}v{-}\rho)\, A_a(v) \el
& \qu + \frac{\bee{i}{v}}{u-v}\, Q_{\at_ia}\, R^t_{a_ia}(-2u-\rho)\, A_a(u) \el &\qu \mp \frac{p(-u-\rho)}{u+v+\rho}\, \bee{i}{v}\, Q_{\at_ia} \,Q_{a_ia}\, A_a(-u-\rho). \label{ABetaHalf} 
}
\end{lemma}

\begin{proof}
We introduce the following rule for obtaining expressions in terms of the $\beta$ operator from those in terms of $B$ operator.
Let $X_{\at}\in \End(V_{\at})$ and $Y_a \in \End(V_a)$. Considering the components of $\beta_{\at a}(u)\, X^t_{\at}\, Y_a$, we have
\aln{
\beta_{\at a}(u)\, X^t_{\at}\, Y_a &= \sum_{1\le i,j,k,l,r,s\le n} (e^*_k \ot e^*_l \ot \scb_{\bar k l}(u)) (e_{ri}\ot e_{sj}\ot x_{\bar\imath \bar r}\, y_{sj}) \\ &= \sum_{1\le i,j,k,l\le n} e_i \ot e_j \ot \scb_{\bar k l}(u)\, x_{\bar\imath \bar k} \, y_{lj},
}
so that $(\beta_{\at a}(u)\, X^t_{\at}\, Y_a)_{\bar\imath j} = \sum_{1\le k,l\le n} b_{k l}(u)\,x_{ik}\,y_{lj}$. On the other hand, taking the $(i,j)$-th matrix element of the expression $X_a \,B_a(u)\, Y_a$ for any $X_a, Y_a \in \End(V_a)$ we obtain
\eqa{ \label{XBY} 
(X_a \,B_a(u)\, Y_a)_{ij} = \sum_{1\le k,l \le n} x_{ik}\, \scb_{kl}(u)\, y_{lj} = (\bee{}{u}\, X^t_{\at}\, Y_{a})_{\bar \imath j} .
}
Using this rule, we can rewrite \eqref{AB} as
\aln{
A_a(v) \beta_{\at_i a_i}(u) &= \beta_{\at_i a_i}(u) R^t_{\at_i a}(u-v) R^t_{a_i a} (-u-v-\rho) A_a(v) 
\\
&\quad + \frac{\beta_{\at_i a_i}(v)}{u-v} Q_{\at_i a} R^t_{a_i a}(-u-v-\rho) A_a(u) \\
&\quad \mp \frac{\beta_{\at_i a_i}(v)}{u+v+\rho} Q_{\at_i a} Q_{a_i a} D^t_a(u),
}
where the identities $X_1 = P_{12} X_2 P_{12}$ and $Q_{12} X_1 = Q_{12} X^t_2$ have been used. 
From here, the symmetry relation \eqref{symmD} may be used to obtain
\aln{
A_a(v) \beta_{\at_i a_i}(u) &= \beta_{\at_i a_i}(u) R^t_{\at_i a}(u-v) R^t_{a_i a} (-u-v-\rho) A_a(v) 
\\
&\quad + \frac{\beta_{\at_i a_i}(v)}{u-v} Q_{\at_i a} R^t_{a_i a}(-u-v-\rho)A_a(u) \\
&\quad \mp \frac{\beta_{\at_i a_i}(v)}{u+v+\rho} Q_{\at_i a} Q_{a_i a} \bigg( p(-u-\rho) \, A_a(-u-\rho) \pm \frac{A_a(u)}{2u+\rho} \bigg).
}
We note that 
\aln{
& \frac{R_{a_i a}^t(-u-v-\rho)}{u-v} - \frac{Q_{a_i a}}{(2u+\rho)(u+v+\rho)} = \frac1{u-v}\bigg(I+\bigg(1-\frac{u-v}{2u+\rho}\bigg)\frac{Q_{a_i a}}{u+v+\rho}\bigg) 
=\frac{R^t(-2u-\rho)}{u-v}.
}
Following these manipulations, we obtain
\ali{
A_a(v)\, \bee{i}{u} &= \bee{i}{u}\, R^t_{\at_ia}(u-v)\, R^t_{a_ia}(-u{-}v{-}\rho)\, A_a(v) \el
& \qu + \frac{\bee{i}{v}}{u-v}\, Q_{\at_ia}\, R^t_{a_ia}(-2u-\rho)\, A_a(u) \el &\qu \mp \frac{p(-u-\rho)}{u+v+\rho}\, \bee{i}{v}\, Q_{\at_ia} \,Q_{a_ia}\, A_a(-u-\rho) , 
}
as required.
\end{proof}

This relation \eqref{ABetaHalf} is convenient as it does not feature the $D$ operator, so the relation can be used repeatedly in the presence of multiple creation operators. However, to obtain the most convenient form of \eqref{ABetaHalf}, we must consider the action of $p(v)A_a(v) + p(-v-\rho)A_a(-v-\rho)$ on $\bee{i}{u}$ rather than of $A_a(v)$ alone (the motivation for this construction will be explained in Section \ref{sec:NABA}). Introduce the following notation for a symmetrised combination of functions or operators,
\equ{ 
\even{v}{f(v)} := f(v)+f(-\,v-\rho) . \label{v-even}
}
\begin{lemma} \label{L:AB}
The following identity holds:
\ali{ \label{ABetafull}
\even{v}{ p(v)\, A_a(v)} \bee{i}{u} & = \bee{i}{u} \even{v}{ p(v)\, R^t_{\at_ia}(u-v)\, R^t_{a_ia}(-u-v-\rho)\, A_a(v)} \el 
& \qu + \frac1{p(u)} \even{v}{ \frac{p(v)}{u-v} \bee{i}{v}} \Res{w \rightarrow u} \Big[ \even{w}{ p(w)\, R^t_{\at_ia}(u-w)\, R^t_{a_ia}(-u-w-\rho)\, A_a(w) } \Big] .
}
\end{lemma}

\begin{proof}  
Starting from \eqref{ABetaHalf}, multiplying by $p(v)$ and symmetrising using \eqref{v-even}, we obtain
\ali{ \label{AB_have} 
\even{v}{p(v)\,A_a(v)} \bee{i}{u} & = \bee{i}{u} \even{v}{R^t_{\at_ia}(u-v)\, R^t_{a_ia}(-u-v-\rho)\, p(v)\, A_a(v)} \el
& \qu + \even{v}{\frac{p(v)}{u-v} \, \bee{i}{v}} Q_{\at_ia}\, R^t_{a_ia}(-2u-\rho)\, A_a(u) 
\el
& \qu \mp p(-u-\rho)\even{v}{\frac{p(v)}{u+v+\rho}\, \bee{i}{v}}  Q_{\at_ia}\, Q_{a_ia}\, A_a(-u-\rho) .
}
We will show that this is equivalent to \eqref{ABetafull} term by term, separating the terms by the parameter carried by $A_a( \cdot )$. Note that the term containing $A_a(v)$ is already the same in both \eqref{ABetafull} and \eqref{AB_have}. For the remaining terms, containing $A_a(u)$ and $A_a(-u-\rho)$, we will work backwards from \eqref{ABetafull}. Let
\[
U = \frac1{p(u)} \even{v}{ \frac{p(v)}{u-v} \,\bee{i}{v}} \Res{w \rightarrow u} \Big[ \even{w}{ p(w) \,R^t_{\at_ia}(u-w)\, R^t_{a_ia}(-u-w-\rho)\, A_a(w) } \Big].
\]
Furthermore, expand the symmetriser inside the residue so that $U = U_+ + U_-$, where
\aln{
U_+ &= \frac1{p(u)} \even{v}{ \frac{p(v)}{u-v} \,\bee{i}{v}} \Res{w \rightarrow u} \Big[ p(w) \,R^t_{\at_ia}(u-w)\, R^t_{a_ia}(-u-w-\rho)\, A_a(w) \Big],
\\
U_- &= \frac1{p(u)} \even{v}{ \frac{p(v)}{u-v} \,\bee{i}{v}} \Res{w \rightarrow u} \Big[ p(-w-\rho) \,R^t_{\at_ia}(u+w+\rho)\, R^t_{a_ia}(w-u)\, A_a(-w-\rho) \Big].
}
Focussing first on $U_+$, we evaluate the residue to obtain
\[
U_+ = \even{v}{ \frac{p(v)}{u-v} \,\bee{i}{v}} Q_{\at_ia}\, R^t_{a_ia}(-2u-\rho)\, A_a(u) .
\]
This now matches the term containing $A_a(u)$ in \eqref{AB_have}. It remains to show that $U_-$ is equal to the term containing $A_a(-u-\rho)$ in \eqref{AB_have}. Again evaluating the residue, we obtain
\aln{
U_- &= - \even{v}{ \frac{p(v)}{u-v} \,\bee{i}{v}} \frac{p(-u-\rho)}{p(u)} \, R^t_{\at_ia}(2u+\rho)\, Q_{a_ia}\, A_a(-u-\rho) 
\\
&= - \even{v}{ \frac{p(v)}{u-v} \bigg(\bee{i}{v}\,Q_{a_ia} -\bee{i}{v}\,\frac{Q_{\at_i a}Q_{a_ia}}{2u+\rho} \bigg)} \frac{p(-u-\rho)}{p(u)}\, A_a(-u-\rho) .
}
We now apply the symmetry relation \eqref{BPQ}, so
\[
U_- = - \even{v}{ \frac{p(v)}{u-v} \bigg(\pm p(-v-\rho)\, \bee{i}{-v-\rho}+ \frac{\bee{i}{v}}{2v+\rho}\,  -\frac{\bee{i}{v}}{2u+\rho} \bigg)} \frac{p(-u-\rho)}{p(u)}\,  Q_{\at_1 a}\,Q_{a_ia} \, A_a(-u-\rho).
\]
Since it lies within the symmetriser, the term containing $\bee{i}{-v-\rho}$ can be rewritten in terms of $\bee{i}{v}$ to obtain
\[
U_- = - \even{v}{ \bigg(\pm \frac{p(-v-\rho)}{u+v+\rho} + \frac{1}{u-v}\bigg(\frac1{2v+\rho} - \frac1{2u+\rho}\bigg) \bigg) \, p(v) \,\bee{i}{v}} \frac{p(-u-\rho)}{p(u)} \, Q_{\at_i a}\,Q_{a_ia} \, A_a(-u-\rho).
\]
All that remains are algebraic manipulations:
\aln{
U_- &= -  \even{v}{ 
\bigg(\pm \frac{p(-v-\rho)}{u+v+\rho} +\frac2{(2v+\rho)(2u+\rho)} \bigg) \, p(v) \,\bee{i}{v} }  \frac{p(-u-\rho)}{p(u)}\, Q_{\at_1 a}\,Q_{a_ia} \, A_a(-u-\rho)
\\
&= - \even{v}{ 
\bigg(\pm 1 -\frac1{2v+\rho} +\frac{2u+2v+2\rho}{(2v+\rho)(2u+\rho)} \bigg) \, \frac{p(v)}{u+v+\rho} \,\bee{i}{v} } \frac{p(-u-\rho)}{p(u)} \, Q_{\at_1 a}\,Q_{a_ia} \, A_a(-u-\rho)
\\
&= - \even{v}{ \bigg(\pm 1  +\frac1{2u+\rho} \bigg) \, \frac{p(v)}{u+v+\rho} \,\bee{i}{v}} \frac{p(-u-\rho)}{p(u)} \, Q_{\at_1 a}\,Q_{a_ia} \, A_a(-u-\rho)
\\
&= \mp p(-u-\rho) \even{v}{ 
\frac{p(v)}{u+v+\rho} \,\bee{i}{v} } Q_{\at_1 a}Q_{a_ia} \, A_a(-u-\rho) .
}
This matches the term containing $A_a(-u-\rho)$ in \eqref{AB_have} and completes the proof.
\end{proof}


\subsection{Creation operator for multiple excitations}

The next step is to generalize the creation operator $\beta$ defined in \eqref{beta} for multiple excitation. Choose $m \in \N$, the excitation number, and consider the tensor product space $W = V_{\at_1} \ot \cdots \ot V_{\at_m} \ot V_{a_1} \ot \cdots \ot V_{a_m}$. Denote its dual by $W^* = V^*_{\at_1} \ot \cdots \ot V^*_{\at_m} \ot V^*_{a_1} \ot \cdots \ot V^*_{a_m}$ and introduce an $m$-tuple of formal parameters $\bm u = (u_1,u_2,\dots,u_m)$.

\begin{defn} \label{D:m-beta}
The creation operator for $m$ excitations is given in terms of the ordered product  of $\beta$ operators and $R$-matrices (cf., \eqref{orderprod}):
\ali{ 
\bethe{m}{\bm u} &= \prod_{i=1}^m \bigg( \bee{i}{u_i} \prod_{j=i-1}^1 R_{a_j \at_i}(-u_j-u_i-\rho) \bigg) \el
& \in W^* \ot Y^\pm(\mfgl_{2n}) [u_1,\dots,u_m][[u^{-1}_1,\dots,u^{-1}_m]].\label{topB}
}
\end{defn}

Note that the creation operator for $m$ excitations satisfies the following recursive relation
\equ{ 
\bethe{m}{\bm u} = \bethe{m-1}{u_1, \dots, u_{m-1}}\, \bee{m}{u_m}\prod_{j=m-1}^1 R_{a_j \at_m}(-u_j-u_m-\rho). \label{beta-rec}
}

Given $i\in \{1,\dots,m-1\}$ denote by $\bm u_{i \leftrightarrow i+1}$ the $m$-tuple obtained from $\bm u$ by interchanging its $i$-th and $(i+1)$-th entries, namely
\equ{
\bm u_{i \leftrightarrow i+1} = (u_1,u_2, \ldots , u_{i-1},u_{i+1},u_i,u_{i+2},\ldots,u_m). \label{uii}
}
The Lemma below states a relation between the operators $\bethe{m}{\bm u}$ and $\bethe{m}{\bm u_{i \leftrightarrow i+1}}$ that will assist us in obtaining the explicit expressions of the so-called ``unwanted terms'' in Section \ref{sec:UWT}.

\begin{lemma} \label{L:B=BaReR}
The following identity holds:
\equ{ 
\bethe{m}{\bm u} = \bethe{m}{\bm u_{i \leftrightarrow i+1}}\, \check{R}_{a_i a_{i+1}}(u_i-u_{i+1})\, \check{R}^{-1}_{\at_i \at_{i+1}}(u_i-u_{i+1}) \label{B=BaReR}
}
for $1 \leq i \leq m-1$. 
\end{lemma}

\begin{proof} 
We use induction on $m$, with the basis case provided by \eqref{betabeta}.
Assume the result holds for $m-1$ excitations. There are two cases to consider, depending on the spaces $a_i, a_{i+1}$ on which $R_{a_i a_{i+1}}(u_i-u_{i+1})$ acts nontrivially. Consider first the case where $i < m-1$ and use the recursive relation \eqref{beta-rec}:
\aln{
\bethe{m}{\bm u} &= \bethe{m-1}{u_1, \dots, u_{m-1}} \, \bee{m}{u_m} \prod_{j=m-1}^1 R_{a_j \at_m}(-u_j-u_m-\rho) \\
& = \bethe{m-1}{u_1, \dots, u_{i+1},u_i, \dots, u_{m-1}} \,\check{R}_{a_i a_{i+1}}(u_i-u_{i+1})  \\
& \qu \times \check{R}^{-1}_{\at_i \at_{i+1}}(u_i-u_{i+1}) \, \bee{m}{u_m} \prod_{j=m-1}^1 R_{a_j \at_m}(-u_j-u_m-\rho).
}
Notice that the matrix $\check{R}^{-1}_{\at_i \at_{i+1}}(u_i-u_{i+1})$ commutes with all matrices to the right of it, so it can be moved to the very right. The matrix $\check{R}_{a_i a_{i+1}}(u_i-u_{i+1})$ may be moved through the product of $R$-matrices using the (braided) Yang-Baxter equation:
\eqn{
&\check{R}_{a_i a_{i+1}}(u_i-u_{i+1})R_{a_{i+1} \at_m}(-u_{i+1}-u_m-\rho) R_{a_i \at_m}(-u_i-u_m-\rho) 
\\
& \qq = R_{a_{i+1} \at_m}(-u_i-u_m-\rho) R_{a_i \at_m}(-u_{i+1}-u_m-\rho) \check{R}_{a_i a_{i+1}}(u_i-u_{i+1}).
}
This then gives \eqref{B=BaReR} for $i<m-1$. For $i=m-1$, we factorise the excitations as follows:
\aln{
\bethe{m}{\bm u} & = \bethe{m-2}{u_1, \dots, u_{m-2}} \, \beta_{\at_{m-1} a_{m-1} \at_m a_m}(u_{m-1},u_m) \\ 
& \qu \times \prod_{j=m-2}^1 \Big( R_{a_j \at_{m-1}}(-u_j-u_{m-1}-\rho)\, R_{a_j \at_m}(-u_j-u_m-\rho) \Big) 
\\
& = \bethe{m-2}{u_1, \dots, u_{m-2}} \, \beta_{\at_{m-1} a_{m-1} \at_m a_m}(u_m,u_{m-1}) \\[.5em] 
& \qu \times\check{R}_{a_{m-1} a_m}(u_{m-1}-u_m) \, \check{R}^{-1}_{\at_{m-1} \at_m}(u_{m-1}-u_m) \\
& \qu \times \prod_{j=m-2}^1 \Big( R_{a_j \at_{m-1}}(-u_j-u_{m-1}-\rho)\, R_{a_j \at_m}(-u_j-u_m-\rho) \Big).
}
The matrix $\check{R}^{-1}_{\at_{m-1} \at_m}(u_{m-1}-u_m)$ may be moved through the product of $R$-matrices using another variant of the Yang-Baxter equation,
\aln{
&\check{R}^{-1}_{\at_{m-1} \at_m}(u_{m-1}-u_m)\, R_{a_j \at_{m-1}}(-u_j-u_{m-1}-\rho)\, R_{a_j \at_m}(-u_j-u_m-\rho) \\
& \qq = R_{a_j \at_{m-1}}(-u_j-u_m-\rho)\, R_{a_j \at_m}(-u_j-u_{m-1}-\rho)\, \check{R}^{-1}_{\at_{m-1} \at_m}(u_{m-1}-u_m).
}
Then, rearranging the commuting matrices in the expression, we reconstruct the full excitation vector and arrive at \eqref{B=BaReR} for $i=m-1$. This completes the induction.
\end{proof}

\begin{rmk} \label{R:B=BaReR}
By definition, the operator $\bethe{m}{\bm u}$ in \eqref{topB} is an analogue of the fused boundary operator of Theorem 4.1 in \cite{BaRe} for the twisted reflection equation. More precisely, it is a solution to a fused analogue of twisted reflection equation \eqref{BB} in the sense of \emph{loc.~cit.} 
\end{rmk}


\subsection{The AB exchange relation for multiple excitations}

We want to move $\even{v}{p(v)\,A_a(v)}$ through the operator $\bethe{m}{\bm u}$. Each time $\even{v}{p(v)\,A_a(v)}$ is moved through one of the excitations $\bee{i}{u_i}$ using \eqref{ABetaHalf}, we obtain a term, where the parameter $v$ of $\even{v}{p(v)\,A_a(v)}$ is unchanged. We will call this term the \emph{wanted term}. All the additional terms will be called the \emph{unwanted terms}; we will denote them by $UWT$ and consider their exact form in Section~\ref{sec:UWT}. Focussing on the wanted term at each step, $\even{v}{p(v)\,A_a(v)}$ accrues $R$-matrices as it moves through the excitations. In the following lemma, we will show that these $R$-matrices may be moved through those appearing in the operator $\bethe{m}{\bm u}$. 

\begin{lemma} \label{L:movement}
The following exchange relation holds
\aln{
&\Bigg( \prod_{k=1}^{i-1} R^t_{\at_k a}(u_k-v) \Bigg)\Bigg( \prod_{l=1}^{i-1} R^t_{a_l a}(-u_l-v-\rho) \Bigg) A_a(v) \,\bee{i}{u_i} \prod_{j=i-1}^1 R_{a_j \at_i}(-u_j-u_i-\rho) \\
& \qu\;\; = \bee{i}{u_i} \,\Bigg( \prod_{j=i-1}^1 R_{a_j \at_i}(-u_j-u_i-\rho) \Bigg) \Bigg( \prod_{k=1}^{i} R^t_{\at_k a}(u_k-v) \Bigg) \Bigg( \prod_{l=1}^{i} R^t_{a_l a}(-u_l-v-\rho) \Bigg) A_a(v)  +UWT.
}
\end{lemma}

\begin{proof} 
We begin by using \eqref{ABetaHalf} and focus on the wanted terms only:
\aln{ 
&\Bigg( \prod_{k=1}^{i-1} R^t_{\at_k a}(u_k-v) \Bigg)\Bigg( \prod_{l=1}^{i-1} R^t_{a_l a}(-u_l-v-\rho) \Bigg) A_a(v)\, \bee{i}{u_i} \prod_{j=i-1}^1 R_{a_j \at_i}(-u_j-u_i-\rho) \\
& \qq = \Bigg( \prod_{k=1}^{i-1} R^t_{\at_k a}(u_k-v) \Bigg)\Bigg( \prod_{l=1}^{i-1} R^t_{a_l a}(-u_l-v-\rho) \Bigg)  \bee{i}{u_i} \, R^t_{\at_i a}(u_i-v) \\
& \qq\qu \times  R^t_{a_i a}(u_i-v) A_a(v) \prod_{j=i-1}^1 R_{a_j \at_i}(-u_j-u_i-\rho) +UWT 
}
yielding
\aln{
& \bee{i}{u_i} \Bigg( \prod_{k=1}^{i-1} R^t_{\at_k a}(u_k-v) \Bigg)\Bigg( \prod_{l=1}^{i-1} R^t_{a_l a}(-u_l-v-\rho) \Bigg) R^t_{\at_i a}(u_i-v) \\
& \hspace{4.5cm} \times \Bigg( \prod_{j=i-1}^1 R_{a_j \at_i}(-u_j-u_i-\rho) \Bigg) R^t_{a_i a}(u_i-v) A_a(v) +UWT.
}
All that remains is to rearrange the product of $R$-matrices in the centre of the expression. The matrices can be reordered using the Yang-Baxter equation
\aln{
& R^t_{a_{i-1} a}(-u_{i-1}-v-\rho) \, R^t_{\at_i a}(u_i-v) \, R_{a_{i-1} \at_i}(-u_{i-1}-u_i-\rho) 
\\ & \qq = R_{a_{i-1} \at_i}(-u_{i-1}-u_i-\rho) \, R^t_{\at_i a}(u_i-v) \, R^t_{a_{i-1} a}(-u_{i-1}-v-\rho) .
}
Thus the product of $R$-matrices becomes
\aln{
& \Bigg( \prod_{l=1}^{i-1} R^t_{a_l a}(-u_l-v-\rho) \Bigg) R^t_{\at_i a}(u_i-v) \Bigg( \prod_{j=i-1}^1 R_{a_j \at_i}(-u_j-u_i-\rho) \Bigg) \\
& \qq = \Bigg( \prod_{l=1}^{i-2} R^t_{a_l a}(-u_l-v-\rho) \Bigg) R_{a_{i-1} \at_i}(-u_{i-1}-u_i-\rho) R^t_{\at_i a}(u_i-v) \\
& \qq\qu \times R^t_{a_{i-1} a}(-u_{i-1}-v-\rho) \Bigg( \prod_{j=i-2}^1 R_{a_j \at_i}(-u_j-u_i-\rho) \Bigg)  \\
& \qq = R_{a_{i-1} \at_i}(-u_{i-1}-u_i-\rho) \Bigg( \prod_{l=1}^{i-2} R^t_{a_l a}(-u_l-v-\rho) \Bigg) R^t_{\at_i a}(u_i-v) \\
& \qq\qu \times \Bigg( \prod_{j=i-2}^1 R_{a_j \at_i}(-u_j-u_i-\rho) \Bigg) R^t_{a_{i-1} a}(-u_{i-1}-v-\rho) \\
& \qq =\Bigg( \prod_{j=i-1}^1 R_{a_j \at_i}(-u_j-u_i-\rho) \Bigg) R^t_{\at_i a}(u_i-v) \Bigg( \prod_{l=1}^{i-1} R^t_{a_l a}(-u_l-v-\rho) \Bigg),
}
where the last equality is achieved by inductively applying the same argument. Putting this together, and noting that the $R^t$-matrices all commute with the $R$-matrices, we arrive to
\aln{
&\Bigg( \prod_{k=1}^{i-1} R^t_{\at_k a}(u_k-v) \Bigg)\Bigg( \prod_{l=1}^{i-1} R^t_{a_l a}(-u_l-v-\rho) \Bigg) A_a(v) \, \bee{i}{u_i} \prod_{j=i-1}^1 R_{a_j \at_i}(-u_j-u_i-\rho) \\
& \qq = \bee{i}{u_i} \Bigg( \prod_{j=i-1}^1 R_{a_j \at_i}(-u_j-u_i-\rho) \Bigg) \Bigg( \prod_{k=1}^{i} R^t_{\at_k a}(u_k-v) \Bigg) \Bigg( \prod_{l=1}^{i} R^t_{a_l a}(-u_l-v-\rho) \Bigg) A_a(v)  + UWT
}
as required. 
\end{proof}

Applying this result to the product of $m$ such excitations in \eqref{topB} yields
\[
A_a(v)\,\bethe{m}{\bm u} = \bethe{m}{\bm u} \Bigg( \prod_{k=1}^{m} R^t_{\at_k a}(u_k-v) \Bigg) \Bigg( \prod_{l=1}^{m} R^t_{a_l a}(-u_l-v-\rho) \Bigg) A_a(v)  +UWT.
\]
We define the matrix on the right side to be the \emph{nested monodromy matrix}, 
\ali{ \label{nMM}
T_a(v;\bm u) := \Bigg( \prod_{k=1}^{m} R^t_{\at_k a}(u_k-v) \Bigg)\Bigg( \prod_{l=1}^{m} R^t_{a_l a}(-u_l-v-\rho) \Bigg) A_a(v).
}
Its matrix entries will be denoted by $t_{ij}(v;\bm u)$. The matrix $T_a(v;\bm u)$ allows us to write the above identity more compactly, 
\[ 
A_a(v)\,\bethe{m}{\bm u} = \bethe{m}{\bm u} \, T_a(v;\bm u) + UWT ,
\]
which leads to the following result.

\begin{crl} \label{C:wantedterm}
The AB exchange relation for the creation operator of multiple excitations has the form
\begin{flalign} 
&& \even{v}{p(v)\,A_a(v)} \bethe{m}{\bm u} = \bethe{m}{\bm u} \even{v}{p(v)\,T_a(v;\bm u)} + UWT. && \nn\qed
\end{flalign}
\end{crl}


\subsection{Exchange relations for the nested monodromy matrix} \label{sec:nested}

In this section we introduce a vector space $M'$, called the \emph{nested vacuum sector}, on which the nested monodromy matrix $T(v;\bm u)$ satisfies the usual RTT relation, \emph{viz.}~\eqref{Y:RTT}. This allows us to identify $T(v;\bm u)$ as the monodromy matrix for a residual $Y(\mfgl_n)$-system. The space $M'$ is then interpreted as the full quantum space of this residual system. We start by introducing certain subspaces of the evaluation modules $M(\mu)$ and $L(\la^{(i)})_{c_i}$ that will be building blocks of the space $M'$.

Denote by $M^0(\mu)$ the subspace of the evaluation module $M(\mu)$ of the twisted Yangian $Y^\pm(\mfgl_{2n})$ consisting of vectors annihilated by the operator $C(u)$ of the matrix $S(u)$, namely
\[
M^0(\mu) := \{ \,\zeta \in M(\mu) \,:\, \mr{c}_{ij}(u)\,\zeta = 0 \;\text{ for }\; 1\le i,j \le n \}.  
\]
The subspace $M^0(\mu)$ corresponds to the natural embedding $\mfgl_n\subset \mfg_{2n}$ with $\mfg_{2n}= \mfso_{2n}$ or $\mfsp_{2n}$ (generated by $F_{ij}$ with $1\le i,j \le n$, \emph{viz.}~(\ref{[F,F]}-\ref{F+F=0})) and is an irreducible $\mfgl_n$-representation of the highest weight $\mu=(\mu_1,\ldots,\mu_n)$.
The space $M^0(\mu)$ is stable under the action of the operator $A(u)$ of the matrix $S(u)$. Moreover, $A(u)$ satisfies the usual RTT relation on this space. Indeed, applying equality \eqref{CA} to $M^0(\mu)$ yields $C_1(v)\,A_2(u)\, M^0(\mu) = 0$. Applying \eqref{AA} instead we obtain 
\[
R(u-v)\,A_1(u)\,A_2(v)\,\zeta = A_2(v)\,A_1(u)\,R(u-v)\,\zeta
\]
for all $\zeta \in M^0(\mu)$. We thus have the following.

\begin{lemma} \label{L:M0}
The mapping
\[
Y(\mfgl_n) \to Y^\pm(\mfgl_{2n}), \qu T(u) \mapsto A(u) 
\]
equips the space $M^0(\mu)$ with a structure of a lowest weight $Y(\mfgl_n)$-module with the lowest weight given by~\eqref{mu(u)}. \qed
\end{lemma}

Note that the operator $A(u)$ of the matrix $S(u)$ acts on the space $M^0(\mu)$ via the Lax operator
\equ{
\mc{L}^{\pm,0}(u) := \sum_{i,j=1}^{n} e_{ij} \ot (\del_{ij} - F_{ji}(u+(\rho\pm1)/2)^{-1}) , \label{blax0}
}
which is the restriction of $\mc{L}^\pm(u)$ defined in \eqref{blax} to the operator $A(u)$.

Next, we denote by $L^0(\la^{(k)})_{c_k}$ the subspace of the evaluation module $L(\la^{(k)})_{c_k}$ of $Y(\mfgl_{2n})$ consisting of vectors annihilated by the operator $\ol{C}(u)$ of the matrix $T(u)$, namely
\equ{
L^0(\la^{(k)})_{c_k} := \{ \,\zeta \in L(\la^{(k)})_{c_k} \,:\,\, \ol{\mr{c}}_{ij}(u)\,\zeta = 0 \;\text{ for }\; 1\le i,j \le n\, \}.  \label{L0}
}
The subspace $L^0(\la^{(k)})_{c_k}$ corresponds to the natural embedding $\mfgl_n \op \mfgl_n\subset \mfgl_{2n}$ (generated by $E_{ij}$ with $1\le i,j \le n$ and $n< i,j \le 2n$) and is isomorphic to a tensor product of irreducible $\mfgl_n$-representations $L(\la^{\prime\,(k)})\ot L(\la^{\prime\prime\,(k)})$ with the highest weights $\la^{\prime\,(k)}=(\la^{(k)}_1,\ldots,\la^{(k)}_n)$ and $\la^{\prime\prime\,(k)}=(\la^{(k)}_{n+1},\ldots,\la^{(k)}_{2n})$.
Indeed, applying equality \eqref{Y:CA} to $L^0(\la^{(k)})_{c_k}$ yields $\ol{C}_1(u)\,\ol{A}_2(v)\, L^0(\la^{(k)})_{c_k} = 0$. Applying \eqref{Y:CD} instead we obtain $\ol{C}_1(u)\,\ol{D}_2(v)\, L^0(\la^{(k)})_{c_k} = 0$. Moreover, applying \eqref{Y:AD} to $L^0(\la^{(i)})_{c_i}$ we get $[ \ol{D}_1(u), \ol{A}_2(v) ]\,L^0(\la^{(k)})_{c_k} = 0$. This, together with \eqref{Y:AA} and \eqref{Y:DD}, implies the following.

\begin{lemma} \label{L:L0}
Each of the mappings 
\[
Y(\mfgl_n) \to Y(\mfgl_{2n}), \qu T(u) \mapsto \ol{A}(u)\qu\text{and}\qu T(u) \mapsto \ol{D}(u) 
\]
is a homomorphism of algebras. Moreover, these mappings equip the spaces $L(\la^{\prime\,(k)})$ and $L(\la^{\prime\prime\,(k)})$ with a structure of a lowest weight $Y(\mfgl_n)$-module with the lowest weight given by $\la^{\prime\,(k)}_k(u) = \la^{(k)}_k(u)$ and $\la^{\prime\prime\,(k)}_k(u) = \la^{(k)}_{n+k}(u)$, respectively, for $1\le k \le n$. \qed
\end{lemma}

Denote the corresponding $Y(\mfgl_n)$-modules by $L^0(\la^{\prime\,(k)})_{c_k}$ and $L^0(\la^{\prime\prime\,(k)})_{c_k}$, respectively.
The operator $\ol{A}(u)$ of the matrix $T(u)$ of $Y(\mfgl_{2n})$ acts on the space $L^0(\la^{\prime\prime\,(k)})_{c_k}$ as the identity operator, and on the space $L^0(\la^{\prime\,(k)})_{c_k}$ via the restriction of the Lax operator \eqref{lax},
\equ{
\mc{L}^{0}(u-c_k) := \sum_{i,j=1}^{n} e_{ij} \ot (\del_{ij} - E_{ji}(u-c_k)^{-1}) . \label{lax0}
}
Likewise, the operator $\ol{D}{}^t(u)$ of the transposed matrix $T^t(u)$ acts on the space $L^0(\la^{\prime\,(k)})_{c_k}$ as the identity operator, and on the space $L^0(\la^{\prime\prime\,(k)})_{c_k}$ via the transposed Lax operator $(\mc{L}^{0}(u-c_k))^t$.

We are now ready to introduce the \emph{vacuum sector} $M^0\subset M$ by
\eqa{
M^0 &:= L^0(\la^{(1)})_{c_1} \ot \cdots \ot L^0(\la^{(\ell)})_{c_\ell} \ot M^0(\mu) . \label{M0}
} 

\begin{lemma} \label{L:M0-stable}
The space $M^0$ is stable under the action of the operator $A(u)$ of the matrix $S(u)$.
\end{lemma}

\begin{proof}
We start by showing that operator $C(u)$ of the matrix $S(u)$ annihilates the space $M^0$: $\mr{c}_{ij}(v)\cdot M^0 = 0$. We use induction on $\ell$. For $\ell=0$ we have $M^0 = M^0(\mu)$ and $\mr{c}_{ij}(v)\, M^0(\mu) = 0$, by definition \eqref{M0}. For any $\ell\ge 1$ denote $M^{(\ell-1)} := L^0(\la^{(1)})_{c_1} \ot \cdots \ot L^0(\la^{(\ell-1)})_{c_{\ell-1}} \ot M^0(\mu)$. Let $\zeta \in M^{(\ell-1)}$ and $\zeta' \in L^0(\la^{(\ell)})_{c_{\ell}}$ be any nonzero vectors. Using \eqref{cop-s} and the notation \eqref{block} we find 
\aln{
\mr{c}_{ij}(u)\cdot (\zeta'\ot\zeta) &= \sum_{k,l=1}^{n} \Big( \;\ol{\mr{c}}_{ik}(u)\,\ol{\mr{d}}_{\bar\jmath\,\bar l}(-u-\rho) \, \zeta' \ot \mr{a}_{kl}(u)\cdot\zeta \pm \ol{\mr{c}}_{ik}(u)\,\ol{\mr{c}}_{\bar\jmath\,\bar l}(-u-\rho) \, \zeta' \ot \mr{b}_{kl}(u)\cdot\zeta \\[-.5em]
& \qq\qq + \ol{\mr{d}}_{ik}(u)\,\ol{\mr{d}}_{\bar\jmath\,\bar l}(-u-\rho) \, \zeta' \ot \mr{c}_{kl}(u)\cdot\zeta \pm \ol{\mr{d}}_{ik}(u)\,\ol{\mr{c}}_{\bar\jmath\,\bar l}(-u-\rho) \, \zeta' \ot \mr{d}_{kl}(u)\cdot\zeta\;\Big) 
\\
& = \sum_{k,l=1}^{n} \Big( \;\ol{\mr{c}}_{ik}(u)\,\ol{\mr{d}}_{\bar\jmath\,\bar l}(-u-\rho) \, \zeta' \ot \mr{a}_{kl}(u)\cdot\zeta + \ol{\mr{d}}_{ik}(u)\,\ol{\mr{d}}_{\bar\jmath\,\bar l}(-u-\rho) \, \zeta' \ot \mr{c}_{kl}(u)\cdot\zeta\;\Big) ,
}
by definition \eqref{L0}; here we used the notation $\bar\imath = n-i+1$. Assuming induction, $\mr{c}_{kl}(u)\,\zeta = 0$. Finally, by \eqref{Y:CD} and \eqref{L0}, we have that
\[
\ol{\mr{c}}_{ik}(u)\,\ol{\mr{d}}_{\bar\jmath\,\bar l}(-u-\rho) \, \zeta' = \ol{\mr{d}}_{\bar\jmath\,\bar l}(-u-\rho)\,\ol{\mr{c}}_{ik}(u)\, \zeta' = 0.
\]
Hence $\mr{c}_{ij}(u)\cdot (\zeta'\ot\zeta) = 0$, as required. Next, we need to show that $\mr{a}_{ij}(u)\cdot M^0 \subseteq M^0 [u^{-1}]$. The base case is given by Lemma \ref{L:M0}. For $\ell\ge1$ we have
\aln{
\mr{a}_{ij}(u)\cdot (\zeta'\ot\zeta) &= \sum_{k,l=1}^{n} \Big( \; \ol{\mr{a}}_{ik}(u)\,\ol{\mr{d}}_{\bar\jmath\,\bar l}(-u-\rho) \, \zeta' \ot \mr{a}_{kl}(u)\cdot\zeta  \pm \ol{\mr{a}}_{ik}(u)\,\ol{\mr{c}}_{\bar\jmath\,\bar l}(-u-\rho) \, \zeta' \ot \mr{b}_{kl}(u)\cdot\zeta \\[-.5em]
& \qq\qq + \ol{\mr{b}}_{ik}(u)\,\ol{\mr{d}}_{\bar\jmath\,\bar l}(-u-\rho) \, \zeta' \ot \mr{c}_{kl}(u)\cdot\zeta \pm \ol{\mr{b}}_{ik}(u)\,\ol{\mr{c}}_{\bar\jmath\,\bar l}(-u-\rho) \, \zeta' \ot \mr{d}_{kl}(u)\cdot\zeta \;\Big)
\\
& = \sum_{k,l=1}^{n} \ol{\mr{a}}_{ik}(u)\,\ol{\mr{d}}_{\bar\jmath\,\bar l}(-u-\rho) \, \zeta' \ot \mr{a}_{kl}(u)\cdot\zeta ,
}
by definition \eqref{L0} and the result above. Assuming induction, $\mr{a}_{kl}(u)\cdot \zeta \in M^{(\ell-1)}[u^{-1}]$ and, by Lemma \ref{L:L0}, $ \ol{\mr{a}}_{ik}(u)\,\ol{\mr{d}}_{\bar\jmath\,\bar l}(-u-\rho) \, \zeta' \in L^0(\la^{(\ell)})_{c_{\ell}}[u^{-1}]$. Hence $\mr{a}_{ij}(u)\cdot (\zeta'\ot\zeta) \in M^0 [u^{-1}]$. This proves the claim.
\end{proof}

The last ingredients we will need are the auxiliary spaces $V_{\at_i}$ and $V_{a_i}$. They are vector representations of $\mfgl_{n}$ of weight $\la^{(\at_i)}=\la^{(a_i)}=(1,0,\dots,0)$. Denote by $L^t(\la)_c$ the evaluation module of $Y(\mfgl_n)$ obtained from the $\mfgl_n$-representation $L(\la)$ by composing the evaluation map $ev_c$ in \eqref{ev-hom} with the algebra automorphism $T(u)\to T^t(-u)$. The spaces $V_{\at_i}$ and $V_{a_i}$ can thus be viewed as evaluation modules $L^t(\la^{(\at_i)})_{-u_i}$ and $L^t(\la^{(a_i)})_{u_i}$ of $Y(\mfgl_{n})$, respectively, with the lowest weights given by
\eqg{
\la^{(\at_i)}_j(u) = \la^{(a_i)}_j(u) = 1 \qu\text{for}\qu 1\le j \le n-1 \qu\text{and}\\
\la^{(\at_i)}_n(v)=\frac{v-u_i+1}{v-u_i} , \qu \la^{(a_i)}_n(v)=\frac{v+u_i+1}{v+u_i}. \label{la(u):aux}
}
In particular, the matrix $T_a(v)$ of $Y(\mfgl_{n})$ acts on the space $L^t(\la^{(\at_i)})_{-u_i}$ as $R^t_{a\at_i}(u_i-v)$ and on the space $L^t(\la^{(a_i)})_{u_i}$ as $R^t_{a a_i}(-u_i-v)$; here note that $R^t_{ab}(u)=R^t_{ba}(u)$.

We define the \emph{nested vacuum sector} as a tensor product the auxiliary spaces and the vacuum sector $M^0$:
\equ{ 
M' := W \ot  M^0 , \qq W = V_{\at_1} \ot \cdots \ot V_{\at_m} \ot V_{a_1} \ot \cdots \ot V_{a_m}. \label{M'}
}

\begin{prop} \label{P:ResTT} 
Let $T(v)$ be the generating matrix of $Y(\mfgl_n)$. Then the mapping
\[
Y(\mfgl_n) \to \End(W) \ot Y^\pm(\mfgl_{2n}) , \qq T(v) \mapsto T(v;\bm u)
\]
equips the space $M'$ with the structure of a lowest weight $Y(\mfgl_n)$-module with the lowest weight given by
\equ{
\la_i(v;{\bm u}) = \la_i(u)\,\la_{2n-i+1}(-u)\,\mu_i(u) \prod_{j=1}^m \la^{(a_j)}_i(v)\, \la^{(\at_j)}_i(v) \label{la(v;u)}
}
for $1\le i \le n$ with $\la_i(v)$ defined in \eqref{L:la(u)}, $\mu_i(v)$ in \eqref{mu(u)} and $\la^{(a_j)}_i(v)$, $\la^{(\at_j)}_i(v)$ in \eqref{la(u):aux}.
\end{prop}

\begin{proof} 
It follows from the definition \eqref{nMM} and Lemma \ref{L:M0-stable}, that the space $M'$ is stable under the action of $T_a(v;{\bm u})$. Moreover, for any $\zeta\in M'$, we have that
\[ 
R_{ab}(v-w)\,T_a(v;{\bm u})\,T_b(w;{\bm u})\cdot\zeta = T_b(w;{\bm u})\,T_a(v;{\bm u})\,R_{ab}(v-w)\cdot\zeta . 
\]
Indeed, we can interleave the matrices on the l.h.s.~of the equality above, then use the transposed Yang-Baxter equation to reorder the product of matrices:
\aln{
&R_{ab}(v-w)\,T_a(v;{\bm u})\,T_b(w;{\bm u}) \\
&\qq = R_{ab}(v-w)\Bigg( \prod_{k=1}^m R^t_{\at_k a}(u_k-v) R^t_{\at_k b}(u_k-w) \Bigg) \Bigg( \prod_{l=1}^m R^t_{a_l a}(-u_l-v) R^t_{a_l b}(-u_l-w) \Bigg) A_a(v) A_b(w) \\
&\qq= \Bigg( \prod_{k=1}^m  R^t_{\at_k b}(u_k-w) R^t_{\at_k a}(u_k-v) \Bigg) \Bigg( \prod_{l=1}^m  R^t_{a_l b}(-u_l-w) R^t_{a_l a}(-u_l-v) \Bigg) R_{ab}(v-w) A_a(v) A_b(w).
}
From here we use \eqref{AA} to obtain the result, plus additional terms. However, $C(u)$ appears as the rightmost operator acting nontrivially on $M^0\subset M'$ in each of these additional terms. Since $C(u)$ annihilates all vectors in $M^0$, these additional terms vanish.
Its lowest vector~is
\equ{
\eta := e_{\at_1} \ot \cdots \ot e_{\at_m} \ot e_{a_1} \ot \cdots \ot e_{a_m} \ot \eta_1 \ot \cdots \ot \eta_\ell \ot \xi, \label{M'-low}
}
where $\xi$ is a lowest vector of $M^0(\mu)$, each $\eta_i$ is a lowest vector of $L(\la^{(i)})_{c_i}^0$ for $1\le i \le \ell$, and  each $e_{\at_i}$ (resp.~$e_{a_i}$) is a lowest vector of $V_{\at_i}$ (resp.~$V_{a_i}$) for $1\le i \le m$ (viewed as an evaluation module $L^t(\la^{(\at_i)})_{-u_i}$ (resp.~$L^t(\la^{(a_i)})_{-u_i}$)). Finally, acting with $t_{ii}(v;\bm u)$ on $\eta$ for $1\le i \le n$ and using \eqref{L:la(u)}, \eqref{mu(u)} and \eqref{la(u):aux} yields \eqref{la(v;u)}. 
\end{proof}

\begin{rmk}
\emph{(i)} 
Recall that $L^0(\la^{(i)})_{c_i} \cong L^0(\la^{\prime\,(i)})_{c_i} \ot L^0(\la^{\prime\prime\,(i)})_{c_i}$ with $\ol{A}(u)$ (resp.~$\ol{D}{}^t(u)$) acting non-trivially on the first (resp.~second) tensorand only. We may thus rewrite the space $M^0$ as
\[
M^0 \cong L^0(\la^{\prime\,(1)})_{c_1} \ot \cdots \ot L^0(\la^{\prime\,(\ell)})_{c_\ell} \ot M^0(\mu) \ot L^0(\la^{\prime\prime\,(\ell)})_{c_\ell} \ot \cdots \ot L^0(\la^{\prime\prime\,(1)})_{c_1} . 
\]
By Proposition \ref{P:ResTT}, we may view this space as a lowest weight $Y(\mfgl_n)$-module. Provided the binary property holds, it is an irreducible $Y(\mfgl_n)$-module, see Theorem 6.5.8 in \cite{Mo3}.  
{\it(ii)}~
Enumerate the tensorands of $M^0$ above by $1,2,\dots,2\ell,2\ell+1$. Then the matrix $T_a(v;\bm u)$ acts on the space $M' = W \ot M^0$ via the operator
\aln{
& \Bigg( \prod_{k=1}^{m} R^t_{a\at_k}(u_k-v) \Bigg)\Bigg( \prod_{l=1}^{m} R^t_{aa_l}(-u_l-v-\rho) \Bigg) \\ & \qq \times\Bigg( \prod_{i=1}^\ell \mc{L}^0_{ai}(u-c_i) \Bigg)  \mc{L}^{\pm0}_{a,\ell+1}(\mu)\Bigg(\prod_{i=\ell}^1 \big(\mc{L}^0_{a,2\ell-i+1}(-u-\rho-c_i)\big)^t \Bigg),
}
where the Lax operators are those defined in \eqref{lax0} and \eqref{blax0}.
\end{rmk}

We end this section with one more technical lemma which will assist us in finding the explicit expressions of the unwanted terms in Section \ref{sec:UWT}.

\begin{lemma} \label{L:RRt=tRR}
The following identities hold:
\gan{
\check{R}_{a_i a_{i+1}}(u_i-u_{i+1}) \check{R}^{-1}_{\at_i \at_{i+1}}(u_i-u_{i+1})\, t_{kl}(w;{\bm u}) = t_{kl}(w;\bm u_{i \leftrightarrow i+1}) \check{R}_{a_i a_{i+1}}(u_i-u_{i+1}) \check{R}^{-1}_{\at_i \at_{i+1}}(u_i-u_{i+1}).
\\
\check{R}_{a_i a_{i+1}}(u_i-u_{i+1}) \check{R}^{-1}_{\at_i \at_{i+1}}(u_i-u_{i+1}) \, \eta = \eta.
}
\end{lemma}

\begin{proof}
The first identity is achieved by moving the $\check{R}$-matrices through each matrix in the definition of the nested monodromy matrix. Indeed, the $\check{R}$-matrices each commute with all but a pair of adjacent $R$-matrices in \eqref{nMM}, for which we use the Yang Baxter equations,
\gan{
\check{R}_{a_i a_{i+1}}(u_i-u_{i+1}) R^t_{a_i a}(-u_i-v) R^t_{a_{i+1} a}(-u_{i+1}-v) = R^t_{a_i a}(-u_{i+1}-v) R^t_{a_{i+1} a}(-u_i-v) \check{R}_{a_i a_{i+1}}(u_i-u_{i+1}) , \\
\check{R}^{-1}_{\at_i \at_{i+1}}(u_i-u_{i+1}) R^t_{\at_i a}(u_i-v) R^t_{\at_{i+1} a}(u_{i+1}-v) = R^t_{\at_i a}(u_{i+1}-v) R^t_{\at_{i+1} a}(u_i-v) \check{R}^{-1}_{\at_i \at_{i+1}}(u_i-u_{i+1}),
}
and the result follows. 

To see why the second identity is true, notice that the lowest weight vector $\eta$ \eqref{M'-low} is an eigenvector of $P_{a_i a_{i+1}}$, and therefore also of $\check{R}_{a_i a_{i+1}}(u_i-u_{i+1})$. This is true also for $P_{\at_i \at_{i+1}}$. 
Thus, acting with both $\check{R}_{a_i a_{i+1}}(u_i-u_{i+1})$ and $\check{R}^{-1}_{\at_i \at_{i+1}}(u_i-u_{i+1})$, the eigenvalues cancel, which gives the result. 
\end{proof}


\section{Nested algebraic Bethe ansatz for \texorpdfstring{$Y^\pm(\mfgl_{2n})$}{}} \label{sec:NABA}

We are now ready to consider the nested algebraic Bethe ansatz for a one-dimensional spin chain with open boundary conditions and having twisted Yangian $Y^\pm(\mfgl_{2n})$ as its underlying symmetry. The full quantum space is the lowest weight $Y^\pm(\mfgl_{2n})$-module $M$ defined in \eqref{M}:
\[
M = L(\la^{(1)})_{c_1} \ot L(\la^{(2)})_{c_2} \ot \cdots \ot L(\la^{(\ell)})_{c_\ell} \ot M(\mu) .
\]
The generating matrix $S(u)$ of $Y^\pm(\mfgl_{2n})$ acts on this space via a product of Lax operators \eqref{S(u).M}:
\[
S_a(v)\cdot M = \Bigg( \prod_{i=1}^\ell \mc{L}_{ai}(v-c_i) \Bigg)\,\mc{L}^\pm_{a,\ell+1}(v)\,\Bigg( \prod_{i=\ell}^1 \mc{L}^t_{ai}(-v-\rho-c_i)\Bigg)\,M .
\]
Taking the trace of the generating matrix we obtain a double-row transfer matrix 
\equ{
\tau(v) := \tr\, S(v) = \tr A(v) + \tr D(v) = \tr A(v) + \tr D^t(v).  \label{TM}
}
One can show using the usual methods that $[\tau(u),\tau(v)]=0$; see Section 2 in \cite{ACDFR1}, also \cite{Sk}.
We seek an eigenvector of $\Psi\in M$ of $\tau(v)$, which we will refer to as the \emph{Bethe vector}. The problem of finding an eigenvector of the transfer matrix \eqref{TM} can be substantially simplified with the help of the symmetry relation \eqref{symmD} which allows us to write the transfer matrix $\tau(v)$ in a symmetric form 
\[
\tau(v)=p(v)\tr A(v) + p(-v-\rho)\tr A(-v-\rho) = \{p(v)\tr A(v)\}^v ,
\]
where $p(v)$ is given by \eqref{p(v)}. Here we used the notation introduced in \eqref{v-even}. It will therefore be sufficient to focus on the action of $A(v)$, without needing to consider $D(v)$.

The last ingredient we will need is the \emph{nested transfer matrix}, see \eqref{nMM}:
\[
t(v;\bm u) := \tr T(v;\bm u) . 
\]
It will play the role of $\tau(v)$ at the nested level of the ansatz. Since we will only consider the action of $T(v;\bm u)$ on a finite-dimensional vector space, we can thus specialize the parameters $u_i$ of $m$-tuple $\bm u$ to nonzero complex numbers. Hence we will further assume that $\bm u\in \C^m$ is an $m$-tuple of distinct nonzero complex numbers.

\begin{rmk}
Recall from Section \ref{sec:Y} that $K\in\End(\C^{2n})$ satisfying $K^t = \veps K$ with $\veps = +1$ or $\veps=-1$ is a matrix solution to the reflection equation. 
The dual reflection equation is obtained by redefining $u\to -u-\rho$ and $v\to -v - \rho$ and therefore has the same set of solutions, denoted by $\wt K\in\End(\C^{2n})$, satisfying $\wt K^t = \tl \veps \wt K$ with $\tl\veps = +1$ or $\tl\veps=-1$.
The transfer matrix for an open spin chain with generic (invertible) boundary conditions is then given by $\tau^{\wt K,K}(v) := \tr (\wt K S^K(v))$, where $S^K(v)$ is defined by \eqref{SK=TKT} and both $K$, $\wt K$ are assumed to be invertible. 
Write $K=A I^\veps A^t$ for some invertible $A\in\End(\C^{2n})$. By Proposition \ref{P:TY-iso}, we have that $\psi_K^{-1}(\tau^{\wt K,K}(v)) = \tr ( G S(v)) =: \tau^{G,I}(v)$ with $G = I^\veps (A^t)^{-1} \wt K A^{-1}$. In other words, the spectrum of an open spin chain with generic boundary conditions coincides, up to an isomorphism, with the spectrum of a spin chain with generic left and trivial right boundary conditions. The nesting method presented in this paper can be applied to the cases when $G$ is a block-diagonal matrix, viz.~\eqref{block}. Denoting $\wt S(v) := S(v) G$ we can then apply the decomposition \eqref{TM}. Moreover, the exchange relations for the blocks of $\wt S(v)$ remain unchanged, thus allowing us to reduce the diagonalization problem of $\tau^{G,I}(v)$ to that of $t^G(v;\bm u) := \tr(T(v;\bm u)[G]_A)$, a transfer matrix for a $\mfgl_n$-symmetric spin chain with periodic twisted boundary conditions; here $[G]_A$ denotes the upper-left block of $G$. 
The residual diagonalisation problem for $t^G(v;\bm u)$ may then be solved by the nested algebraic Bethe ansatz only if $[G]_A$ is a diagonal matrix. For a non-diagonal $[G]_A$ some other methods, such as separation of variables or functional relations, need to be used; for the six-vertex models see e.g., \cite{BBRY,Ga}, for the higher rank cases see \cite{WYCS}.
\end{rmk}

\begin{rmk}
The open spin chains of this type (with ``soliton non-preserving'' open boundary conditions) were first considered by Doikou in \cite{Do}. Let $\ell\in2\N$. Then, for each $i\in\{1\dots,\ell/2\}$ let $L(\la^{(2i-1)})_{c_{2i-1}}$ (resp.~$L(\la^{(2i)})_{c_{2i}}$) be the fundamental (resp.~the anti-fundamental) representation of $\mfgl_{2n}$. Fix the boundary representation $M(\mu)$ to be a one-dimensional representation of $\mfso_{2n}$ or $\mfsp_{2n}$. Doikou showed that such a spin chain can be equipped with a local Hamiltonian having interactions up to four nearest neighbours. 
\end{rmk}


\subsection{Bethe vector for a single excitation}

To introduce the nesting technique, we start by constructing the Bethe vector with a single excitation, i.e., $m=1$, as this case allows us to expose the main idea of our approach while keeping the technical difficulties to the minimum; for example, in this case we find the unwanted terms without additional computations. Recall the definition of the vacuum sector $M^0$ \eqref{M0} and the nested vacuum sector~$M'$~\eqref{M'}. For $m=1$ we have $M' = V_{\at_1} \ot V_{a_1} \ot M^0$. Let $\Phi \in M'$, which we will refer to as the \emph{nested Bethe vector}. The vector $\Phi$ may depend on $u\in\C$, hence we will write $\Phi=\Phi(u)$. Using the creation operators defined in Definition~\ref{D:beta}, we write an ansatz for the Bethe vector with a single excitation
\equ{
\Psi(u) := \beta_{\tl a_1 a_1}(u)\cdot\Phi(u) \in M. \label{Psi(u):1}
}
We now compute the action of the transfer matrix $\tau(v)$ on this Bethe vector. Using Lemma \ref{L:AB} we have 
\ali{
\tau(v)\cdot \Psi(u) & = \{p(v) \tr A(v)\}^v \,\beta_{\tl a_1 a_1}(u)\cdot \Phi(u) \el
& = \bee{1}{u} \tr_a \Big(\even{v}{ p(v)\, R^t_{\at_1a}(u-v)\, R^t_{a_1a}(-u-v-\rho)\, A_a(v)}\Big) \cdot \Phi(u) \el 
&\qu + \frac1{p(u)} \even{v}{ \frac{p(v)}{u-v}\, \bee{1}{v}} \Res{w \rightarrow u} \tr_a \Big( \even{w}{ p(w) R^t_{\at_1a}(u-w)\, R^t_{a_1a}(-u-w-\rho) A_a(w) } \Big) \cdot \Phi(u) \el 
& = \bee{1}{u} \,\even{v}{p(v)\,t(v;u)} \cdot\Phi(u) \el & \qu + \frac1{p(u)} \even{v}{ \frac{p(v)}{u-v} \,\bee{1}{v}} \Res{w \rightarrow u} \even{w}{p(w)\,t(w;u)} \cdot \Phi(u) \label{tauPhi1}.
} 
The first term in the r.h.s.~of the equality above is the wanted term, as the parameter carried by $\bee{1}{u}$ is unchanged. The second term has $\bee{1}{v}$ and is the unwanted term, which we will require to vanish.

 Let us now make the additional requirement, which we will justify later, that vector $\Phi(u)$ is an eigenvector of the nested transfer matrix $t(v; u)$ with an eigenvalue $\Gamma(v;u)$:
\equ{ \label{gamma}
t(v;u)\cdot\Phi(u) = \Gamma(v;u)\,\Phi(u) .
}
This allows us to rewrite \eqref{tauPhi1} as
\aln{
\tau(v)\cdot\Psi(u) 
&= \even{v}{p(v) \,\Gamma(v;u)} \Psi(u) + \frac1{p(u)} \Res{w \rightarrow u} \even{w}{p(w)\,\Gamma(w;u)} \even{v}{ \frac{p(v)}{u-v} \,\bee{1}{v}} \!\cdot \Phi(u) \\
&= \Lambda(v;u) \, \Psi(u) + \Res{w \rightarrow u} \Lambda(w;u) \,\frac1{p(u)} \even{v}{ \frac{p(v)}{u-v} \,\bee{1}{v}} \!\cdot \Phi(u),
}
where $\Lambda(v;u) := \even{v}{p(v) \, \Gamma(v;u)}$.
We thus conclude that $\Phi(u)$ is an eigenvector of $\tau(v)$ with eigenvalue $\Lambda(v;u)$ if 
\[
\Res{w \rightarrow u} \Lambda(w;u) = 0.
\]
This is the \emph{Bethe equation} for $u$, solutions of which, by \eqref{Psi(u):1}, give a set of possible eigenvectors of $\tau(v)$.

It remains to find a nested Bethe vector $\Phi(u)$ satisfying \eqref{gamma}: we seek an eigenvector $\Phi(u)\in M'$ of $t(v;u)$. By Proposition \ref{P:ResTT}, the nested monodromy matrix $T_a(v;u)$ and the nested vacuum sector $M'$ form a $Y(\mfgl_n)$-system. The spectral problem of this system can be solved by means of the usual nested algebraic Bethe ansatz presented in \cite{KuRs}, which we have recalled in detail in Appendix \ref{app:gln}.
For example, the ansatz for $\Phi(u)$ has the form
\[
\Phi(u) = \Phi(\bm u';u) := B'_{a'_1}(u'_{1})\cdots B'_{a'_{m'}}(u'_{m'})\cdot \Phi'(\bm u';u),
 \]
where $\bm u' = (u'_1,\ldots,u'_{m'})\in\C^{m'} $ and $B'_{a'_j}(u'_{j})$ are creation operators taken from the $T_a(v;u)$. Continuing this nesting procedure, we obtain an eigenvector $\Phi(u;\bm u')$ with eigenvalue, see \eqref{gln:full_eig},
\eqn{
\Gamma(v;u) &= \lambda_1(v;u)\prod_{i=1}^{m'} \frac{v-u_i^{(1)}+1}{v-u_i^{(1)}} + \lambda_n(v;u)\prod_{i=1}^{m^{(n-1)}} \frac{v-u^{(n-1)}_i-1}{v-u^{(n-1)}_i} 
\\
& \qu + \sum_{k=2}^{n-1} \lambda_k(v;u) \prod_{i=1}^{m^{(k-1)}} \frac{v-u^{(k-1)}_i-1}{v-u^{(k-1)}_i} \prod_{j=1}^{m^{(k)}} \frac{v-u^{(k)}_j+1}{v-u^{(k)}_j} ,
}
where $\la_k(v;u)$ are given by \eqref{P:ResTT} and the $u_i^{(k)}$ with $1 \leq k \leq n-1$ are parameters introduced at level $k$ of nesting when diagonalizing the $\mfgl_n$-symmetric periodic spin chain. These parameters are fixed to be solutions of their respective Bethe equations, given in \eqref{gln:BE}. 

The boundary eigenvalue $\Lambda(v;u)$ and Bethe equation for $u$ can then be found by substituting the values for $\lambda_k(v;u)$ from \eqref{la(v;u)} into the above expression. These are given explicitly for multiple excitations by Theorem~\ref{T:eig} in Section \ref{sec:res}.


\subsection{Bethe vector for multiple excitations}

For multiple excitations the argument proceeds similarly to the previous section. Recall that $m\in\N$ is the excitation number and $\bm u \in \C^m$ is an $m$-tuple of distinct nonzero complex parameters. Let $\Phi$, the \emph{nested Bethe vector}, be a vector from the lowest weight $Y(\mfgl_n)$-module $M'$ defined in \eqref{M'}: 
\[
\Phi \in M' = V_{\at_1} \ot \cdots \ot V_{\at_m} \ot V_{a_1} \ot \cdots \ot V_{a_m} \ot M^0 .
\]
The vector $\Phi$ may also depend on the parameters $\bm u$, and we will write $\Phi = \Phi(\bm u)$. From the nested Bethe vector, we make the following ansatz for the full Bethe vector:
\eqa{ \label{bethevec}
\Psi(\bm u) := \bethe{m}{\bm u}\cdot \Phi(\bm u) \in M. 
}
We now act with the transfer matrix $\tau(v)$ on this Bethe vector. Using Corollary \ref{C:wantedterm} we find 
\aln{
\tau(v)\cdot\Psi(\bm u) = \bethe{m}{\bm u}\even{v}{p(v)\,t(v;\bm u)}\cdot \Phi(\bm u) + UWT .
}
The unwanted terms $UWT$ are less simple than in the $m=1$ case, and will be discussed in detail in the section below. With the expectation that the $\bm u$ may be chosen such that the unwanted terms vanish, the Bethe vector $\Phi(\bm u)$ will be an eigenvector of $\tau(v)$ if we take the additional requirement, as for $m=1$, that
\eqa{ \label{nested}
t(v;\bm u)\cdot\Phi(\bm u) = \Gamma(v;\bm u)\,\Phi(\bm u) .
}
We therefore seek an eigenvector $\Phi(\bm u)\in M'$ of $t(v;\bm u)$. This is found again by the algebraic Bethe ansatz for $Y(\mfgl_n)$ with the full quantum space $M'$ and monodromy matrix $T(v,\bm u)$. 

From here, proceeding as we did in the $m=1$ case, we have that
\equ{ \label{eig_full}
\tau(v)\cdot \Psi(\bm u) = \Lambda(v;\bm u)\,\Psi(\bm u) + UWT , \qu\text{where}\qu \Lambda(v;\bm u) = \even{v}{p(v)\, \Gamma(v;\bm u)}.
}


\subsection{Dealing with unwanted terms} \label{sec:UWT}

In this section, we will find an exact expression for the unwanted terms from the action of $\tau(v)$ on the Bethe vector and, by setting these terms to zero, we will obtain the Bethe equations. 

We begin by introducing some notation for the unwanted terms. Let $U(v;\bm u)$ denote the terms initially excluded from the expression in 
\[
\tau(v) \, \bethe{m}{\bm u} = \bethe{m}{\bm u} \, \even{v}{p(v)\,t(v;\bm u)} + U(v;\bm u).
\]
To find an expression for $U(v;\bm u)$, begin by acting on $\bethe{m}{\bm u}$. By repeated applications of \eqref{ABetafull}, as in Lemma~\ref{L:movement}, we may move $A_a(\cdot)$ through each of the remaining creation operators in $\bethe{m}{\bm u}$, generating a sum of terms in which the rightmost operator is a matrix element of $A_a(u)$ for $u \in \{v,u_1,\dots, u_m, -v-\rho,-u_1-\rho, \dots, -u_m-\rho \}$.
To find a more concise expression for $U(v;\bm u)$, it will be useful to partition the terms by the parameter appearing in $A_a(\cdot)$. Let $\mc B$ denote the subalgebra of $Y^{\pm}(\mfgl_{2n})$ generated by coefficients of $b_{ij}(u)$ for $1 \leq i,j \leq n$, the closure of which is guaranteed by \eqref{BB}. Then
\[
U(v;\bm u) = \sum_{j=1}^m \Big( U_{+,j}(v;\bm u) + U_{-,j}(v;\bm u) \Big) ,
\]
where
\[
U_{+,j}(v;\bm u) = \sum_{k,l=1}^n B_{+,j,kl} \, a_{kl}(u_j), \qq
U_{-,j}(v;\bm u) = \sum_{k,l=1}^n B_{-,j,kl} \, a_{kl}(-u_j-\rho) 
\]
for some $B_{\pm,j,kl} \in \mc B \ot (\C^n)^{\ot 2m}$. Additionally, let us define $U_j(v;\bm u):= U_{+,j}(v;\bm u) + U_{-,j}(v;\bm u) $. We will now proceed to find $U_1(v;\bm u)$ using the standard techniques. Indeed, consider moving $\tau(v)$ through only the first creation operator. From \eqref{ABetafull},
\aln{ 
& \tau(v) \, \bethe{m}{\bm u} \\ & \qq= \bigg( \bee{1}{u_1} \tr_a \even{v}{ p(v) R^t_{\at_1a}(u_1-v) R^t_{a_1a}(-u_1-v-\rho) A_a(v)} \el
& \qq\qq + \frac1{p(u_1)} \even{v}{ \frac{p(v)}{u_1-v} \bee{1}{v}} \Res{w \rightarrow u_1} \tr_a\even{w}{ p(w) R^t_{\at_1a}(u_1-w) R^t_{a_1a}(-u_1-w-\rho) A_a(w) } \bigg) \\
&\qq\qu \times \prod_{i=2}^m \Big(\bee{i}{u_i} \prod_{j=i-1}^{1} R_{a_j \at_i}(-u_j-u_i-\rho) \Big).
}
We focus on the second term here, which, upon taking the residue, contains $A_a(u_1)$ and $A_a(-u_1-\rho)$. As all the entries of the $m$-tuple $\bm u$ are distinct, all contributions to $U_{1}(v;\bm u)$ must originate from moving $A_a(u_1)$ and $A_a(-u_1-\rho)$ through the remaining creation operators without any further parameter exchanges. Therefore, by repeated applications of Lemma~\ref{L:movement},
\[
U_{1}(v;\bm u) = \frac1{p(u_1)} \even{v}{ \frac{p(v)}{u_1-v}\, \bee{1}{v}} \prod_{i=2}^m \bigg(\bee{i}{u_i} \prod_{j=i-1}^{1} R_{a_j \at_i}(-u_j-u_i-\rho) \bigg) \Res{w \rightarrow u_1} \even{w}{ p(w) \, t(w;\bm u) } .
\]
It now remains to find similar expressions for $U_j(v;\bm u)$ for $2 \leq j \leq m$. Recall Lemma~\ref{L:B=BaReR}. By repeatedly applying such transpositions, we may apply an arbitrary permutation to the parameters $\bm u$ in the $m$-excitation creation operator. For $\sigma \in \mf S_m$, let $\bm u_{\sigma}$ denote the ordered set $(u_{\sigma(1)},u_{\sigma(2)}, \dots, u_{\sigma(m)})$. Additionally, let $\sigma_j$ denote the cyclic permutation $\sigma_j: (1,2, \dots, m) \mapsto (j, j+1, \dots, j-1)$. We have
\[
\bethe{m}{\bm u} = \bethe{m}{\bm u_{\sigma_j}} \, \check{R}_{a_1 \dots a_m}[\sigma_j](\bm u) \, \check{R}^{-1}_{\at_1 \dots \at_m}[\sigma_j](\bm u)
\]
where $\check{R}[\sigma_j](\bm u)$ is the product of $\check{R}$-matrices that generates this permutation. Using this expression for $\bethe{m}{\bm u}$ and following the argument above, we obtain an expression for $U_k(v;\bm u)$:
\aln{
U_{k}(v;\bm u) &= \frac1{p(u_k)} \even{v}{ \frac{p(v)}{u_k-v}\, \bee{1}{v}} \prod_{i=2}^m \bigg(\bee{i}{u_{\sigma_k(i)}} \prod_{j=i-1}^{1} R_{a_j \at_i}(-u_{\sigma_k(j)}-u_{\sigma_k(i)}-\rho) \bigg) 
\\
&\hspace{5cm}\times\Res{w \rightarrow u_k} \even{w}{ p(w) \, t(w;\bm u_{\sigma_k}) } \check{R}_{a_1 \dots a_m}[\sigma_k](\bm u) \, \check{R}^{-1}_{\at_1 \dots \at_m}[\sigma_k](\bm u) .
}
By applying this to the nested Bethe vector, we obtain an expression for all the unwanted terms from the action of $\tau(v)$ on $\Psi(\bm u)$. However, in order to obtain the Bethe equations, we must assume one additional property of the nested Bethe vector $\Phi$. We require $\Phi = \Phi(\bm u)$ such that:
\[
\check{R}_{a_i a_{i+1}}(u_i-u_{i+1}) \, \check{R}^{-1}_{\at_i \at_{i+1}}(u_i-u_{i+1}) \, \Phi(\bm u) = \Phi(\bm u_{i \leftrightarrow i+1}) \qu \text{for} \qu 1 \leq i \leq n-1.
\] 
By Proposition \ref{P:ResTT} and Lemma~\ref{L:RRt=tRR}, this is true for a vector generated by acting with the nested monodromy matrix on the lowest vector
\aln{
& \Phi(\bm u) \in {\rm span}_\C \big\{  t_{i_1 j_1}(w_1; \bm u) \cdots t_{i_K j_K}(w_K; \bm u) \cdot \eta \, : \, K \geq 0, \, 1 \leq i_1, j_1, \dots, i_K, j_K \leq n-1, \, \bm w \in \C^K \big\}.
}
The $\Phi(\bm u)$ constructed using the nested Bethe ansatz for $Y(\mfgl_n)$ will have this form, so we may simultaneously assume \eqref{nested}. Note that the action of $\check{R}$-matrices may be combined with \eqref{B=BaReR} in such a way that transpositions of the parameters $\bm u$ leave the full Bethe vector $\Psi(\bm u)$ unchanged. Therefore, $\Psi(\bm u) = \Psi(\bm u_{\sigma})$ for all $\sigma \in \mf S_m$.
Acting on $\Phi(\bm u)$ with the expression for $U_k(v;\bm u)$ above, and summing over $k$, we obtain
\aln{ 
U(v;\bm u) \cdot \Phi(\bm u) &= \sum_{k=1}^m \frac1{p(u_k)} \even{v}{ \frac{p(v)}{u_k-v}\, \bee{1}{v}} \prod_{i=2}^m \bigg(\bee{i}{u_{\sigma_k(i)}} \prod_{j=i-1}^{1} R_{a_j \at_i}(-u_{\sigma_k(j)}-u_{\sigma_k(i)}-\rho) \bigg) \\[.25em]
&\qu \times\Res{w \rightarrow u_k} \even{w}{ p(w) \, t(w;\bm u_{\sigma_k}) } \cdot \Phi(\bm u_{\sigma_k}) \\
&= \sum_{k=1}^m \frac1{p(u_k)}  \Res{w \rightarrow u_k} \even{w}{ p(w) \, \Gamma(w;\bm u_{\sigma_k}) } \even{v}{ \frac{p(v)}{u_k-v}\, \bee{1}{v}}   \\
&\qu \times \prod_{i=2}^m \bigg(\bee{i}{u_{\sigma_k(i)}} \prod_{j=i-1}^{1} R_{a_j \at_i}(-u_{\sigma_k(j)}-u_{\sigma_k(i)}-\rho) \bigg)  \cdot \Phi(\bm u_{\sigma_k}) .
}
Note that we may use $\check{R}$-matrices to permute the parameters in \eqref{nested} to show that $\Gamma(v;\bm u_\sigma) = \Gamma(v;\bm u)$ for all $\sigma \in \mf S_m$. The Bethe equations are then extracted by demanding $U(v;\bm u)\cdot \Phi(\bm u)=0$. Since each summand is independent, we obtain
\[
 \Res{w \rightarrow u_k} \even{w}{ p(w) \, \Gamma(w;\bm u) } = 0 \qu \text{for}\qu 1 \leq k \leq m
\]
or, more concisely
\eqa{ \label{full_bae}
 \Res{w \rightarrow u_k} \Lambda(w;\bm u) = 0 \qu\text{for} \qu 1 \leq k \leq m .
}


\subsection{Boundary eigenvalues and Bethe equations} \label{sec:res}

From the algebraic Bethe ansatz for a $Y(\mfgl_n)$-system, we have explicit values for the eigenvalues of the nested system, see \eqref{gln:full_eig},
\aln{
\Gamma(v;\bm u) &= \lambda_1(v;\bm u)\prod_{i=1}^{m^{(1)}} \frac{v-u_i^{(1)}+1}{v-u_i^{(1)}}  + \lambda_n(v;\bm u)\prod_{i=1}^{m^{(n-1)}}\frac{v-u^{(n-1)}_i-1}{v-u^{(n-1)}_i} 
\\
& \qu + \sum_{k=2}^{n-1} \lambda_k(v;\bm u) \prod_{i=1}^{m^{(k-1)}} \frac{v-u^{(k-1)}_i-1}{v-u^{(k-1)}_i} \prod_{i=1}^{m^{(k)}} \frac{v-u^{(k)}_i+1}{v-u^{(k)}_i} ,
}
where $\la_k(v;\bm u)$ are given by Proposition~\ref{P:ResTT}. (Note that the $(i+1)$-th level of nesting for $Y^\pm(\mfgl_{2n})$ corresponds to $i$-th level for $Y(\mfgl_n)$.) The parameters $u_i^{(k)}$ satisfy the appropriate Bethe equations of $Y(\mfgl_n)$ given in~\eqref{gln:BE}. 
The full eigenvalues $\Lambda(v;\bm u) = \even{v}{p(v)\,\Gamma(v;\bm u)}$ of the Bethe vectors, cf., \eqref{eig_full}, are then obtained by substituting our values for $\la_k(v;\bm u)$ from \eqref{la(v;u)}. This leads to the following statement.

\begin{thrm} \label{T:eig} 
The eigenvalues of the Bethe vectors for a $Y^\pm(\mfgl_{2n})$-system are given by 
\equ{ 
\Lambda(v;\bm u) = \bigg(1 \pm \frac1{2v+\rho}\bigg)\,\Gamma(v;\bm u) + \bigg(1 \mp \frac1{2v+\rho}\bigg)\, \Gamma(-v-\rho;\bm u),  \label{final_eig}
}
where 
\ali{ 
\Gamma(v;\bm u) &=
\Bigg( \prod_{j=1}^\ell \frac{v-c_j-\la^{(j)}_1}{v-c_j} \cdot \frac{v+\rho+c_j+\la^{(j)}_{2n}}{v+\rho+c_j} \Bigg) \Bigg( \prod_{i=1}^{m^{(1)}} \frac{v-u_i^{(1)}+1}{v-u_i^{(1)}} \Bigg)  
\Bigg( \frac{v+(\rho\pm1)/2-\mu_1}{v+(\rho\pm1)/2} \Bigg)
\el
& \qu + \Bigg( \prod_{j=1}^\ell \frac{v-c_j-\la^{(j)}_n}{v-c_j} \cdot \frac{v+\rho+c_j+\la^{(j)}_{n+1}}{v+\rho+c_j} \Bigg) \Bigg( \prod_{i=1}^m \frac{v-u_i+1}{v-u_i} \cdot \frac{v+\rho+u_i+1}{v+\rho+u_i} \Bigg)
\el
&\qq \times \Bigg( \prod_{i=1}^{m^{(n-1)}} \frac{v-u^{(n-1)}_i-1}{v-u^{(n-1)}_i} \Bigg) 
\Bigg( \frac{v+(\rho\pm1)/2-\mu_n}{v+(\rho\pm1)/2} \Bigg)
\el
& \qu + \sum_{k=2}^{n-1} \Bigg( \prod_{j=1}^\ell \frac{v-c_j-\la^{(j)}_k}{v-c_j} \cdot \frac{v+\rho+c_j+\la^{(j)}_{2n-k+1}}{v+\rho+c_j}  \Bigg) \el 
&\qq \times \Bigg( \prod_{i=1}^{m^{(k-1)}} \frac{v-u^{(k-1)}_i-1}{v-u^{(k-1)}_i} \Bigg)\Bigg( \prod_{i=1}^{m^{(k)}} \frac{v-u^{(k)}_i+1}{v-u^{(k)}_i} \Bigg) \Bigg( \frac{v+(\rho\pm1)/2-\mu_k}{v+(\rho\pm1)/2} \Bigg). \label{gln_eig}
}
\end{thrm}

By \eqref{full_bae}, the Bethe equations for $\bm u$ are found by demanding that the residue of the eigenvalue \eqref{final_eig} vanishes at each of the $u_k$ with $1\le k \le m$. Evaluating this residue and using the fact that the $Y^\pm(\mfgl_{2n})$-system reduces to a $Y(\mfgl_n)$-system we obtain the following.

\begin{thrm} \label{T:BE} 
The Bethe equations for a $Y^\pm(\mfgl_{2n})$-system are given by \eqref{gln:BE} and
\ali{
&\frac{u_k+(\rho-1)/2}{u_k+(\rho+1)/2} \cdot \frac{u_k+(\rho\mp1)/2+\mu_n}{u_k+(\rho\pm1)/2-\mu_n}
\Bigg( \prod_{i \neq k} \frac{u_k-u_i-1}{u_k-u_i+1} \cdot \frac{u_k+u_i+\rho-1}{u_k+u_i+\rho+1} \Bigg)
\el
&\qu =\Bigg( \prod_{j=1}^\ell \frac{u_k-c_j-\la_n^{(j)}}{u_k-c_j-\la_{n+1}^{(j)}} \cdot 
\frac{u_k+\rho+c_j+\la_{n+1}^{(j)}}{u_k+\rho+c_j+\la_n^{(j)}} \Bigg) \Bigg( \prod_{i=1}^{m^{(n-1)}} \frac{u_k+\rho+u_i^{(n-1)}}{u_k+\rho+u_i^{(n-1)}+1} \cdot \frac{u_k-u_i^{(n-1)}-1}{u_k-u_i^{(n-1)}} \Bigg)  \label{top-bethe}
}
for $1\le k \le m$.
\end{thrm}

\begin{rmk}
The condition \eqref{gln:betheq} is equivalent to the vanishing of the residue of $\Lambda(v;\bm u)$ in \eqref{final_eig} at each of the $u_i^{(k)}$, which is the expected Bethe equation for a system of equations.
\end{rmk}

\begin{rmk} \label{R:NBE}
The eigenvalue $\Lambda(v;\bm u)$ for a $Y^\pm(\mfgl_{2n})$-system, when the evaluation module $M(\mu)$ of $Y^\pm(\mfgl_{2n})$ in \eqref{M} is a one-dimensional, was calculated in \cite{ACDFR1} by means of the analytical Bethe ansatz. By shifting the roots of the equations and introducing the assumption that the roots come in pairs, one can recover the eigenvalue found in \cite{ACDFR1} from \eqref{gln_eig} and \eqref{final_eig}. However, this assumption appears to be unnecessary for the algebraic Bethe ansatz.
\end{rmk}


\subsection{A trace formula for the Bethe vectors} \label{sec:trace}

Recall the trace formula for the Bethe vectors for the closed $\mfgl_n$-symmetric spin chain \cite{BeRa1,TaVa}. 
This formula allows us to write an expression of the nested Bethe vector of the residual $Y(\mfgl_n)$-system in terms of a trace of elements of the nested monodromy matrix as follows:
\ali{ \label{nested_tf}
& \Phi_{a_1,\dots,a_m,\at_1,\dots,\at_m} (\bm u,\bm{u}^{(1)},\dots ,\bm{u}^{(n-1)})
\el 
& \qq = \tr_{\ol{V}} \Bigg[ \! \Bigg(\prod_{k=1}^{n-1}\prod_{i=1}^{m^{(k)}\!\!} T_{a^{k}_i}(u^{(k)}_i,\bm u) \!\Bigg)\! \Bigg(\prod_{k=2}^{n-1} \prod_{l=1}^{k-1} \prod_{i=1}^{m^{(k)}\!\!\!} \prod_{j=m^{(l)}\!\!}^1 \!\! R_{a^k_i a^l_j}(u^{(k)}_i{-}u^{(l)}_j) \!\Bigg) 
\el 
& \hspace{6cm} \times (e_{21})^{\ot m^{(1)}} \ot \cdots \ot (e_{n,n-1})^{\ot m^{(n-1)}} \Bigg]
 \cdot  \eta ,
}
where the trace is taken over the space $\ol{V} := V_{a_1^1} \ot \cdots \ot V_{a^{n-1}_{m^{(n-1)}}} \cong (\C^n)^{\ot \ol{m}}$ with $\ol{m} = \sum_{i=1}^{n-1} m^{(i)}$ and $\eta$ is the lowest vector \eqref{M'-low} of the nested vacuum sector $M'$, c.f.~\eqref{M'}.
Our goal is to extended this formula for Bethe vectors \eqref{bethevec} of the $Y^\pm(\mfgl_{2n})$-system. 

Recall the notation from Section \ref{sec:block}. The $R$-matrix \eqref{R(u)} acting on $\C^{2n} \ot \C^{2n}$ is denoted by $\RR(u)$ and the matrix units of $\End(\C^{2n})$ (resp.~$\End_U(\C^{2})$) by ${\tt e}_{ij}$ for $1\le i,j \le 2n$ (resp.~$x_{ij}$ for $1\le i,j \le 2$). 
We will use symbols $W_a$ ($W_{a_i}$, $W_{\at_i}$, $W_{a_i^k}$, ...) to denote copies of $\C^{2n}$; symbols $V_a$ ($V_{a_i}$, $V_{\at_i}$, $V_{a_i^k}$, ...) will be reserved for copies of $\C^{n}$. 
When necessary, we will write $({\tt e}_{ij})_a$ to indicate that ${\tt e}_{ij} \in \End(W_a)$, and similarly for $(x_{ij})_a$ and $(e_{ij})_a$; here recall \eqref{e=x*e}.

\begin{prop} \label{P:tf} 
The Bethe vector for the $Y^\pm(\mfgl_{2n})$-system can be written as
\ali{ \label{tf_full}
& \Psi(\bm u, \bm u^{(1)}, \dots, \bm u^{(n-1)}) \el
& \qq = \tr_{\ol{W}} \Bigg[
\prod_{l=1}^m  \Bigg( \! \Bigg( \prod_{j=1}^{l-1} \RR^t_{a_j a_l}(-u_j-u_l-\rho) \!\Bigg) \hat{S}_{a_l}(u_l; \bm u^{(1)}, \dots, \bm u^{(n-1)}) \Bigg) \Bigg(\prod_{k=1}^{n-1}\prod_{i=1}^{m^{(k)}\!\!}S_{a^{k}_i}(u^{(k)}_i)\Bigg)
\el
&\qq\qq\qu \times \Bigg(\prod_{k=2}^{n-1} \prod_{l=1}^{k-1} \prod_{i=1}^{m^{(k)}\!\!} \prod_{j=m^{(l)}\!}^1 \! 
\RR_{a^k_i a^l_j}(u^{(k)}_i{-}u^{(l)}_j) \Bigg) ({\tt e}_{n+1,n})^{\ot m} \ot ({\tt e}_{21})^{\ot m^{(1)}} \ot \cdots \ot ({\tt e}_{n,n-1})^{\ot m^{(n-1)}} \Bigg] \cdot \xi, \nn\\[-2.6em]
}
where the trace is taken over the space $\ol{W} := W_{a_1} \ot \cdots \ot W_{a_m} \ot W_{a_1^1} \ot \cdots \ot W_{a^{n-1}_{m^{(n-1)}}} \cong (\C^{2n})^{\ot (m+\ol{m})}$
with $\ol{m} = \sum_{i=1}^{n-1} m^{(i)}$, and
\ali{ \label{S_hat}
&\hat{S}_a(u;\bm u^{(1)},\dots, \bm u^{(n-1)}) = \Bigg(\prod_{k=n-1}^{1}\prod_{i=m^{(k)}\!\!}^1 \RR_{a a^{k}_i}(-u{-}u^{(k)}_i{-}\rho) \Bigg) S_a(u) 
\Bigg( \prod_{k=1}^{n-1}\prod_{i=1}^{m^{(k)}\!\!} \RR^t_{a a^{k}_i}(u{-}u^{(k)}_i) \Bigg) ,
}
and $\xi$ is the lowest vector of the $Y^{\pm}(\mfgl_{2n})$-module $M$ defined in \eqref{M}.
\end{prop}

\begin{proof} 
We start from \eqref{bethevec}, with $\Phi$ replaced by \eqref{nested_tf},
\ali{ \label{tf_start}
& \Psi (\bm u, \bm u^{(1)}, \dots, \bm u^{(n-1)}) \el
& \qq = \tr_{\ol{V}} \Bigg[
\prod_{l=1}^m \bigg(\beta_{\at_l a_l}(u_l) \prod_{j=l-1}^1 R_{a_j \at_l} (-u_j-u_l-\rho) \bigg) \Bigg(\prod_{k=1}^{n-1}\prod_{i=1}^{m^{(k)}}T_{a^{k}_i}(u^{(k)}_i;\bm u)\Bigg)
\el
&\qq\qq\qu \times 
\Bigg(\prod_{k=2}^{n-1} \prod_{l=1}^{k-1} \prod_{i=1}^{m^{(k)}} \prod_{j=m^{(l)}}^1 
R_{a^k_i a^l_j}(u^{(k)}_i-u^{(l)}_j)
\Bigg) (e_{21})^{\ot m^{(1)}} \ot \cdots \ot (e_{n,n-1})^{\ot m^{(n-1)}} \Bigg]
 \cdot \eta. 
}
The proof shall proceed in two steps. 
First, we shall rewrite the above formula in terms of the $B$-block operator, c.f.~\eqref{block}, rather than creation operators $\beta$, which will allow us to introduce a trace over the corresponding auxiliary spaces.
Then, from this formula, we will argue that the $n\times n$ matrix operators $B$, $T$ and $R$ may be replaced by their $2n \times 2n$ counterparts to complete the proof.


Note that, by commuting matrices which act on different spaces, the creation operator and the product of nested monodromy matrices may be reordered as follows:
\[
\prod_{l=1}^m \bigg(\beta_{\at_l a_l}(u_l) \prod_{j=l-1}^1 R_{a_j \at_l} (-u_j-u_l-\rho) \bigg)=\Bigg( \prod_{l=1}^m \beta_{\at_l a_l}(u_l) \Bigg) \Bigg( \prod_{l=1}^m \prod_{j=l-1}^1 R_{a_j \at_l} (-u_j-u_l-\rho) \Bigg)
\]
and
\aln{
\prod_{k=1}^{n-1}\prod_{i=1}^{m^{(k)}}T_{a^{k}_i}(u^{(k)}_i;\bm u) & =\prod_{k=1}^{n-1} \prod_{i=1}^{m^{(k)}\!\!} 
\Bigg[\! \Bigg( \prod_{l=1}^m R^t_{\at_l a^{k}_i}(u_l-u^{(k)}_i) \Bigg) \!
\Bigg( \prod_{l=1}^m R^t_{a_l a^{k}_i}(-u_l\!-\!u^{(k)}_i\!-\!\rho) \Bigg) A_{a^{k}_i}(u^{(k)}_i)\Bigg]
\\
& =\Bigg( \prod_{l=1}^m \big( {\bm R}_{\at_l}(u_l) \big)^{t_{\at_l}} \Bigg) 
\Bigg( \prod_{l=1}^m \big({\bm R}_{a_l}(-u_l-\rho)\big)^{t_{a_l}} \Bigg) 
\Bigg( \prod_{k=1}^{n-1}\prod_{i=1}^{m^{(k)}} A_{a^{k}_i}(u^{(k)}_i) \Bigg),
}
where we have introduced
\equ{
{\bm R}_{a}(u)= \prod_{k=n-1}^{1}\prod_{i=m^{(k)}}^1 R_{a a^{k}_i}(u-u^{(k)}_i). \label{BR}
}
Dependence on $\bm u^{(1)}, \dots, \bm u^{(n-1)}$ has been omitted for clarity. 
Note also that, as a product of $R$-matrices, ${\bm R}_{a}(u)$ satisfies the RTT relation
\[
R_{ab}(u-v) \, {\bm R}_{a}(u) \, {\bm R}_{b}(v) = {\bm R}_{b}(v) \, {\bm R}_{a}(u) \,  R_{ab}(u-v).
\]
Including these new expressions in \eqref{tf_start}, we make the following reordering, 
\aln{
&\Bigg( \prod_{l=1}^m \prod_{j=l-1}^1 R_{a_j \at_l} (-u_j-u_l-\rho) \Bigg)
\Bigg( \prod_{l=1}^m  {\bm R}^{t_{\at_l}}_{\at_l}(u_l) \Bigg) 
\Bigg( \prod_{l=1}^m {\bm R}^{t_{a_l}}_{a_l}(-u_l-\rho) \Bigg) 
\\
&\qq =
\Bigg( \prod_{l=1}^m   \Bigg( \prod_{j=l-1}^1 R_{a_j \at_l} (-u_j-u_l-\rho) \Bigg) {\bm R}^{t_{\at_l}}_{\at_l}(u_l) \Bigg)
\Bigg( \prod_{l=1}^m {\bm R}^{t_{a_l}}_{a_l}(-u_l-\rho) \Bigg).
}
We now proceed to make repeated applications of the RTT relation, in a similar manner to the proof of Lemma~\ref{L:movement}. 
For example we have, at the centre of the expression,
\aln{
&\Bigg( \prod_{j=m-1}^1 \! R_{a_j \at_m} (-u_j-u_m-\rho) \!\Bigg) {\bm R}^{t_{\at_m}}_{\at_m}(u_m)\,
{\bm R}^{t_{a_1}}_{a_1}(-u_1\!-\!\rho) \cdots {\bm R}^{t_{a_m}}_{a_m}(-u_m\!-\!\rho)
\\
&\qq ={\bm R}^{t_{a_1}}_{a_1}(-u_1-\rho) \cdots {\bm R}^{t_{a_{m-1}}}_{a_{m-1}}(-u_{m-1}-\rho)  {\bm R}^{t_{\at_m}}_{\at_m}(u_m) \Bigg( \prod_{j=m-1}^1 R_{a_j \at_m} (-u_j-u_m-\rho) \Bigg)
{\bm R}^{t_{a_m}}_{a_m}(-u_m-\rho).
}
Continuing inductively, we obtain the equality
\aln{
&\prod_{l=1}^m\Bigg( \! \Bigg( \prod_{j=l-1}^1 R_{a_j \at_l} (-u_j-u_l-\rho) \Bigg) {\bm R}^{t_{\at_l}}_{\at_l}(u_l)  \Bigg)
\Bigg( \prod_{l=1}^m {\bm R}^{t_{a_l}}_{a_l}(-u_l-\rho) \Bigg)
\\
&\qq\qq =\prod_{l=1}^m\Bigg(  {\bm R}^{t_{\at_l}}_{\at_l}(u_l)  \Bigg( \prod_{j=l-1}^1 R_{a_j \at_l} (-u_j-u_l-\rho) \Bigg) {\bm R}^{t_{a_l}}_{a_l}(-u_l-\rho) \Bigg).
}
Therefore, \eqref{tf_start} is equivalent to 
\ali{ \label{tf_beta}
& \Psi(\bm u, \bm u^{(1)}, \dots, \bm u^{(n-1)}) \el
& \qq = \tr_{\ol{V}} \Bigg[ \prod_{l=1}^m\Bigg( \beta_{\at_l a_l}(u_l)\, {\bm R}^{t_{\at_l}}_{\at_l}(u_l)  \Bigg( \prod_{j=l-1}^1 R_{a_j \at_l} (-u_j-u_l-\rho) \Bigg) {\bm R}^{t_{a_l}}_{a_l}(-u_l-\rho) \Bigg) \Bigg(\prod_{k=1}^{n-1}\prod_{i=1}^{m^{(k)}}A_{a^{k}_i}(u^{(k)}_i)\Bigg)
\el
&\qq\qq\qu\; \times
\Bigg(\prod_{k=2}^{n-1} \prod_{l=1}^{k-1} \prod_{i=1}^{m^{(k)}} \prod_{j=m^{(l)}}^1 
R_{a^k_i a^l_j}(u^{(k)}_i-u^{(l)}_j) (e_{21})^{\ot m^{(1)}} \ot \cdots \ot (e_{n,n-1})^{\ot m^{(n-1)}} \Bigg]
 \cdot \eta. \nn\\[-2.6em]
}
To obtain an expression in terms of the $B$-block operator (as opposed to the creation operator $\beta$), we utilise \eqref{XBY}. Indeed, 
\aln{ 
&\Psi(\bm u, \bm u^{(1)}, \dots, \bm u^{(n-1)}) 
\\
& \qu = \tr_{\ol{V}} \left[
\left( \sum_{\substack{r_1, \dots, r_m,\\ s_1, \dots, s_m=1}}^n
\Bigg(\prod_{l=1}^m\Bigg( \Bigg( \prod_{j=1}^{l-1} R^t_{a_j a_l} (-u_j-u_l-\rho) \Bigg) {\bm R}_{a_l}(-u_l-\rho) B_{a_l}(u_l) {\bm R}^{t_{a_l}}_{a_l}(u_l)  
 \Bigg) \Bigg)_{s_1, \bar{r}_1, \dots, s_m, \bar{r}_m } \right.\right.
 \\
& \qq\qu \left. \phantom{\sum_{\substack{r_1, \dots, r_m,\\ s_1, \dots, s_m=1}}^n} \ot  e^*_{r_1} \ot \cdots \ot e^*_{r_m} \ot e^*_{s_1} \ot \cdots \ot e^*_{s_m} \right) \Bigg(\prod_{k=1}^{n-1}\prod_{i=1}^{m^{(k)}}A_{a^{k}_i}(u^{(k)}_i)\Bigg)
\\
& \qq\qu \left. \phantom{\sum_{\substack{aa\\ aa}}^n} \times 
\Bigg(\prod_{k=2}^{n-1} \prod_{l=1}^{k-1} \prod_{i=1}^{m^{(k)}} \prod_{j=m^{(l)}}^1 
R_{a^k_i a^l_j}(u^{(k)}_i-u^{(l)}_j)
\Bigg) (e_{21})^{\ot m^{(1)}} \ot \cdots \ot (e_{n,n-1})^{\ot m^{(n-1)}} \right]
 \cdot \eta.
}
Recall that $\eta = (e_1)^{\ot m} \ot (e_1)^{\ot m} \ot \xi$. After contracting the dual vectors with the vector $\eta$, the resulting matrix element may then be written in terms of a trace over $\wt V := V_{a_1}\ot\cdots\ot V_{a_m} \cong (\C^{n})^{\ot m}$, using the identity $(M)_{ji} =\tr(M e_{ij})$. This gives the expression
\ali{ \label{tf_half}
& \Psi(\bm u, \bm u^{(1)}, \dots, \bm u^{(n-1)}) \el 
& \qq = \tr_{\wt V,\ol{V}} \Bigg[ \prod_{l=1}^m\Bigg(
\Bigg( \prod_{j=1}^{l-1} R^t_{a_j a_l} (-u_j-u_l-\rho) \Bigg) {\bm R}_{a_l}(-u_l-\rho) B_{a_l}(u_l) \, {\bm R}^{t_{a_l}}_{a_l}(u_l) \Bigg)  \Bigg(\prod_{k=1}^{n-1}\prod_{i=1}^{m^{(k)}}A_{a^{k}_i}(u^{(k)}_i)\Bigg)
\el
&\qq\qq\qq \times
\Bigg(\prod_{k=2}^{n-1} \prod_{l=1}^{k-1} \prod_{i=1}^{m^{(k)}} \prod_{j=m^{(l)}}^1 
R_{a^k_i a^l_j}(u^{(k)}_i-u^{(l)}_j)
\Bigg) (e_{1n})^{\ot m} \ot (e_{21})^{\ot m^{(1)}} \ot \cdots \ot (e_{n,n-1})^{\ot m^{(n-1)}} \Bigg]
 \cdot \xi . \nn\\[-2.6em]
}


It remains to show that this expression may be rewritten in terms of the original monodromy matrix $S(u)$ and the $R$-matrix $\RR(u)$.
We will do this by showing that the expression \eqref{tf_full} reduces to the above expression \eqref{tf_half}. 
We put $\hat{S}_a := \hat{S}_a(u;\bm u^{(1)},\dots, \bm u^{(n-1)})$ and rewrite the r.h.s.~of \eqref{tf_full} as
\ali{ 
& \tr_{\ol{W}}\!\Bigg[ \Bigg(\prod_{l=1}^m \Bigg(\prod_{j=1}^{l-1} \RR^t_{a_ja_l}(-u_j-u_l-\rho) \Bigg) \hat{S}_{a_l} \Bigg) ({\tt e}_{n+1,n})^{\ot m} \nn\\[-.5em] & \qq\qq\qu \ot \Big( \hat{\bm A}\, \hat{\bm R}\, \Big( ({\tt e}_{21})^{\ot m^{(1)}} \ot \cdots \ot ({\tt e}_{n,n-1})^{\ot m^{(n-1)}} \Big)\Big) \Bigg] \cdot \xi ,
\label{tf_full-2}
}
where operators $\hat{\bm A}$ and $\hat{\bm R}$ denote the products in the third line of \eqref{tf_full}.
Recall \eqref{e=x*e} and write $({\tt e}_{n+1,n})^{\ot m} = ({x}_{21})^{\ot m} \ot (e_{1n})^{\ot m}$ and
\[
({\tt e}_{21})^{\ot m^{(1)}} \ot \cdots \ot ({\tt e}_{n,n-1})^{\ot m^{(n-1)}} = (x_{11})^{\ot \ol{m}} \ot (e_{21})^{\ot m^{(1)}} \ot \cdots \ot (e_{n,n-1})^{\ot m^{(n-1)}} .
\]
From \eqref{R:new} we see that 
\[
\RR_{a_i^k a_j^l}(u_i^{(k)}-u_j^{(l)})\, ({x}_{11})_{a_i^k} ({x}_{11})_{a_j^l} =  ({x}_{11})_{a_i^k} ({x}_{11})_{a_j^l}  R_{a_i^k a_j^l}(u_i^{(k)}-u_j^{(l)})\, .
\]
Moreover,
\[
S_{a_i^k}(u_i^{(k)})\, (x_{11})_{a_i^k} = (x_{11})_{a_i^k}\, A_{a_i^k}(u_i^{(k)}) + (x_{21})_{a_i^k}\, C_{a_i^k}(u_i^{(k)}).
\] 
Since $C_{a_i^k}(u_i^{(k)})\cdot \xi =0$, we can neglect the $C$ operator above. Therefore we can replace $\hat{\bm A}\, \hat{\bm R}\, \Big( ({\tt e}_{21})^{\ot m^{(1)}} \ot \cdots \ot ({\tt e}_{n,n-1})^{\ot m^{(n-1)}} \Big)$ in \eqref{tf_full-2} with $(x_{11})^{\ol{m}} \ot \Big( \hat{A}\, \hat{R}\, \Big( (e_{21})^{\ot m^{(1)}} \ot \cdots \ot (e_{n,n-1})^{\ot m^{(n-1)}} \Big)\Big)$, where operators $\hat{A}$ and $\hat{R}$ denote the operators in the third line of \eqref{tf_half}. 
Now set $\ol{U} := (\C^2)^{\ot(m+\ol{m})}$ and consider the expression
\equ{
\tr_{\ol{U}} \Bigg[ \Bigg(\prod_{l=1}^m \Bigg(\prod_{j=1}^{l-1} \RR^t_{a_ja_l}(-u_j-u_l-\rho) \Bigg) \hat{S}_{a_l} \Bigg) (x_{21})^{\ot m} \ot (x_{11})^{\ol{m}}  \Bigg] . \label{tf_part}
}
To complete the proof we need to show that the trace above is equivalent to the operators in the second line of \eqref{tf_half}. 
Observe from \eqref{R:new} that operators $\RR$ and $\RR^t$ acting on tensor products of $x_{11}$'s and $x_{21}$'s preserve their numbers in the tensor product. Hence the trace \eqref{tf_part} is only nonzero when each $\hat{S}_{a_l}$ maps $(x_{21})_{a_l}$ to $(x_{11})_{a_l}$. 
In particular, using \eqref{R:new}, and the notation \eqref{BR}, we find that
\[
\hat{S}_{a_l} (x_{21})^{\ot (l)} \ot (x_{11})^{m-l+\ol{m}} = (x_{21})^{\ot (l-1)} \ot (x_{11})^{m-l+1+\ol{m}} {\bm R}_{a_l}(-u_l-\rho) B_{a_l}(u_l) {\bm R}^{t_{a_l}}_{a_l}(u_l) + (...) ,
\]
where $(...)$ denotes the terms that do not contribute to the trace. Moreover,
\aln{
& \Bigg(\prod_{j=1}^{l-1} \RR^t_{a_ja_l}(-u_j-u_l-\rho) \Bigg) (x_{21})^{\ot (l)} \ot (x_{11})^{m-l+\ol{m}} \\
& \qq = (x_{21})^{\ot (l)} \ot (x_{11})^{m-l+\ol{m}} \Bigg(\prod_{j=1}^{l-1} R^t_{a_ja_l}(-u_j-u_l-\rho) \Bigg) + (...),
}
where we have used the same notation as above. Hence the trace \eqref{tf_part} is indeed equivalent to the operators in the second line of \eqref{tf_half}. This completes the proof. 
\end{proof}

The trace formula \eqref{tf_full} allows us to obtain the explicit form of the Bethe vectors in terms of the matrix elements of the monodromy matrix $S_a(u)$.

\begin{exam} \label{Ex1}
Let $n\ge 2$, $m\ge 1$ and $m^{(1)}=\dots=m^{(n-1)}=0$. Then
\[
\Psi(u_1,\dots,u_m) = \Bigg(\prod_{i=1}^m \Bigg(\prod_{j=1}^{i-1} \frac{u_i + u_j + \rho+1}{u_i + u_j + \rho}\Bigg) s_{n,n+1}(u_i) \Bigg) \cdot \xi . 
\]
\end{exam}

\begin{exam} \label{Ex2}
Let $n\ge 2$, $m=m^{(i)}=1$ and $m^{(j)}=0$ for all $j\ne i$. Then
\aln{
\Psi(u_1,u_1^{(i)}) &= \Bigg( s_{n,n+1}(u_1)\,s_{i,i+1}(u_1^{(i)})  \\ 
 & \qq + \frac{1}{u_1-u_1^{(i)}} \Bigg( \frac{u_1-u_1^{(i)}-1}{u_1+u_1^{(i)}+\rho} \, s_{i,n+1}(u_1) - s_{n,2n-i+1}(u_1) \Bigg) s_{n,i+1}(u_1^{(i)}) \Bigg) \cdot \xi . 
}
\end{exam}

\begin{exam} \label{Ex3}
Let $n\ge 2$, $m=2$, $m^{(i)}=1$ and $m^{(j)}=0$ for all $j\ne i$. Then
\aln{
\Psi(u_1,u_2,u_1^{(i)}) &= \frac{u_1+u_2+\rho+1}{u_1+u_2+\rho} \Bigg( s_{n,n+1}(u_1)\,s_{n,n+1}(u_2)\,s_{i,i+1}(u_1^{(i)}) \\ & 
\qu - \Bigg( \frac{1}{u_2-u_1^{(i)}}\,s_{n,n+1}(u_1)\,s_{n,2n-i+1}(u_2) \\ 
& \qq\qu - \frac{u_2-u_1^{(i)}-1}{u_2-u_1^{(i)}}\cdot \frac{1}{u_2+u_1^{(i)}+\rho} \Bigg( s_{n,n+1}(u_1)\,s_{i,n+1}(u_2) + \frac{u_2+u_1^{(i)}+\rho+1}{u_1-u_1^{(i)}} \\ 
& \qq\qq\qu \times \Bigg( \frac{u_1-u_1^{(i)}-1}{u_1+u_1^{(i)}+\rho}\,s_{i,n+1}(u_1) - s_{n,2n-i+1}(u_1) \Bigg) s_{n,n+1}(u_2) \Bigg) s_{n,i+1}(u_1^{(i)}) \Bigg) \Bigg) \cdot \xi . 
}
\end{exam}

\begin{rmk}
Note that in Examples \ref{Ex2} and \ref{Ex3} $s_{n,i+1}(u_{1}^{(i)}) \cdot \xi=0$ unless $i=n-1$.
\end{rmk}


\appendix


\section{Nested algebraic Bethe ansatz for \texorpdfstring{$Y(\mfgl_n)$}{}} \label{app:gln}

In this appendix we give, in full detail, the nested algebraic Bethe ansatz for the Yangian $Y(\mfgl_n)$, first constructed by Kulish and Reshetikhin in \cite{KuRs}, to which the algebraic Bethe ansatz for the twisted Yangian $Y^\pm(\mfgl_{2n})$ reduces. Many technical details are omitted in \emph{loc.~cit.} (especially the steps required to derive the explicit form of the unwanted terms; these steps are also omitted in \cite{BeRa1}). Our aim is to fill in these gaps and provide the reader with complete details. We assume that the full quantum space $L$ of the system is a tensor product of the evaluation modules of $Y(\mfgl_n)$, as defined in \eqref{L},
\equ{
L := L(\la^{(1)})_{c_1} \ot L(\la^{(2)})_{c_2} \ot \dots \ot L(\la^{(\ell)})_{c_\ell}, \label{gln:L}
}
with the lowest weight $\la(u)$ given by \eqref{L:la(u)} for $1\le i \le n$.


\rnc{\bee}[1]{\mathscr{b}_{#1}}

\nc{\ayy}[1]{\mathscr{a}(#1)}

\nc{\Bee}[2]{B_{#1}(#2)}
\nc{\Cee}[2]{C_{#1}(#2)}
\nc{\Dee}[2]{D_{#1}(#2)}
\nc{\tee}[1]{t_{i_#1j_#1}(u_{i_#1})}

\nc{\cee}[1]{\mathscr{c}_{#1}}
\nc{\dee}[1]{\mathscr{d}_{#1}}

\nc{\arr}[2]{R'_{#1}(#2)}
\nc{\arrc}[2]{\check{R}'_{#1}(#2)}


\subsection{Exchange relations} \label{sec:gln:rels}

Consider the Yangian $Y(\mfgl_n)$, as defined in Section~\ref{sec:Y}. The $R$-matrix is $R(u) = I - u^{-1} P \in \End(\C^n \ot \C^n)[[u^{-1}]]$, and the generating matrix $T_a(u) \in \End(V_a) \ot Y(\mfgl_n)[[u^{-1}]]$; here $V_a=\C^n$ is an auxiliary space. We will refer to $T_a(v)$ as the monodromy matrix.

Let $V'_a=\C^{n-1}$ and $V^{(k)}_a=\C^{n-k}$ for any $0\le k < n$, so that $V_a=V_a^{(0)}$ and $V'_a=V_a^{(1)}$. We begin by splitting the auxiliary space $V_a = \C + V'_a$. Accordingly, the monodromy matrix $T_a(u)$ splits as follows: 
\[ \renewcommand*{\arraystretch}{1.4}
T_a(u) = 
  \left(
  \begin{array}{c|c}
  \ayy{u} & \Bee{a}{u} \\
  \hline
  \Cee{a}{u} & \Dee{a}{u}
  \end{array}
  \right) ,
\]
where $\ayy{u}=t_{11}(u)$ and
\aln{
  \Bee{a}{u} &= \left(t_{12}(u), \dots , t_{1n}(u)\right) && \in (V'_a)^* \ot Y(\mfgl_n)[[u^{-1}]], \\
  \Cee{a}{u} &= \left(t_{21}(u), \dots , t_{n1}(u)\right)^{T} && \in V'_a \ot Y(\mfgl_n)[[u^{-1}]], \\
  \Dee{a}{u} &= \left(
  \begin{array}{ccc}
    t_{22}(u) & \dots & t_{2n}(u) \\
    \vdots & \ddots & \vdots \\
    t_{n2}(u) & \dots & t_{nn}(u) \\
  \end{array}
  \right)
  \qu && \in \End(V'_a) \ot Y(\mfgl_n)[[u^{-1}]].
}
In particular, $\Bee{a}{u}$ is a row-vector and $\Cee{a}{u}$ is a column-vector. It will be convenient to denote the matrix entries of $\Bee{a}{u}$ by $\bee{i}(u)$ with $1\le i \le n-1$, and similarly for $\Cee{a}{u}$ and $\Dee{a}{u}$. 
Additionally, we introduce a reduced $R$-matrix $R'(u)$ acting on $\C^{n-1} \ot \C^{n-1}$,
\[ \arr{}{u} := I- u^{-1}\! \sum_{i,j=1}^{n-1} e'_{ij} \ot e'_{ji} = I - u^{-1}P'_{}. \]
The defining relations of $Y(\mfgl_n)$ imply the following exchange relations for $\ayy{v}$, $B_a(v)$ and $D_a(v)$: 
\ali{
\ayy{v} \Bee{a_1}{u}&=\frac{v-u+1}{v-u} \,\Bee{a}{u} \,\ayy{v}-\frac1{v-u}\Bee{a}{v} \,\ayy{u}, 
 \label{gln:ab} \\
\Dee{a}{v}\Bee{a_1}{u}&=\Bee{a_1}{u} \,\Dee{a}{v}\,\arr{a a_1}{v-u}+\frac1{v-u}\Bee{a_1}{v}\, \Dee{a}{u}\,P'_{a a_1}, \label{gln:db} \\
\!\!\Bee{a_1}{v}\Bee{a_2}{u}&=\frac{v-u}{v-u-1} \Bee{a_2}{u}\Bee{a_1}{v}\,\arr{a_1a_2}{v-u} , \label{gln:bb} 
}
along with an RTT relation 
\eqa{
\arr{a_1a_2}{u-v}\,\Dee{a_1}{u}\,\Dee{a_2}{v}=\Dee{a_2}{v}\,\Dee{a_1}{u}\,\arr{a_1a_2}{u-v} . \label{gln:rdd} 
}
In particular, the coefficients of the matrix entries of $\Dee{a}{v}$ generate a subalgebra $Y(\mfgl_{n-1})\subset Y(\mfgl_n)$ (note that this is not a Hopf subalgebra). Two additional relations will be used, which can be stated more clearly in terms of individual matrix entries of~$T_a(u)$. For any $1 \leq i,j,k \leq n-1$,
\ali{
\cee{k}(u)\,\dee{ij}(v) &= \dee{ij}(v)\,\cee{k}(u)-\frac1{u-v}\big(\dee{kj}(u)\,\cee{i}(v)-\dee{kj}(v)\,\cee{i}(u)\big), \label{gln:cd} \\
[\ayy{v},\dee{ij}(u)] &= \frac1{v-u}\big( \bee{j}(u)\,\cee{i}(v)- \bee{j}(v)\,\cee{i}(u) \big) . \label{gln:ad} 
}
%


\subsection{Exchange relations}

\nc{\Bees}{\Bee{a_1}{u_1} \cdots \Bee{a_m}{u_m}}
\nc{\create}[2]{\mathscr{B}_{a_1 \dots a_{#1}}(#2)}
\nc{\setu}{{\bm u}}
\nc{\aux}{a_1 \dots a_m}

We now use the exchange relations stated above to establish algebraic relations that will be important in the nested algebraic Bethe ansatz.

Choose $m \in \N$ and introduce an $m$-tuple $\bm u = ( u_1, \dots, u_m)$ of formal parameters. Let $V'_{a_1}, \dots, V'_{a_m}$ be copies of $V'_a$. The \emph{creation operator} for $m$ excitations is 
\[
\create{m}{\bm u} := B_{a_1}(u_1) \cdots B_{a_m}(u_m).
\]
The operator $\create{m}{\bm u}$ is a row-vector in $(V'_{a_1})^* \ot \dots \ot (V'_{a_m})^*$ with entries in $Y(\mfgl_n)[u_1,\dots,u_m][[u_1^{-1},\dots,u_m^{-1}]]$. The parameters carried by $\create{m}{\bm u}$ may be exchanged by the braided $R$-matrix defined by 
\equ{
\check{R}'(u) := \frac{u}{u-1} \,R'(u) P'. \label{gln:checkR'}
}
This $R$-matrix allows us to rewrite \eqref{gln:bb} in a more elegant form,
\[
B_{a_1}(u_1) B_{a_2}(u_2) = B_{a_1}(u_2) B_{a_2}(u_1) \check{R}_{a_1 a_2}'(u_1-u_2) .
\]
Consequently, for $m$ excitations, we have that
\eqa{ \label{gln:Br'}
\create{m}{\bm u} = \create{m}{\bm u_{i \leftrightarrow i+1}} \,\check{R}_{a_i a_{i+1}}'(u_i-u_{i+1}) 
\qu \text{for} \qu 1 \leq i \leq m-1,
}
where we used the notation \eqref{uii}.

We now move the $\ayy{v}$ and $D_a(v)$ operators through the $m$-excitation creation operator. 
Consider first the action of $\ayy{v}$ on $\create{m}{\bm u}$, 
\[
\ayy{v} \, \create{m}{\bm u} = \bigg(\frac{v-u_1+1}{v-u_1} \, B_{a_1}(u_1) \, \ayy{v} - \frac{1}{v-u_1} \, B_{a_1}(v) \, \ayy{u_1} \bigg)  B_{a_2}(u_2) \cdots B_{a_m}(u_m). 
\]
Note that we may repeatedly apply this relation to move $\ayy{v}$ through each of the excitations, resulting in a sum of $2^m$ terms in which $\ayy{\cdot}$ is the rightmost operator. From this sum we note that there is a unique term in which $\ayy{v}$ retains its parameter each time we apply \eqref{gln:ab}. We will refer to this term as the \emph{wanted term}, and the other terms as \emph{unwanted terms} (UWT). Then, 
\eqa{ \label{gln:abb}
\ayy{v} \, \create{m}{\bm u} = \prod_{i=1}^m \frac{v-u_i+1}{v-u_i} \, \create{m}{\bm u} \, \ayy{v} + UWT.
}
The unwanted terms will be discussed in detail in Section~\ref{sec:gln:uwt}.

The wanted term for the action of $D_a(v)$ on $\create{m}{\bm u}$ is found similarly. From repeated applications of \eqref{gln:db},
\[
D_a(v) \, \create{m}{\bm u}= \create{m}{\bm u} \, D_a(v) \, R'_{a a_m}(v-u_m) \cdots R'_{a a_1}(v-u_1) + UWT.
\]
Here note that the rightmost matrix acting on the auxiliary space $V'_a$ is 
\eqa{ \label{gln:nmm}
T_{a;\aux}'(v; \bm u) := D_a(v) \, R'_{a a_m}(v-u_m) \cdots R'_{a a_1}(v-u_1) .
}
We will refer to this matrix as the \emph{nested monodromy matrix}. The nontrivial action on the auxiliary spaces $V'_{a_1}, \dots, V'_{a_m}$ will often be omitted for clarity, and we will write instead simply $T'_a(v;\bm u)$. Using this notation we get
\eqa{ \label{gln:dbb}
D_a(v) \, \create{m}{\bm u} = \create{m}{\bm u} \, T'_{a}(v;\bm u) + UWT.
}

The nested monodromy matrix satisfies the following properties.

\begin{lemma} \label{L:gln:nRTT}
The matrix $T'_a(v;\bm u)$ satisfies the RTT relation,
\[
R'_{ab}(v-w)\, T'_a(v;\bm u)\, T'_b(w;\bm u) = T'_b(w;\bm u)\, T'_a(v;\bm u)\, R'_{ab}(v-w).
\] 
\end{lemma}

\begin{proof}
Starting from the l.h.s.~of the equation and using the definition \eqref{gln:nmm} of $T'_a(v;\bm u)$,
\begin{align*}
\text{l.h.s.} &= R'_{ab}(v-w) \, D_a(v) \, R'_{a a_m}(v-u_m) \cdots R'_{a a_1}(v-u_1) \, D_b(w) \, R'_{b a_m}(w-u_m) \cdots R'_{b a_1}(w-u_1) \\
&= R'_{ab}(v-w) \, D_a(v) D_b(w) \, R'_{a a_m}(v-u_m) \, R'_{b a_m}(w-u_m) \cdots R'_{a a_1}(v-u_1) R'_{b a_1}(w-u_1)  \\
&=  \, D_b(w) D_a(v) R'_{ab}(v-w) \, R'_{a a_m}(v-u_m) \, R'_{b a_m}(w-u_m) \cdots R'_{a a_1}(v-u_1) R'_{b a_1}(w-u_1) \qu \text{ by \eqref{gln:rdd} } \\
&= \, D_b(w) D_a(v) \, R'_{b a_m}(w-u_m) \, R'_{a a_m}(v-u_m) \cdots R'_{b a_1}(w-u_1) R'_{a a_1}(v-u_1) R'_{ab}(v-w) \qu \text{ by YBE } \\
&= T'_b(w;\bm u)\, T'_a(v;\bm u)\, R'_{ab}(v-w). \qedhere
\end{align*}
\end{proof}

By the above Lemma, the matrix $T'_a(v;\bm u)$ is a homomorphic image of the generating matrix $T'_a(v)$ of $Y(\mfgl_{n-1})$. We may use $R$-matrices $\check{R}'$ to exchange the ordering of the parameters in $\bm u$.

\begin{lemma} \label{L:gln:R'T}
Matrix elements $t'_{ij}(v;\bm u)$ of $T'_a(v;\bm u)$ satisfy the relation:
\[
\check{R}'_{a_i a_{i+1}}(u_i-u_{i+1})\, t'_{jk}(v; \bm u) 
= t'_{jk}(v; \bm u_{i \leftrightarrow i+1})\, \check{R}'_{a_i a_{i+1}}(u_i-u_{i+1}). 
\]
\end{lemma}

\begin{proof}
Moving $\check{R}'_{a_i a_{i+1}}(u_i-u_{i+1})$ from left to right through each of the $R$-matrices in the definition \eqref{gln:nmm}, the $R$-matrices with which it does not commute will undergo parameter exchange $u_i \leftrightarrow u_{i+1}$ due to the (braided) Yang-Baxter equation:
\[
\check{R}'_{a_i a_{i+1}}(u_i-u_{i+1}) R'_{a a_{i+1}}(v-u_{i+1}) R'_{a a_i}(v-u_i) =  R'_{a a_{i+1}}(v-u_i) R'_{a a_i}(v-u_{i+1}) \check{R}'_{a_i a_{i+1}}(u_i-u_{i+1}) .
\]
The required identity is now immediate. 
\end{proof}


We now construct a finite-dimensional vector space, called the \emph{nested vacuum sector}, which the matrix $T'_a(v;\bm u)$ will act on.
Denote by $L(\la^{(i)})^0_{c_i}$ the subspace of the $Y(\mfgl_n)$-evaluation module $L(\la^{(i)})_{c_i}$ consisting of vectors annihilated by all operators $\cee{j}(u)$, namely
\[
L(\la^{(i)})^0_{c_i} := \{ \zeta \in L(\la^{(i)})_{c_i} \,:\, \cee{j}(u)\,\zeta = 0 \text{ for } 1\le j \le n-1 \}. 
\]
This subspace corresponds to the natural embedding $\mfgl_{n-1}\subset\mfgl_n$ and is an irreducible lowest weight $Y(\mfgl_{n-1})$-module with the lowest weight given by 
\equ{
\la_{i}(u)^0 = \la_{i+1}(u) \qu\text{for}\qu 1\le i \le n-1 \label{gln:la0(u)}
}
and $\la_i(u)$ defined in \eqref{L:la(u)}.

We define the \emph{vacuum sector} $L^0\subset L$ by
\[
L^0 = L(\la^{(1)})^0_{c_1} \ot L(\la^{(2)})^0_{c_2} \ot \dots \ot L(\la^{(\ell)})^0_{c_\ell} .
\]
By the initial assumption, the space $L$ is an irreducible $Y(\mfgl_n)$-module. Then, by Lemma 6.2.2 and Theorem 6.5.8 in \cite{Mo3}, the space $L^0$ is an irreducible $Y(\mfgl_{n-1})$-module. In particular, the space $L^0$ is annihilated by all operators $\cee{i}(u)$,
\[
L^0 = \{ \zeta \in L \;:\; \cee{i}(u)\cdot \zeta =0\;\text{ for }\; 1\le i \le n-1\},
\]
and is stable under the action of the operators $\dee{ij}(u)$ for $1\le i,j\le n-1$, see \eqref{gln:cd}.

Each auxiliary space $V'_{a_i}$ is a vector representation of the Lie algebra $\mfgl_{n-1}$ of weight $\la'=(1,0,\dots,0)$ and may be viewed as an evaluation module $L(\la')_{u_i}$ of $Y(\mfgl_{n-1})$ with the lowest weight given~by
\equ{
\la'_1(u)=\frac{u-u_i-1}{u-u_i} \qu\text{and}\qu \la'_j(u) = 1 \qu\text{for}\qu 2\le j \le n-1. \label{gln:la'(u)}
}
In particular, the generating matrix $T'_a(u)$ of $Y(\mfgl_{n-1})$ acts on $L(\la')_{u_i}$ as $R'_{aa_i}(u-u_i)$.

We have now all the necessary ingredients to define the \emph{nested vacuum sector}
\equ{
L' := L^0\ot V'_{a_m} \ot \cdots \ot V'_{a_1} . \label{gln:L'}
}

\begin{prop} \label{P:gln:T'.L'}
Let $T'(v)$ denote the generating matrix of $Y(\mfgl_{n-1})$. Then the map
\equ{
Y(\mfgl_{n-1}) \to Y(\mfgl_{n})\ot \End(V'_{a_m} \ot \cdots \ot V'_{a_1}) , \qu T'(v) \mapsto T'(v;\bm u) \label{gln:T'(v)->T'(v;u)}
}
is a homomorphism of algebras. Moreover, it equips the space $L'$ with a structure of a lowest weight $Y(\mfgl_{n-1})$-module with the lowest weight given by
\eqa{
\lambda_1'(v;{\bm u}) &= \prod_{j=1}^\ell \frac{v-\la_2^{(j)}-c_j-1}{v-\la_2^{(j)}-c_j} \prod_{k=1}^m \frac{v-u_k-1}{v-u_k} \qu\text{and}\qu\\
\lambda_i'(v;{\bm u}) &= \prod_{j=1}^\ell \frac{v-\la_{i+1}^{(j)}-c_j-1}{v-\la_{i+1}^{(j)}-c_j} \qu\text{for}\qu  2\le i \le n-1. \label{gln:la'(v;u)}
}
\end{prop}

\begin{proof}
The homomorphism property follows from Lemma \ref{L:gln:nRTT}. We already know that $L^0$ is an irreducible $Y(\mfgl_{n-1})$-module. It follows from \eqref{gln:nmm} and \eqref{gln:L'} that the space $L'$ is stable under the action of $T'_a(v;\bm u)$. Thus the map \eqref{gln:T'(v)->T'(v;u)} equips the space $L'$ with a structure of $Y(\mfgl_{n-1})$-module with each tensorand a lowest weight $Y(\mfgl_{n-1})$-module. The lowest vector is
\equ{
\eta = \eta_1 \ot \cdots \ot \eta_\ell \ot e'_1 \ot \cdots \ot e'_1 , \label{gln:L'-low}
}
where each $\eta_i$ is a lowest vector of $L(\la^{(i)})_{c_i}^0$ for $1\le i \le \ell$ and each $e'_1$ is a lowest vector of $V_{a_i}$ for $1\le i \le m$ (viewed as an evaluation module $L(\la')_{u_i}$). Finally, acting with $t'_{ii}(v;\bm u)$ on $\eta$ for $1\le i \le n$ and using \eqref{gln:la0(u)} and \eqref{gln:la'(u)} yields \eqref{gln:la'(v;u)}.
\end{proof}

\begin{lemma} \label{L:gln:a.L'}
For any vector $\zeta \in L'$ we have that $\ayy{u} \cdot \zeta = \lambda_1(u) \, \zeta$, where $\la_1(u)$ is defined by \eqref{L:la(u)}.
\end{lemma}

\begin{proof}
By Proposition \ref{P:gln:T'.L'} we know that $L' = Y(\mfgl_{n-1})\,\eta$ for $\eta$ defined in \eqref{gln:L'-low} and $\cee{i}(u)\cdot L' = 0$. Using \eqref{gln:ad} and definition of $t'_{ij}(v;\bm u)$, we find that $[\ayy{u},t'_{ij}(v;\bm u)] \cdot \zeta = 0$ for any $1\le i,j \le n$. 
Hence it is enough to act with $\ayy{u}$ on the lowest vector $\eta$, which yields the required result.
\end{proof}


\subsection{Nested algebraic Bethe ansatz} \label{sec:gln:naba}

We are now ready to consider the nested algebraic Bethe ansatz for a $Y(\mfgl_n)$-system. The monodromy matrix $T_a(v)$ acts on the space $L$ in \eqref{gln:L} by
\[
T_a(v)\cdot L = \Bigg( \prod_{i=1}^\ell \mc{L}_{ai}(v-c_i) \Bigg) L \in \End(\C^n)\ot L [[v^{-1}]].
\]
The \emph{transfer matrix} is defined as
\[
t(v) := \tr_a T_a(v) \in Y(\mfgl_n) [[ v^{-1}]].
\]
Taking the trace of the RTT relation reveals that $[t(v),t(u)]=0$, and so $t(v)$ is a generating series for conserved quantities. 
We diagonalise $t(v)$ by means of the nested algebraic Bethe ansatz, adhering closely to \cite{KuRs}. In particular, we construct an ansatz for our eigenvector recursively, at each stage reducing the diagonalisation problem to a similar problem with a smaller symmetry algebra, relying on the chain of subalgebras $Y(\mfgl_n)\supset Y(\mfgl_{n-1})\supset \dots \supset Y(\mfgl_2)$ and the irreducibility criterion of a tensor product of evaluation modules.

Recall the definition of the full quantum space \eqref{gln:L} and the nested vacuum sector \eqref{gln:L'}.
Let $\Phi' \in L'$. We will refer to this as the \emph{nested Bethe vector}, imposing additional properties in a later section. The ansatz for the eigenvector of the transfer matrix, the \emph{Bethe vector}, is given by
\[
\Phi(\bm u) := \create{m}{\bm u} \cdot \Phi' \in L .
\]
Since $L$ is a finite dimensional vector space, the parameters $u_i\in\bm u$ can be evaluated to nonzero complex numbers, hence from now on we will assume that $\bm u\in\C^{m}$ is an $m$-tuple of nonzero complex numbers. To find the eigenvalues and Bethe equations, we act on the Bethe vector with $t(v) = \ayy{v} + \tr_a  D_a(v)$.  Using \eqref{gln:abb} and \eqref{gln:dbb} we write
\[
t(v) \cdot \Phi(\bm u) = \create{m}{\bm u}\, \Bigg( \prod_{i=1}^m \frac{v-u_i+1}{v-u_i} \, \ayy{v} + \tr_a T'_a(v;\bm u) \Bigg) \cdot \Phi' + UWT .
\]
By Lemma~\ref{L:gln:a.L'}, $\ayy{v} \cdot \Phi' = \lambda_1(v) \, \Phi'$. We now impose that $\Phi'$ is an eigenvector of the nested transfer matrix $t'(v;\bm u) := \tr_a T'_a(v;\bm u)$, with the eigenvalue $\Gamma'(v;\bm u)$, namely
\eqa{ \label{gln:nGamma}
t'(v;\bm u) \cdot \Phi' = \Gamma'(v;\bm u) \, \Phi'. 
}
With this condition, the full action of the transfer matrix is diagonal, plus unwanted terms,
\equ{ 
t(v) \cdot \Phi(\bm u) = \Gamma(v;\bm u) \, \Phi(\bm u) + UWT, \qu \text{where}\qu \Gamma(v;\bm u) = \lambda_1(v) \prod_{i=1}^m \frac{v-u_i+1}{v-u_i} + \Gamma'(v;\bm u) . \label{gln:Gamma}
}
Finding $\Phi'$ satisfying \eqref{gln:nGamma} defines another transfer matrix diagonalisation problem, namely for the Yangian $Y(\mfgl_{n-1})$. The monodromy matrix in this case is given by $T'_a(v;\bm u)$ and the full quantum space is the lowest weight $Y(\mfgl_{n-1})$-module $L'$ defined in \eqref{gln:L'}, so the problem may again be reduced by means of the nested algebraic Bethe ansatz; this is ensured by Proposition \ref{P:gln:T'.L'}. For example, constructing the ansatz for the nested Bethe vector, we fix $m' \in \N$ and introduce an $m'$-tuple $\bm u'= (u'_1,\ldots,u'_{m'})$ of distinct complex parameters, so that
\[
\Phi' = \Phi'(\bm u';\bm u) = B'_{a'_1}(u'_1;\bm u)\cdots B'_{a'_{m'}}(u'_{m'};\bm u) \cdot \Phi'',
\]
where, upon decomposing the nested transfer matrix $T'(v,\bm u)$ in the same way as we did for $T(v)$, 
\[
\Phi''\in L^{\prime\, 0} \ot V''_{a'_{m'}}\ot \cdots \ot V''_{a'_1} . 
\]
Here $L^{\prime\,0}$ is the vacuum sector of $L'$ defined analogously to that of $L$, and each $V''_{a'_{i'}}$ is a $\mfgl_{n-2}$-module of weight $\la''=(1,0,\dots,0)$.
Repeating this process, we reduce the problem to a $Y(\mfgl_2)$-system, the solution of which is well known, see e.g., \cite{FdTk}.

The transfer matrix will therefore act diagonally on $\Phi(\bm u)$ if all the unwanted terms vanish. We will show in the next section how this requirement leads to a set of Bethe equations for~$\bm u$.


\subsection{Dealing with unwanted terms} \label{sec:gln:uwt}

Recall \eqref{gln:abb} and \eqref{gln:dbb}. Introduce the following notation for the unwanted terms:
\aln{
\ayy{v} \, \create{m}{\bm u} &= \prod_{i=1}^m \frac{v-u_i+1}{v-u_i} \, \create{m}{\bm u} \,\ayy{v} + U_1(v; \bm u)
, 
\\
\tr_a \Dee{a}{v} \, \create{m}{\bm u} &= \create{m}{\bm u} \, \tr_a T'_a(v; \bm u) + U_2(v; \bm u)
.
}
By applying \eqref{gln:ab} repeatedly, we obtain an expression for $U_1(v; \bm u)$ as a sum of $2^m$ terms, in each of which $\ayy\cdot$ is the rightmost operator. This inspires a further partition of the unwanted terms. Let $\mc B$ denote the subalgebra of $Y(\mfgl_n)$ generated by the coefficients of the series $\bee{i}(u)$ for $1 \leq i \leq n-1$, whose closure is guaranteed by \eqref{gln:bb}. 
We decompose the unwanted terms $U_1(v; \bm u)$ and $U_2(v; \bm u)$ by
\aln{ 
U_1(v;\bm u) &= \sum_{j=1}^m U_{1,j}(v; \bm u) \qu \text{such that} \qu U_{1,j}(v; \bm u) = B_{1,j} \, \ayy{u_j}  \qu\text{and} \\
U_2(v;\bm u) &= \sum_{j=1}^m U_{2,j}(v; \bm u) \qu \text{such that} \qu U_{2,j}(v; \bm u) = \sum_{k,l=1}^{n-1} B_{2,j;kl} \, \dee{kl}(u_j)
}
for some $B_{1,j}, B_{2,j;kl} \in \mc B \ot (\C^{n-1})^{\ot m}[[v^{-1}]]$.

To find $U_{1,1}(v;\bm u)$, we begin by acting on $\create{m}{\bm u}$ with $\ayy{v}$. From \eqref{gln:ab}, we have
\[
\ayy{v} \, \create{m}{\bm u} = \bigg( \frac{v-u_1+1}{v-u_1} \, B_{a_1}(u_1) \, \ayy{v} - \frac1{v-u_1} \, B_{a_1}(v) \, \ayy{u_1} \bigg)  B_{a_2}(u_2) \cdots B_{a_m}(u_m).
\]
Now, moving $\ayy{v}$ through the remaining creation operators, we note that the only contribution to $U_{1,1}(v;\bm u)$ will be from the second term in the above expression, in the instance when there are no further parameter swaps in the remaining commutations. Therefore,
\[
U_{1,1}(v;\bm u) = - \frac1{v-u_1} \prod_{j=2}^m \frac{v-u_j+1}{v-u_j}\, B_{a_1}(v) \, B_{a_2}(u_2) \cdots B_{a_m}(u_m)\, \ayy{u_1}.
\]
We can find $U_{2,1}(v;\bm u)$ in a similar way. Acting with $\tr_a \Dee{a}{v}$ on $\create{m}{\bm u}$,
\aln{ 
\tr_a \Dee{a}{v} \, \create{m}{\bm u} &= \bigg( B_{a_1}(u_1) \, \tr_a \Dee{a}{v} \, R'_{aa_1}(v-u_1) \\ & \qq\qu - \frac1{v-u_1} \, B_{a_1}(v) \tr_a \Dee{a}{u_1}\, P'_{aa_1} \bigg) B_{a_2}(u_2) \cdots B_{a_m}(u_m) .
}
As above, we move the $D_a(u_1)$ operator through the remaining creation operators, and the only contribution to $U_{2,1}(v;\bm u)$ is
\aln{
U_{2,1}(v;\bm u) &= - \frac1{v-u_1}\, B_{a_1}(v) \, B_{a_2}(u_2) \cdots B_{a_m}(u_m) \, \tr_a \Big( \Dee{a}{u_1} \,R'_{aa_m}(u_1-u_m) \cdots R'_{aa_2}(v-u_2) \, P'_{a a_1} \Big).
}
Note that the operators to the right of $D_a(u_1)$ act trivially on $L$, so this is indeed in the correct form. It will be useful to rewrite this in terms of a residue as follows
\aln{
U_{2,1}(v;\bm u) &= \frac1{v-u_1} B_{a_1}(v) \, B_{a_2}(u_2) \cdots B_{a_m}(u_m) \\ & \qu \times\Res{w \rightarrow u_1} \! \tr_a \Big( \Dee{a}{w} R'_{aa_m}(w-u_m) \cdots R'_{aa_2}(w-u_2) \,  R'_{aa_1}(w-u_1) \Big) \\
&= \frac1{v-u_1} B_{a_1}(v) \, B_{a_2}(u_2) \cdots B_{a_m}(u_m) \Res{w \rightarrow u_1} \!\tr_a T_a'(w;\bm u).
}

Proceeding this way we can find the remaining unwanted terms. Consider the relation \eqref{gln:Br'}. We may use this relation to transpose the parameters $\bm u$ and, since the transpositions generate the symmetric group $\mf S_m$, we can apply an arbitrary permutation to the parameters prior to acting with $\ayy{v}$. Indeed for $\sigma \in \mf S_m$, let $\bm u_{\sigma}$ denote the ordered set $(u_{\sigma(1)},u_{\sigma(2)}, \dots, u_{\sigma(m)})$. Additionally, let $\sigma_j$ denote the cyclic permutation $\sigma_j: (1,2, \dots, m) \mapsto (j, j+1, \dots, j-1)$. We have
\[
\create{m}{\bm u} = \create{m}{\bm u_{\sigma_j}} \check{R}'_{a_1 \dots a_m}[\sigma_j](\bm u),
\]
where $\check{R}'_{a_1 \dots a_m}[\sigma_j](\bm u)$ is the product of $\check{R}'$-matrices required to realise this permutation. Acting now with $\ayy{v}$ on the r.h.s.~and following the argument above, we obtain an exact expression for $U_{1,j}(v;\bm u)$, namely
\[
U_{1,j}(v; \bm u) = -\frac1{v-u_j} \prod_{k \neq j} \frac{u_j-u_k+1}{u_j-u_k} \,\create{m}{\bm u_{\si_j,u_j\rightarrow v}}\, \ayy{u_j} \, \check{R}'_{a_1 \dots a_m}[\sigma_j](\bm u) ,
\]
where $\create{m}{\bm u_{\si_j,u_j\rightarrow v}} = B_{a_1}(v) \, B_{a_2}(u_{j+1}) \cdots B_{a_m}(u_{j-1})$.
The expression for $U_{2,j}(v;\bm u)$ obtained by the same method. Indeed,
\[
U_{2,j}(v; \bm u) = -\frac1{v-u_j} \,\create{m}{\bm u_{\si_j,u_j\rightarrow v}} \, \Res{w \rightarrow u_1} t'(w;\bm u_{\sigma_j}) \, \check{R}'_{a_1 \dots a_m}[\sigma_j](\bm u).
\]

Having found these expressions, the next step is to act with them on the nested Bethe vector to find the full expression for the action of the transfer matrix on the Bethe vector. However, we will first make an assumption about the form of $\Phi'(\bm u'; \bm u)$, namely
\[
\check{R}'_{a_i a_{i+1}}(u_i - u_{i+1}) \,\Phi'(\bm u';\bm u) = \Phi'(\bm u';\bm u_{i \leftrightarrow i+1}) \qu \text{for} \qu 1 \leq i \leq m-1.
\]
This may be achieved if the nested Bethe vector is of the form
\[ 
 \Phi'(\bm u';\bm u) \in {\rm span}_\C \big\{  t'_{i_1 j_1}(w_1; \bm u) \cdots t'_{i_K j_K}(w_K; \bm u) \cdot \eta' \, : \, K\geq 0, \; 1 \leq i_1, j_1, \dots, i_K, j_K \leq n-1, \, \bm w \in \C^K \big\},
\]
where $\eta'$ is a lowest weight vector of $L'$. Indeed, for any such vector we may use Lemma~\ref{L:gln:R'T} to move $\check{R}_{a_i a_{i+1}}(u_i - u_{i+1})$ through the $t'_{i_k j_k}(w_k;\bm u)$, exchanging $u_i$ with $u_{i+1}$. Then, by definition and \eqref{gln:checkR'}, $\eta'$ is an eigenvector of $\check{R}_{a_i a_{i+1}}(u_i - u_{i+1})$ with eigenvalue 1. The nested Bethe vector constructed using the nested algebraic Bethe ansatz is exactly of this form. 

Note that this, combined with the relation $\eqref{gln:Br'}$, gives the parameter symmetry of the full Bethe vector $\Phi(\bm u) = \Phi(\bm u_{\sigma})$ for all $\sigma \in \mf S_m$.
Applying the expressions for the unwanted terms to the nested Bethe vector, noting \eqref{gln:nGamma} and Lemma~\ref{L:gln:a.L'},
\aln{
U_{1,j}(v;\bm u)\cdot \Phi'(\bm u';\bm u) & = -\frac1{v-u_j} \,\lambda_1(u_j)  \prod_{k \neq j} \frac{u_j-u_k+1}{u_j-u_k} \,\, \create{m}{\bm u_{\si_j,u_j\rightarrow v}} \cdot \Phi'(\bm u';\bm u_{\sigma_j}), \\
U_{2,j}(v;\bm u)\cdot \Phi'(\bm u';\bm u) & = -\frac1{v-u_j}\Res{w \rightarrow u_j} \Gamma'(w;\bm u_{\sigma_j}) \,\create{m}{\bm u_{\si_j,u_j\rightarrow v}} \cdot \Phi'(\bm u';\bm u_{\sigma_j}) .
}
In fact, by acting with $\check{R}'$-matrices on $t'(v;\bm u)\cdot \Phi'(\bm u';\bm u) = \Gamma'(v;\bm u)\, \Phi'(\bm u';\bm u)$, we obtain $\Gamma'(v;\bm u ) = \Gamma'(v;\bm u_{\sigma} )$. Putting everything together, we have
\aln{ 
U_1(v;\bm u) + U_2(v;\bm u) & = - \sum_{j=1}^m \frac1{v-u_j} \Bigg[ \lambda_1(u_j) \prod_{k \neq j} \frac{u_j-u_k+1}{u_j-u_k} + \Res{w \rightarrow u_j} \Gamma'(w;\bm u) \Bigg] \\ & \hspace{6cm}\times \create{m}{\bm u_{\si_j,u_j\rightarrow v}} \cdot \Phi'(\bm u';\bm u_{\sigma_j}) \\
& = - \sum_{j=1}^m  \frac1{v-u_j} \Res{w \rightarrow u_j} \Gamma(w;\bm u) \, \create{m}{\bm u_{\si_j,u_j\rightarrow v}} \cdot \Phi'(\bm u';\bm u_{\sigma_j}).
}
From \eqref{gln:Gamma}, $\Phi(\bm u)$ is an eigenvector of the transfer matrix $t(v)$ if the parameters $\bm u$ are chosen such that the above expression vanishes. Since each summand is independent, we require
\equ{ \label{gln:betheq}
\Res{w \rightarrow u_j} \Gamma(w;\bm u) = 0 \qu \text{for} \qu 1 \leq j \leq m.
}
These are the \emph{Bethe equations} for $\bm u$.


\subsection{End of recursion}

Upon reducing to the residual $Y(\mfgl_2)$-system, we have the familiar $2 \times 2$  monodromy matrix
\[
T_a^{(n-2)}(v)
=
\left(
\begin{array}{cc}
\mathscr{a}^{(n-2)}(v) & \mathscr{b}^{(n-2)}(v) \\
\mathscr{c}^{(n-2)}(v) & \mathscr{d}^{(n-2)}(v)
\end{array}
\right).
\]
Dependence on parameters $\setu, \setu', \dots, \setu^{(n-3)}$ has been suppressed. The RTT relation yields the relations
\aln{
\mathscr{a}^{(n-2)}(v) \, \mathscr{b}^{(n-2)}(u) &= \frac{v-u+1}{v-u} \, \mathscr{b}^{(n-2)}(u) \, \mathscr{a}^{(n-2)}(v) - \frac1{v-u} \, \mathscr{b}^{(n-2)}(v) \, \mathscr{a}^{(n-2)}(u) , \\
\mathscr{d}^{(n-2)}(v) \, \mathscr{b}^{(n-2)}(u) &= \frac{v-u-1}{v-u} \, \mathscr{b}^{(n-2)}(u) \, \mathscr{d}^{(n-2)}(v) + \frac1{v-u} \, \mathscr{b}^{(n-2)}(v) \, \mathscr{d}^{(n-2)}(u) , \\
[\mathscr{b}^{(n-2)}(v),\mathscr{b}^{(n-2)}(u)]& = 0 .
}
The Bethe vector with $m^{(n-2)}$ excitations is
\[
\Phi^{(n-2)}(\setu) = \mathscr{b}^{(n-2)}(u_1^{(n-2)}) \cdots \mathscr{b}^{(n-2)}(u_{m^{(n-2)}}^{(n-2)}) \cdot \eta^{(n-2)},
\]
where $\eta^{(n-2)}$ is a lowest vector of the nested vacuum sector $L^{(n-2)}$.
The associated eigenvalue of the transfer matrix $t^{(n-2)}(v)$ is
\aln{
\Gamma^{(n-2)}(v;\setu, \dots , \setu^{(n-2)}) &= \lambda_1^{(n-2)}(v;\setu, \dots , \setu^{(n-3)}) \prod_{i=1}^{m^{(n-2)}} \frac{v-u^{(n-2)}+1}{v-u^{(n-2)}} \\
& \qu + \lambda_2^{(n-2)}(v;\setu, \dots , \setu^{(n-3)}) \prod_{i=1}^{m^{(n-2)}} \frac{v-u^{(n-2)}-1}{v-u^{(n-2)}_i},
}
provided the $\setu^{(n-2)}$ satisfy the Bethe equations
\[
\Res{w \rightarrow u^{(n-2)}_j} \Gamma^{(n-2)}(w;\setu, \dots \setu^{(n-2)}) = 0 \qu \text{for} \qu 1 \leq j \leq m^{(n-2)} .
\]


\subsection{Full expressions for eigenvalues and Bethe equations}

In this section, we unpack the recursion steps to give the explicit expressions for the eigenvalues of the transfer matrix in terms of the parameters of the $Y(\mfgl_n)$-system.  
In order to match the notation used in the Bethe ansatz for the $Y^{\pm}(\mfgl_{2n})$ chain, we begin by relabelling the spectral parameters as follows. For the initial step, relabel parameters $u_i \rightarrow u^{(1)}_i$ and excitation number $m \rightarrow m^{(1)}$, and for subsequent levels of nesting $u^{(k)}_i \rightarrow u^{(k+1)}_i$ and $m^{(k)} \rightarrow m^{(k+1)}$. 
We use Proposition \ref{P:gln:T'.L'} to rewrite the weights $\lambda_1^{(k)}(v;\bm u,\dots, \bm u^{(k-1)})$ of the nested system in terms of the weights of the initial $Y(\mfgl_n)$-system,
\gan{
\lambda_1^{(k)}(v;\setu^{(1)}, \dots, \setu^{(k)}) =\lambda_{k+1}(v) \prod_{i=1}^{m^{(k)}} \frac{v-u_j^{(k)}-1}{v-u_j^{(k)}} \qu \text{for} \qu 1 \leq k \leq n-1 \qu\text{and}
\\
\lambda_l^{(k)}(v;\setu^{(1)}, \dots, \setu^{(k)}) = \lambda_{k+l}(v) \qu\text{for} \qu  l>1,  \qu  1 \leq k \leq n-l.
}
From the recursion relation in \eqref{gln:Gamma}, a general expression can be found for $\Gamma^{(k)}(v; \setu^{(1)}, \dots, \setu^{(n-1)})$, for $1 \leq k \leq n-2$:
\aln{
& \hspace{-.5cm} \Gamma^{(k)}(v; \setu^{(1)}, \dots, \setu^{(n-1)}) \el & \qu = \lambda^{(n-2)}_2(v; \setu^{(1)}, \dots, \setu^{(n-2)}) 
\prod_{i=1}^{m^{(n-1)}} \frac{v-u_i^{(n-1)}\,{-}\,1}{v-u_i^{(n-1)}} + \sum_{l=k}^{n-2} \lambda^{(l)}_1(v;\setu^{(1)}, \dots, \setu^{(l)}) \prod_{i=1}^{m^{(l+1)}} \frac{v-u^{(l+1)}\,{+}\,1}{v-u^{(l+1)}} \!\! \el
& \qu = \lambda_n(v) \prod_{i=1}^{m^{(n-1)}} \frac{v-u_i^{(n-1)}-1}{v-u_i^{(n-1)}} + \sum_{l=k}^{n-2} \lambda_{l+1}(v) \prod_{i=1}^{m^{(l)}} \frac{v-u_i^{(l)}-1}{v-u_i^{(l)}} \prod_{i=1}^{m^{(l+1)}} \frac{v-u_i^{(l+1)}+1}{v-u_i^{(l+1)}}. 
}
We have thus shown the following.

\begin{thrm}
The eigenvalues of the Bethe vectors for a $Y(\mfgl_n)$-system are given by
\ali{ \label{gln:full_eig} 
\Gamma(v) &= \lambda_1(v)\, \prod_{i=1}^{m^{(1)}} \frac{v-u^{(1)}_i+1}{v-u^{(1)}_i} + \lambda_n(v) \prod_{i=1}^{m^{(n-1)}} \frac{v-u_i^{(n-1)}-1}{v-u_i^{(n-1)}} \el & \qu + \sum_{l=1}^{n-2} \lambda_{l+1}(v) \prod_{i=1}^{m^{(l)}} \frac{v-u_i^{(l)}-1}{v-u_i^{(l)}} \cdot \prod_{i=1}^{m^{(l+1)}} \frac{v-u_i^{(l+1)}+1}{v-u_i^{(l+1)}} .
}
\end{thrm}

Recall also the Bethe equations \eqref{gln:betheq} satisfied by parameters $u_j^{(k)}$. In fact, comparing the above two expressions, we note that equivalent Bethe equations can be obtained by demanding instead that the residue of the full eigenvalue $\Gamma(v)$ vanishes at each $u_j^{(k)}$ for $1\leq k \leq n-1$, $1 \leq j \leq m^{(k)}$. This is exactly the condition that the eigenvalue of the transfer matrix is analytic. We may now evaluate the residue to obtain the Bethe equations in terms of $\lambda_k(v)$ with $1 \leq k \leq n$ leading to the following statement.

\begin{thrm}
The Bethe equations for a $Y(\mfgl_n)$-system are
\eqa{ \label{gln:BE}
\frac{\lambda_{k}(u_j^{(k)}) }{\lambda_{k+1}(u_j^{(k)}) }
&=
\prod_{i=1}^{m^{(k-1)}} \frac{u_j^{(k)}-u_i^{(k-1)}}{u_j^{(k)}-u_i^{(k-1)}-1} \cdot \prod_{i \neq j} \frac{u_j^{(k)}-u_i^{(k)}-1}{u_j^{(k)}-u_i^{(k)}+1}
\cdot \prod_{i=1}^{m^{(k+1)}} \frac{u_j^{(k)}-u_i^{(k+1)}+1}{u_j^{(k)}-u_i^{(k+1)}} ,
\\
\frac{\lambda_{1}(u_j^{(1)}) }{\lambda_{2}(u_j^{(1)}) }
&=
\prod_{i \neq j} \frac{u_j^{(1)}-u_i^{(1)}-1}{u_j^{(1)}-u_i^{(1)}+1}
\cdot \prod_{i=1}^{m^{(2)}} \frac{u_j^{(1)}-u_i^{(2)}+1}{u_j^{(1)}-u_i^{(2)}} ,
\\
\frac{\lambda_{n-1}(u_j^{(n-1)}) }{\lambda_{n}(u_j^{(n-1)}) }
&=
\prod_{i=1}^{m^{(n-2)}} \frac{u_j^{(n-1)}-u_i^{(n-2)}}{u_j^{(n-1)}-u_i^{(n-2)}-1} 
\cdot \prod_{i \neq j} \frac{u_j^{(n-1)}-u_i^{(n-1)}-1}{u_j^{(n-1)}-u_i^{(n-1)}+1} ,
}
for $1\le k \le n-1$ and $1\le j \le m^{(k)}$.
\end{thrm}


\section*{Acknowledgements}  \enlargethispage{1.5em}

The authors thank Samuel Belliard, Nicolas Cramp\'e, Nicolas Guay, Bart Vlaar and Curtis Wendlandt for useful discussions and the anonymous referee for comments and suggestions. V.R.~was in part supported by the UK EPSRC under the grant EP/K031805/1 and by the European Social Fund, grant number 09.3.3-LMT-K-712-02-0017.


\bigskip

\end{document}